\pdfoutput=1
\documentclass[journal,twoside]{IEEEtran}

\usepackage{packages}
\usepackage{theoremstyles}

\usepackage{tikz}
\usepackage{pgfplots}
\usepackage{pgfplotstable}
\pgfplotsset{compat=newest}

\usepackage{balance}
\usepackage{setspace}
\usepackage[disable]{todonotes}

\usepackage{import}

\usepackage[misc,geometry]{ifsym}

\usepackage[caption=false]{subfig}

\usepackage[colorlinks=false,pdfborder={0 0 0}]{hyperref}
\usepackage[capitalize]{cleveref}

\newcommand{\eg}{e.g.\xspace}
\newcommand{\ie}{i.e., }

\newcommand{\wrt}{with respect to\xspace}

\newcommand{\colregs}{COLREGS\xspace}
\newcommand{\nasheq}{NE\xspace}
\newcommand{\nasheqs}{NEs\xspace}
\newcommand{\ibr}{IBR\xspace}

\newcommand{\vect}[1]{%
    \mathbf{#1}%
}

\newcommand{\set}[1]{%
    \mathcal{#1}%
}
\newcommand{\powerset}[1]{%
    2^{#1}%
}

\DeclarePairedDelimiter\abs{\lvert}{\rvert}
\DeclarePairedDelimiter\norm{\lVert}{\rVert}

\DeclareMathOperator*{\argmax}{arg\,max}

\newcommand*{\bigO}{\mathcal{O}}

\algnewcommand{\LineComment}[1]{\State $\triangleright$ #1}

\DeclareSIUnit\mps{\meter \per \second}
\DeclareSIUnit\mpss{\meter \per \second \squared}

\DeclareMathOperator{\globally}{\mathbf{G}}

\NewDocumentCommand{\robustness}{mmo}{%
    \rho_{#1}(#2\IfNoValueF{#3}{, #3})%
}
\newcommand{\negRobustnessSymb}[1]{%
    \rho^{-}_{#1}%
}
\NewDocumentCommand{\negRobustness}{mmo}{%
    \negRobustnessSymb{#1}(#2\IfNoValueF{#3}{, #3})%
}

\NewDocumentCommand{\stateDyn}{o}{%
    \vect{x}\IfNoValueF{#1}{_{#1}}%
}

\NewDocumentCommand{\stateSpace}{o}{%
    \set{X}\IfNoValueF{#1}{_{#1}}%
}

\NewDocumentCommand{\inputDyn}{o}{%
    \vect{u}\IfNoValueF{#1}{_{#1}}%
}

\NewDocumentCommand{\inputSpace}{o}{%
    \set{U}\IfNoValueF{#1}{_{#1}}%
}

\NewDocumentCommand{\trajectory}{o}{%
    \xi\IfNoValueF{#1}{_{#1}}%
}

\NewDocumentCommand{\feasibleTrajs}{o}{%
    \Xi\IfNoValueF{#1}{_{#1}}%
}

\newcommand{\scenario}{%
    S%
}

\newcommand{\egoIndex}{%
    i%
}
\newcommand{\obstacleIndex}{%
    j%
}

\newcommand{\poset}[1]{%
    \mathsf{#1}%
}
\newcommand{\posetPrec}{%
    \prec%
}
\newcommand{\posetPreceq}{%
    \preceq%
}
\newcommand{\posetSucc}{%
    \succ%
}
\newcommand{\posetSucceq}{%
    \succeq%
}

\newcommand{\preorderPrec}{%
    \prec%
}
\newcommand{\preorderPreceq}{%
    \precsim%
}
\newcommand{\preorderSucc}{%
    \succ%
}
\newcommand{\preorderSucceq}{%
    \succsim%
}

\newcommand{\refinedby}{\sqsubseteq}
\newcommand{\refines}{\sqsupseteq}

\newcommand{\game}{G}
\newcommand{\players}{\mathcal{A}}
\newcommand{\otherPlayers}{\mathcal{A}_{- \egoIndex}}
\NewDocumentCommand{\actionSpace}{o}{%
    \Gamma\IfNoValueF{#1}{_{#1}}%
}
\NewDocumentCommand{\actionProfile}{o}{%
    \gamma\IfNoValueF{#1}{_{#1}}%
}
\NewDocumentCommand{\metrics}{o}{%
    \set{M}\IfNoValueF{#1}{_{#1}}%
}
\NewDocumentCommand{\preference}{o}{%
    \poset{P}\IfNoValueF{#1}{_{#1}}%
}
\NewDocumentCommand{\actionPrec}{o}{%
    \preorderPrec^{\actionSpace}\IfNoValueF{#1}{_{#1}}%
}
\NewDocumentCommand{\actionPreceq}{o}{%
    \preorderPreceq^{\actionSpace}\IfNoValueF{#1}{_{#1}}%
}
\NewDocumentCommand{\actionSucc}{o}{%
    \preorderSucc^{\actionSpace}\IfNoValueF{#1}{_{#1}}%
}
\NewDocumentCommand{\actionSucceq}{o}{%
    \preorderSucceq^{\actionSpace}\IfNoValueF{#1}{_{#1}}%
}
\newcommand{\bestResponse}[1]{%
    \mathsf{BR}_{#1}%
}
\newcommand{\nes}[1]{%
    \mathsf{NE}(#1)%
}

\newcommand{\metricTags}{%
    \mathcal{L}%
}
\newcommand{\concern}{%
    \mathcal{C}%
}
\newcommand{\concerns}{%
    \mathfrak{C}%
}
\newcommand{\concernOfMetric}[1]{%
    \concern_{#1}%
}

\newcommand{\guideline}{%
    \poset{G}%
}

\newcommand{\inducedPreference}[2][\metrics]{%
    \preference(#1, #2)%
}

\newcommand{\inject}[1]{%
    \ltimes_{#1}%
}

\newcommand{\metricTag}[1]{%
    \mathtt{#1}%
}
\newcommand{\collTag}{%
    \metricTag{collision}%
}
\newcommand{\giveWayTag}{%
    \metricTag{give\text{-}way}%
}
\newcommand{\standOnTag}{%
    \metricTag{stand\text{-}on}%
}
\newcommand{\nucTag}{%
    \metricTag{nuc}%
}
\newcommand{\hamperedTag}{%
    \metricTag{hampered}%
}
\newcommand{\fishingTag}{%
    \metricTag{fishing}%
}
\newcommand{\sailingTag}{%
    \metricTag{sailing}%
}
\newcommand{\powerDrivenTag}{%
    \metricTag{power\text{-}driven}%
}
\newcommand{\highVulnerabilityTag}{%
    \metricTag{high\text{-}vul}%
}
\newcommand{\lowVulnerabilityTag}{%
    \metricTag{low\text{-}vul}%
}

\newcommand{\collMetric}[2]{%
    m^\mathtt{coll}_{#1, #2}%
}

\newcommand{\giveWayMetric}[2]{%
    m^\mathtt{gw}_{#1, #2}%
}
\newcommand{\standOnMetric}[2]{%
    m^\mathtt{so}_{#1, #2}%
}

\newcommand{\legalMetrics}[1]{%
    \metrics[#1]^\mathtt{legal}
}

\newcommand{\colregsMetrics}[1]{%
    \metrics[#1]^\mathtt{\colregs}
}

\newcommand{\legalGuideline}{%
    \guideline^\mathtt{legal}%
}

\newcommand{\colregsGuideline}{%
    \guideline^\mathtt{\colregs}%
}

\newcommand{\legalGame}{%
    \game_\mathtt{legal}%
}
\newcommand{\trueGame}{%
    \game_\mathtt{true}%
}
\newcommand{\localGame}[1]{%
    \game^{#1}%
}

\newcommand{\socialMetric}{%
    s%
}

\definecolor{RPTHRed}{HTML}{E31B23}

\colorlet{trueVal}{TUMGreen!70}
\colorlet{falseVal}{TUMRed!70}
\colorlet{unk1}{TUMPurple!70}
\colorlet{unk2}{TUMBlue!70}

\colorlet{backdrop}{TUMLightGray}

\colorlet{innerego}{TUMDarkYellow}
\colorlet{outerego}{TUMGreen}

\definecolor{egofill}{HTML}{e37222}
\definecolor{egostroke}{HTML}{9c4100}

\definecolor{obstaclefill}{HTML}{1d7eea}
\definecolor{obstaclestroke}{HTML}{00478f}

\definecolor{staticobstaclefill}{HTML}{d95558}
\definecolor{staticobstaclestroke}{HTML}{831d20}

\definecolor{goalfill}{HTML}{f1b514}
\definecolor{goalstroke}{HTML}{302404}

\definecolor{reachfill}{HTML}{7bc4c3}
\definecolor{reachstroke}{HTML}{5c9392}

\definecolor{refpath}{HTML}{ff477e}

\definecolor{initial}{HTML}{008000}

\definecolor{sampled}{HTML}{0c0cf5}

\newcommand{\comment}[3]{}

\renewcommand{\Cref}[1]{\cref{#1}}
\crefname{section}{Sec.}{Sec.}
\crefname{table}{Tab.}{Tab.}
\crefname{equation}{Eq.}{Eq.}
\crefname{algorithm}{Alg.}{Alg.}
\crefname{definition}{Def.}{Def.}
\crefname{lemma}{Lemma}{Lemma}
\crefname{problem}{Problem}{Problem}
\crefname{theorem}{Theorem}{Theorem}
\crefname{corollary}{Corollary}{Corollary}
\crefname{proposition}{Proposition}{Proposition}

\usetikzlibrary{arrows,backgrounds,calc,cd,fit,plotmarks,positioning,shapes.callouts,shapes.geometric}
\usepgfplotslibrary{groupplots}

\usetikzlibrary{external}
\AtBeginDocument{\pgfplotsset{cycle list name=tumcolormark}}

\pgfplotsset{
    colormap={trajTUMBlue}{
        color=(TUMBlue!20)
        color=(TUMBlue)
    },
    colormap={trajTUMOrange}{
        color=(TUMOrange!20)
        color=(TUMOrange)
    },
    colormap={trajTUMGreen}{
        color=(TUMGreen!20)
        color=(TUMGreen)
    },
    colormap={trajTUMPurple}{
        color=(TUMPurple!20)
        color=(TUMPurple)
    },
    colormap={trajTUMDarkRed}{
        color=(TUMDarkRed!20)
        color=(TUMDarkRed)
    },
}

\pgfplotscreateplotcyclelist{trajectory}{%
    TUMBlue,colormap name=trajTUMBlue,every mark/.append style={fill=TUMBlue}\\%
    TUMOrange,colormap name=trajTUMOrange,every mark/.append style={fill=TUMOrange}\\%
    TUMGreen,colormap name=trajTUMGreen,every mark/.append style={fill=TUMGreen}\\%
    TUMPurple,colormap name=trajTUMPurple,every mark/.append style={fill=TUMPurple}\\%
}

\pgfplotsset{
    trajectory/.style={
        ultra thick,
        line cap=round,
    },
    gradient/.style={
        mesh,
        shader=interp,
        point meta={\thisrow{time}},
    },
    pick traj/.style={
        x filter/.expression={\thisrow{traj} == #1 ? x : nan}
    },
    initial state/.style={
        only marks,
        mark=*,
    },
    goal/.style={
        only marks,
        mark=asterisk,
        mark size=3pt,
        every mark/.append style={thick},
    },
    defector/.style={
        TUMDarkRed,
        colormap name=trajTUMDarkRed,
        every mark/.append style={fill=TUMDarkRed},
    },
    original/.style={line width=2pt, TUMMediumGray},
    refined/.style={line width=0.65pt, TUMBlue},
}

\NewDocumentCommand{\plotTrajectories}{sO{1}mO{}}{%
    \pgfmathtruncatemacro{\lastid}{#2 - 1}
    \foreach \id in {0, ..., \lastid} {
        \addplot+[forget plot, trajectory, pick traj=\id, #4] table[x=x, y=y] {#3};
    }
    \IfBooleanT{#1}{
        \pgfplotsset{cycle list shift=1}
    }
}

\pgfplotsset{
    signal/.style={
        color=TUMDarkerBlue,
        no marks,
    },
    signalIncluded/.style={
        color=TUMDarkerBlue,
        mark=*,
        only marks,
        every mark/.append style={fill=TUMDarkerBlue},
    },
    signalExcluded/.style={
        signalIncluded,
        every mark/.append style={fill=TUMWhite},
    },
}

\usetikzlibrary{shapes.misc}
\tikzset{cross/.style={cross out, draw=black, minimum size=2*(#1-\pgflinewidth), inner sep=0pt, outer sep=0pt},cross/.default={1pt}}

\tikzset{inner xseps/.style 2 args={inner xsep=#1 + #2,xshift=#2 - #1}}
\tikzset{inner yseps/.style 2 args={inner ysep=#1 + #2,yshift=#2 - #1}}

\newcommand{\vessel}[5]{
    \begin{scope}[shift={#1}, rotate=#2]
        \draw [fill=#3](1, -1) -- (2, 0) -- (1, 1) -| (-1, -1) -- cycle;
        \node at (0.5, 0) {#4};
        \coordinate (front_#5) at (2, 0) {};
    \end{scope}
}

\newcommand{\robot}[2][TUMBlack]{%
    \begin{scope}[scale=0.2]
        \draw[fill=#1] (-1.2,-2) rectangle (1.2,2);

        \draw[fill=#1] (-1.5,-1.5) rectangle (-1.2,1.5);
        \draw[fill=#1] ( 1.2,-1.5) rectangle ( 1.5,1.5);

        \draw[fill=#1] (-0.5,2) arc (180:0:0.5 and 0.3) -- cycle;

        \coordinate (#2_front) at (0, 2.3) {};
        \coordinate (#2_center) at (0, 0) {};
    \end{scope}
}

\begin{document}
\title{Auditable and Deployable Multi-Agent Motion Planning via Posetal Games from Legal Guidelines}
\author{%
    Florian Lercher$^{1,2}$\textsuperscript{*},
    Karan Mahesh$^{3,4}$\textsuperscript{*},
    Yujun Huang$^{1}$,
    Meshal Alharbi$^{1}$,
    Runyu Zhang$^{1}$,
    Hanna Krasowski$^{5}$,
    \\
    Michael R. Benjamin$^{3}$,
    Matthias Althoff$^{2}$,
    Gioele Zardini$^{1}$
    \thanks{%
        \textsuperscript{*}Florian Lercher and Karan Mahesh contributed equally to this work and are listed in alphabetical order.
    }%
    \thanks{%
        This work was partially supported by the MIT Maritime Consortium, the German Research Foundation (DFG) under grant AL 1185/20-1, and the National Science Foundation Graduate Research Fellowship Program under Grant No. 2141064. Any opinions, findings, and conclusions or recommendations expressed in this material are those of the author(s) and do not necessarily reflect the views of the National Science Foundation.
    }%
    \thanks{$^1$Laboratory for Information and Decision Systems, Massachusetts Institute of Technology,
        Cambridge, MA 02139, USA. \{\href{mailto:yujun233@mit.edu}{\texttt{\small yujun233}}, \href{mailto:meshal@mit.edu}{\texttt{\small meshal}},
        \href{mailto:runyuzha@mit.edu}{\texttt{\small runyuzha}},
        \href{mailto:gzardini@mit.edu}{\texttt{\small gzardini}}\}\texttt{\small @mit.edu}}
    \thanks{%
        $^2$School of Computation, Information and Technology,
        Technical University of Munich,
        85748 Garching,
        Germany.
        \{\href{mailto:florian.lercher@tum.de}{\texttt{\small florian.lercher}}, \href{mailto:florian.lercher@tum.de}{\texttt{\small althoff}}\}\texttt{\small @tum.de}%
    }%
    \thanks{%
        $^3$Department of Mechanical Engineering,
        Massachusetts Institute of Technology,
        Cambridge, MA 02139, USA.
        \{\href{mailto:karanmah@mit.edu}{\texttt{\small karanmah}}, \href{mailto:mikerb@mit.edu}{\texttt{\small mikerb}}\}\texttt{\small @mit.edu}%
    }%
    \thanks{%
        $^4$Woods Hole Oceanographic Institution, Woods Hole,
        MA 02543, USA.
    }%
    \thanks{%
        $^5$Department of Electrical Engineering and Computer Sciences,
        University of California, Berkeley,
        Berkeley, CA, USA.
        \href{mailto:krasowski@mit.edu}{\texttt{\small krasowski@berkeley.edu}}%
    }%
    \thanks{%
        This research was conducted during Florian Lercher's research stay at Massachusetts Institute of Technology in Prof. Zardini's group.
    }
}

\maketitle

\begin{abstract}
    When robots are deployed in the real-world, they must interact with other agents and accomplish their own goals while adhering to legal constraints.
These potentially conflicting objectives are often only partially prioritized, making scalarized objectives or fully specified hierarchies unsuitable.
We present a framework for multi-agent motion planning under legal requirements based on games with partially ordered preferences (posetal games).
Scenario-independent legal guidelines order abstract legal concerns, from which we automatically derive agent preferences and a legal game for any concrete scenario.
Agents may refine their legal preferences with private objectives, and regulators may refine the guideline as laws evolve.
We prove that such customizations can only remove equilibria and thus preserve legality.
To coordinate agents without explicitly communicating private preferences, we introduce belief play, which provably recovers the true equilibria from conservative local games, and a distributed iterative best-response protocol requiring less information.
We apply the framework to maritime navigation, systematically lifting pairwise traffic rules to multi-vessel encounters.
On seven encounters with up to four vessels, every rule violation of the committed equilibria is certified as forced by the legal hierarchy or unsatisfiable, and defectors are identified from public information in all tested cases.
These results indicate that guideline-based posetal games offer an auditable route from partially specified laws to multi-agent motion planning.

\end{abstract}


\tikzexternaldisable
\section{Introduction}
\label{sec:introduction}

Robots that share roads, waterways, or airspace with other agents must plan their motion subject to many potentially conflicting objectives.
These objectives often include legal requirements imposed by regulators, which take precedence over individual objectives.
Since not all objectives can be satisfied simultaneously in general, \emph{minimum-violation} planning prioritizes them and systematically violates lower-priority objectives when necessary to satisfy higher-priority ones~\cite{tumovaLeastviolatingControlStrategy2013,9483174}.
Legally compliant multi-agent motion planning is nevertheless difficult for at least three reasons.
First, some objectives are incomparable, in the sense that one cannot assert that either objective is more important than the other.
This is particularly relevant for legal regulations, where authorities typically specify only a partial priority order~\cite{censiLiabilityEthicsCultureAware2019}.
Second, many traffic laws are stated as rules between pairs of agents, so the obligations of an agent depend on which other agents are present, and obligations toward different agents might conflict (\cref{fig:conflict_encounter}).
Third, in multi-agent settings, agents respond strategically to the actions of other agents~\cite{zanardiPosetalGamesEfficiency2022}, whose private objectives are typically unknown.
Taken together, these obstacles call for a framework that derives the preferences of the agents from the law for every scenario, preserves the partial priorities of the law, and coordinates the agents without requiring them to disclose their private objectives.

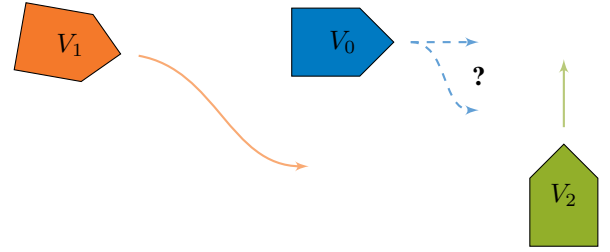
\begin{figure}[t]
      \centering
\begin{tikzpicture}[scale=0.45]
    \vessel{(0,2)}{0}{TUMBlue}{$V_0$}{ego}
    \vessel{(-8,2)}{-10}{TUMOrange}{$V_1$}{over}
    \vessel{(7,-3)}{90}{TUMGreen}{$V_2$}{cross}

    \draw [-latex', thick, TUMGreen!60] (7,-0.5) -- (7,1.5);
    \draw (-5.5,1.6) edge[out=-10,in=180,-latex', thick, TUMOrange!60] (-0.5,-1.65);
    
    \draw [-latex', thick, dashed, TUMBlue!60] (2.5,2) -- (4.5,2);
    \draw (2.5,2) edge[out=0,in=180,-latex', thick, dashed, TUMBlue!60] (4.5,0);
    \node at (4.5, 1) {\textbf{?}};
\end{tikzpicture}
      \caption{%
            A vessel encounter with conflicting \colregs rules for $V_0$.
            The vessel is required to keep its current course and speed for the overtaking vessel $V_1$ and, simultaneously, give way to the crossing vessel $V_2$ by evading to starboard~\cite[Rules~13 and 15]{internationalmaritimeorganizationimoCOLREGsConventionInternational1972}.
      }
      \label{fig:conflict_encounter}
\end{figure}

A particularly important instance of this problem is motion planning for autonomous vessels in the maritime domain.
The International Code of Safety for Maritime Autonomous Ships~\cite{internationalmaritimeorganizationimoInternationalCodeSafety2026}, adopted in 2026, requires autonomous vessels to comply with the maritime traffic rules laid out in the Convention on the International Regulations for Preventing Collisions at Sea (\colregs)~\cite{internationalmaritimeorganizationimoCOLREGsConventionInternational1972}.
However, the \colregs are only defined for pairs of vessels which may lead to conflicting requirements in encounters involving more than two vessels, as shown in \cref{fig:conflict_encounter}.
While the \colregs prioritize collision avoidance over strict compliance with the encounter rules~\cite[Rule~2]{internationalmaritimeorganizationimoCOLREGsConventionInternational1972}, they do not specify which rules should be violated when conflicts arise.

Considering these challenges, we identify \emph{auditability}, \emph{customizability}, and \emph{deployability} as key requirements for a multi-agent motion planning algorithm under legal constraints.
\emph{Auditability} requires agents to exhibit predictable behavior and provide explanations for their actions, particularly when not all conflicting objectives can be satisfied.
\emph{Customizability} enables regulators to revise legal priorities and manufacturers to refine priorities that laws leave unspecified.
Moreover, any \emph{deployable} algorithm has to allow agents to coordinate despite their limited knowledge of the private objectives of other agents.
To meet these requirements, we present, to the best of our knowledge, the first systematic method for automatically constructing legally compliant motion-planning games from scenario-independent legal guidelines.
Our framework represents objective hierarchies as partially ordered preferences, resulting in \emph{posetal games}~\cite{zanardiPosetalGamesEfficiency2022}, for which we provide solution protocols that operate under limited information.
The key idea is to separate what the law prescribes from what a scenario instantiates.
Regulators specify a \emph{guideline}, i.e., a partial order over abstract legal concerns such as collision avoidance, and every concrete scenario inherits the preferences of its agents by matching their legal metrics to these concerns.
Since every customization of the resulting preferences is a refinement, and refining preferences can only remove equilibria, legality is preserved by construction.
We demonstrate the practical applicability of the framework through a case study of \colregs-compliant motion planning for multi-vessel encounters.

\subsection{Related Work}
\label{subsec:related_work}

Our work lies at the intersection of formal specifications for robotics, game-theoretic motion planning with ordered objectives, and COLREGS-compliant maritime navigation.
Since formal specifications provide an interpretable representation of complex objectives and serve as the building blocks of posetal games, we first review their use in robotics, followed by specification-based approaches to multi-agent motion planning and related game-theoretic formulations with ordered preferences.
Finally, we discuss existing approaches for \colregs-compliant multi-vessel motion planning to provide context for our case study.

\subsubsection{Specifications for Robotic Interactions}

Formal specifications, initially used primarily for offline monitoring and runtime verification, have become core components of learning, synthesis, and control methods~\cite{vasileReactiveSamplingbasedPath2020,Meng2023,SrinivasanCoogan2021-TL2CBF-TRO,pek2023spatial,deshmukh2017robust,Li2017,Jiang_Bharadwaj_Wu_Shah_Topcu_Stone_2021,Hasanbeig2019}.
They provide interpretable representations of objectives and constraints and, in multi-agent settings, requirements governing interactions among agents.
For example, autonomous-driving systems must satisfy specifications ranging from lane-keeping requirements to appropriate responses to emergency vehicles~\cite{althoffNoMoreTraffic2026,mehdipourFormalMethodsComply2023}.
In unstructured environments, interaction requirements are commonly expressed through pairwise rules or fixed priority structures, which favor concise representations at the expense of completeness.
For instance, the airborne collision-avoidance system ACAS-X determines priorities for pairwise aircraft encounters and relies on heuristics to combine the resulting evaluations~\cite{kochenderfer2012next-ACASX,Jeannin2015-ACASX-verify}.
Even in this pairwise setting, the precomputed lookup tables are sufficiently large to motivate dedicated compression methods~\cite{julian2019deepACASX}.
This scalability challenge is further amplified in domains such as warehouse navigation and maritime robotics, where agents and interactions are more diverse~\cite{Hang2019-MAPF,Li2020-warehouse,krasowskiProvableTrafficRule2024,heistand2024approach}.
Rulebooks~\cite{censiLiabilityEthicsCultureAware2019,wongpiromsarnFormalSpecificationControl2026} reduce the need for exhaustive specification by prioritizing objectives and permitting lower-priority objectives to be violated when necessary to satisfy higher-priority ones.
Our framework builds on this principle and extends it to multiple strategic agents whose obligations depend on the scenario.

\subsubsection{Motion Planning with Formal Objectives}

Motion-planning methods based on formal objectives generally adopt one of two perspectives: single-agent planning, in which other agents are treated as uncooperative, or multi-agent planning, in which varying degrees of cooperation are assumed.
In single-agent settings, the behavior of surrounding agents is typically modeled using worst-case assumptions, reflecting minimal or no cooperation~\cite{koschiSetBasedPredictionTraffic2021,krasowskiProvableTrafficRule2024}.
Moreover, formal safety specifications are assigned higher priority than performance objectives to guarantee specification compliance by the ego agent~\cite{manzingerUsingReachableSets2021,Halder2026-lexicographicplanning,krasowskiProvableTrafficRule2024}.

In multi-agent settings, game-theoretic motion planning naturally models strategic interactions between agents.
Most existing approaches require the objectives of each agent to be expressed as a scalar cost function, \eg, \cite{drevesGeneralizedNashEquilibrium2018,fridovich-keilIterativeQuadraticMethod2020,wangGameTheoreticPlanningSelfDriving2021,bhattStrategicDecisionMakingMultiagent2025}.
However, scalarization can obscure the relative priorities among objectives so that the resulting behavior becomes difficult to interpret and predict.

\subsubsection{Games with Ordered Objectives}
\label{subsubsec:games_with_ordered_objectives}

To increase interpretability when rich preference structures are present, game-theoretic formulations that allow for enforcing orders of preferences have been proposed. The simplest structure is a two-level hierarchy, common when constraints, such as safety requirements, have strict priority over the performance objective.
For example, in urban driving games~\cite{zanardiUrbanDrivingGames2021}, a shared collision-avoidance objective takes strict priority over individual costs.
These personal costs may only depend on the personal action of the agent, a limitation that is lifted in \cite{millerMultiagentMotionPlanning2022} by incorporating zero-sum components.
Games of ordered preferences~\cite{leeYouCantAlways2025} further generalize the objective hierarchy by allowing an arbitrary number of levels, corresponding to a total order over objectives.
An iterative best-response solver for these games is presented in~\cite{delasherasmolinsApproximateSolutionsGames2025}.
However, requiring a total order over legal requirements places a significant burden on regulators.
Posetal games~\cite{zanardiPosetalGamesEfficiency2022} relax this assumption by enabling objectives to remain incomparable, \ie they admit a partial order over the objectives.
Thus, they support the most general preference structure among the methods considered here.

Although the methods in~\cite{zanardiUrbanDrivingGames2021,millerMultiagentMotionPlanning2022,zanardiPosetalGamesEfficiency2022,leeYouCantAlways2025,delasherasmolinsApproximateSolutionsGames2025} compute Nash equilibria, their practical deployment remains challenging.
In particular, they assume complete knowledge of the preferences of all involved agents.
Moreover, they do not provide a mechanism for translating general legal requirements into the preferences required to model a specific scenario.

\subsubsection{\colregs-Compliant Multi-Vessel Motion Planning}
\label{subsubsec:colregs_compliant_multi_vessel_motion_planning}

The \colregs have been translated into formal specifications to provide a precise, machine-interpretable representation~\cite{krasowskiTemporalLogicFormalization2021,torbenAutomaticSimulationbasedTesting2023,pedersenFormalizingTestingCollision2025}.
These formalizations support the monitoring, testing, and synthesis of behavior in pairwise vessel encounters.
However, because the \colregs primarily specify pairwise obligations, they do not explicitly prescribe how to prioritize conflicting obligations arising from simultaneous encounters with multiple vessels, as illustrated in \cref{fig:conflict_encounter}.

Specification-based approaches to maritime navigation have predominantly adopted an ego-centric perspective in which the behavior of surrounding vessels is monitored or modeled rather than strategically coordinated.
Examples include real-world monitoring of \colregs encounter situations for collision avoidance~\cite{dietrich2026stayingspecrealtimemonitoring} and reactive vessel simulations in which the first detected conflict determines the executed maneuver~\cite{krasowskiIntelligentSailingModel2026}.
Such methods can enforce or evaluate the obligations of an individual vessel, but do not explicitly model the potentially heterogeneous preferences of multiple interacting vessels.

Existing multi-agent approaches commonly employ centralized game-theoretic methods or homogeneous reinforcement-learning controllers with scalarized objectives~\cite{Lin2025-MARL-Maritime,Cui2026-MARL-And-Nonreactive,choEfficientCOLREGCompliantCollision2022,weiCOLREGsCompliantMultiShipCollision2022,grgicevicCentralisedDecisionSupport2025}.
However, scalarization obscures the priority structure among objectives, and centralized or homogeneous formulations generally cannot accommodate private, agent-specific preferences.
To the best of our knowledge, our framework is the first to systematically lift the pairwise \colregs rules to encounters involving an arbitrary number of vessels while preserving an interpretable hierarchy of legal objectives.
It further supports agents with private preferences and solution protocols that operate without requiring those preferences to be disclosed.

\subsection{Contributions}
\label{subsec:contributions}

We present a framework for multi-agent motion planning under legal requirements with partially ordered objectives, which automates the translation of practical scenarios into the theory of posetal games~\cite{zanardiPosetalGamesEfficiency2022}.
Our contributions are fourfold.
\begin{itemize}
      \item \textbf{Systematic construction of legal posetal games.}
            We introduce \emph{guidelines} to represent priorities among legal concerns independently of the concrete scenario.
            
            From the legal guideline, we automatically derive the partially ordered preferences of the agents and the corresponding multi-agent game for any concrete scenario.
            This removes the need to specify preferences by hand for every scenario, which is a prerequisite for using posetal games in practice.
      \item \textbf{Customization with legality guarantees.}
            We enable agents to incorporate private objectives and regulators to refine legal guidelines, and we prove that such refinements can only restrict the set of legal behaviors rather than introduce new ones.
            This sharply contrasts with scalarization methods (\eg, \cite{drevesGeneralizedNashEquilibrium2018,fridovich-keilIterativeQuadraticMethod2020,wangGameTheoreticPlanningSelfDriving2021,bhattStrategicDecisionMakingMultiagent2025}), in which changing even a single weight can alter the resulting behavior in ways that are to predict and audit.
      \item \textbf{Deployable solution protocols under private preferences.}
            We develop two protocols that avoid the complete-information assumption of prior work~\cite{zanardiUrbanDrivingGames2021,millerMultiagentMotionPlanning2022,zanardiPosetalGamesEfficiency2022,leeYouCantAlways2025,delasherasmolinsApproximateSolutionsGames2025}. We prove that \emph{belief play}, which relies on conservative local beliefs, recovers exactly the equilibria of the true game and remains legal when interrupted. The second protocol, distributed iterative best-response, requires even less information about other agents.
      \item \textbf{Auditable \colregs-compliant multi-vessel motion planning.}
            We instantiate our framework for autonomous vessels and lift the pairwise COLREGS rules to encounters with any number of vessels.
            We evaluate our framework empirically on seven encounters involving two to four vessels. %
            Contrary to previous approaches~\cite{Lin2025-MARL-Maritime,Cui2026-MARL-And-Nonreactive,choEfficientCOLREGCompliantCollision2022,weiCOLREGsCompliantMultiShipCollision2022,grgicevicCentralisedDecisionSupport2025}, we can audit every rule violation and certify it as either forced by conflicting legal requirements of higher or incomparable priority or unsatisfiable within the action space.
            Moreover, defectors can be identified from public information alone.
\end{itemize}

This paper is organized as follows:
In \cref{sec:background}, we provide preliminaries and our problem statement.
We present our framework for deriving posetal games from legal requirements in \cref{sec:guideline_games} before explaining our solution protocols for the resulting games in \cref{sec:solution_protocols}.
In \cref{sec:case_study}, we instantiate our framework for \colregs-compliant multi-vessel motion planning and evaluate it in \cref{sec:evaluation}.
Finally, \cref{sec:conclusion} draws a conclusion and discusses future work, and \cref{app:preference_proofs} contains the proofs of the preference refinement operations.

\section{Preliminaries and Problem Statement}
\label{sec:background}

In this section, we recall the necessary background on games with partially ordered preferences and then state the problem addressed in this paper.

\subsection{Games with Partially Ordered Preferences}
\label{subsec:posetal_games}

In classical games, every agent possesses a scalar utility function to evaluate the outcomes of the game.
Games with partially ordered preferences, or \emph{posetal games}, are more general in that they allow each agent to evaluate an outcome \wrt multiple metrics that are prioritized according to a partial order~\cite{zanardiPosetalGamesEfficiency2022}.
We first introduce partially ordered sets and preferences to formalize this evaluation and then define posetal games.

\subsubsection{Partially Ordered Sets}
\label{subsubsec:posets}

Let $\set{P}$ be a set.
A \emph{pre-order} $\preorderPreceq$ over $\set{P}$ is a reflexive and transitive relation over $\set{P}$.
If a pre-order is also antisymmetric, we call it a \emph{partial order} and write $\posetPreceq$.
For convenience, we use $p \posetPreceq p'$ interchangeably with $\langle p, p' \rangle \in {\posetPreceq}$, and analogously for pre-orders.
In both pre-orders and partial orders, we denote the \emph{strict} relation with $\posetPrec$, where $p \posetPrec p'$ if and only if $p \preorderPreceq p'$ ($p \posetPreceq p'$) but $p' \not\preorderPreceq p$ ($p' \not\posetPreceq p$).
Observe that the strict relation is asymmetric and transitive.
A \emph{partially ordered set} or \emph{poset} is a pair $\poset{P} := \langle \set{P}, \posetPreceq_{\poset{P}} \rangle$ where $\set{P}$ is the ground set and $\posetPreceq_{\poset{P}}$ is a partial order over $\set{P}$.
We can visualize a poset using its \emph{Hasse diagram} as shown in \cref{fig:hasse_diagram}.

\begin{figure}
    \centering
    \begin{tikzcd}[row sep=tiny,column sep=tiny]
    & p_1 \ar[dl,dash] \ar[dr,dash] & \\
    p_2 \ar[dr,dash] & & p_3 \ar[dl,dash] \\
    & p_4 &
\end{tikzcd}
    \caption{%
        The Hasse diagram of the poset $\poset{P} = \langle \{ p_1, p_2, p_3, p_4 \}, \posetPreceq_{\poset{P}} \rangle$ with $p_4 \posetPrec_{\poset{P}} p_2 \posetPrec_{\poset{P}} p_1$ and $p_4 \posetPrec_{\poset{P}} p_3 \posetPrec_{\poset{P}} p_1$;
        $p_2$ and $p_3$ are incomparable.
    }
    \label{fig:hasse_diagram}
\end{figure}
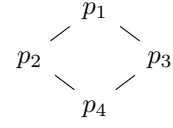

\subsubsection{Preferences}
\label{subsubsec:preferences}

Let $\actionSpace$ be a set representing an \emph{action space}.
A \emph{metric} $m : \actionSpace \rightarrow \mathbb{R}$ maps every action $\actionProfile \in \actionSpace$ to a utility $m(\actionProfile) \in \mathbb{R}$.\footnote{While the utility value could come from any pre-ordered set, we restrict ourselves to $\mathbb{R}$ to simplify the presentation.}
We allow metrics to be named so that the metrics $m^n$ and $m^{n'}$ may be distinct even if they represent the same underlying function; that is, $m^n \neq m^{n'}$ even if $m^n(\actionProfile) = m^{n'}(\actionProfile)$ for every $\actionProfile \in \actionSpace$.
Then, a \emph{preference} $\preference := \langle \metrics, \posetPreceq_{\preference} \rangle$ is a finite poset of metrics, where $m \posetPreceq_{\preference} m'$ indicates that $m'$ is at least as important as $m$.
As explained in \cite{censiLiabilityEthicsCultureAware2019,zanardiPosetalGamesEfficiency2022}, a preference induces a pre-order $\actionPreceq[\preference]$ over the action space $\actionSpace$:
We have $\actionProfile \actionPreceq[\preference] \actionProfile'$ if and only if~\cite[Def.~5]{censiLiabilityEthicsCultureAware2019}
\begin{equation}
    \label{eq:action_order}
    \begin{aligned}
         & \forall m \in \metrics : m(\actionProfile) > m(\actionProfile') \implies                                             \\
         & \quad \exists m' \in \metrics : m' \posetSucc_{\preference} m \text{ and } m'(\actionProfile) < m'(\actionProfile').
    \end{aligned}
\end{equation}
In other words, $\actionProfile'$ is at least as preferable as $\actionProfile$ if and only if each metric in which $\actionProfile'$ performs worse than $\actionProfile$ is dominated by a more important metric in which $\actionProfile'$ performs better.
\Cref{fig:example_robots_in_corridor} illustrates the induced order on actions using a simple example.

\begin{figure}[tb]
    \centering
    \begin{tikzpicture}
    \pgfdeclarelayer{op_area}
    \pgfdeclarelayer{trajectories}
    \pgfsetlayers{op_area,trajectories,main}

    \tikzset{trajectory/.style={thick,-latex'}}

    \begin{scope}[shift={(0, 0)}, rotate=-90]
        \robot[TUMBlue]{r1}
        \begin{pgfonlayer}{trajectories}
            \draw[trajectory, TUMBlue!60] (r1_front) -- +(0, 1.25);
            \draw[trajectory, TUMBlue!60] (r1_front) to[out=90, in=-90] +(1.4, 1.25);
            \draw[trajectory, TUMBlue!60] (r1_front) to[out=90, in=-90] +(0.7, 1.25);
            \draw[trajectory, TUMBlue!60] (r1_front) to[out=90, in=-90] +(-0.7, 1.25);
        \end{pgfonlayer}
        \node[font=\footnotesize] at ($(r1_front) + (-0.7, 1.45)$) {$a$};
        \node[font=\footnotesize] at ($(r1_front) + (0, 1.45)$) {$b$};
        \node[font=\footnotesize] at ($(r1_front) + (0.7, 1.45)$) {$c$};
        \node[font=\footnotesize] at ($(r1_front) + (1.4, 1.45)$) {$d$};
    \end{scope}
    \begin{scope}[shift={(4, 0)}, rotate=90]
        \robot[TUMOrange]{r2}
        \begin{pgfonlayer}{trajectories}
            \draw[trajectory, TUMOrange!60] (r2_front) -- +(0, 1.25);
        \end{pgfonlayer}
    \end{scope}

    \begin{pgfonlayer}{op_area}
        \draw[ultra thick, TUMMediumGray] (-0.5, 1.1) -- (4.5, 1.1);
        \draw[ultra thick, TUMMediumGray] (-0.5, -1.1) -- (4.5, -1.1);
    \end{pgfonlayer}

    \coordinate (preference) at (-2.5, -0.5) {};

    \node[
        cloud callout,
        thick,
        draw=TUMBlue,
        fill=TUMBlue!20,
        callout absolute pointer={($(r1_center) + (-0.6, 0.25)$)},
        aspect=2,
        inner sep=14pt,
    ] at ($(preference) + (0.15, 0)$) {};

    \begin{scope}[font=\small, shift={(preference)}]
        \node[font=\small] (coll) at (0, 0.3) {collision};
        \node[font=\small] (area) at (-0.6, -0.3) {area\vphantom{l}};
        \node[font=\small] (clear) at (0.6, -0.3) {clearance};

        \draw (area) -- (coll);
        \draw (clear) -- (coll);
    \end{scope}
\end{tikzpicture}%
    \caption{%
        Two robots meet head-on in a narrow corridor.
        By regulation, robots have to avoid collisions, keep a sufficient clearance from other robots, and remain inside the operational area marked by the bold lines.
        Collision avoidance is the highest-priority objective, and the other two are incomparable.
        We take the perspective of the left robot and assume the trajectory of the right robot is fixed.
        Trajectories $a$, $c$, and $d$ are \emph{preferred} over $b$, since $b$ causes a collision;
        $a$ and $c$ are \emph{indifferent}, as they have the same insufficient clearance but stay in the operational area;
        $d$ is \emph{incomparable} to $a$ and $c$, as it achieves sufficient clearance by leaving the operational area (adapted from \cite[Fig.~2]{zanardiPosetalGamesEfficiency2022}).
    }
    \label{fig:example_robots_in_corridor}
\end{figure}
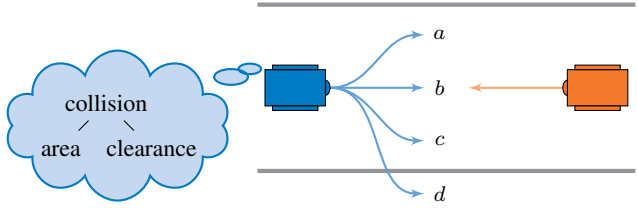

Some preferences are more specific than others in the sense that they induce a finer ordering on the action space.
Formally, we say that a preference $\preference'$ is a \emph{refinement} of $\preference$ if and only if
\begin{equation*}
    \actionProfile \actionPrec[\preference] \actionProfile' \implies \actionProfile \actionPrec[\preference'] \actionProfile'
\end{equation*}
for any pair of actions $\actionProfile, \actionProfile' \in \actionSpace$~\cite[Def.~6]{zanardiPosetalGamesEfficiency2022}.
In words, every strict preference between two actions under~$\preference$ persists under~$\preference'$.
For a fixed action space, refinement is a pre-order on preferences, which we denote $\refinedby$:
We write $\preference \refinedby \preference'$ to mean that $\preference'$ refines $\preference$, ordering preferences from coarsest to finest.

We introduce three operations that we can use to refine preferences~\cite[Sec.~II-C]{zanardiPosetalGamesEfficiency2022}.
In \cref{app:preference_proofs}, we prove that these operations indeed produce valid refinements.
First, we can add a priority relation, or equivalently \enquote{add an edge}, to the partial order over metrics.
\begin{proposition}[Preference priority refinement~{\cite[Def.~7]{zanardiPosetalGamesEfficiency2022}}]
    \label{prop:preference_priority_refinement}
    Let $\preference = \langle \metrics, \posetPreceq_{\preference} \rangle$ be a preference.
    Any preference $\preference' = \langle \metrics, \posetPreceq_{\preference'} \rangle$ with ${\posetPreceq_{\preference}} \subseteq {\posetPreceq_{\preference'}}$ is a refinement of $\preference$.
\end{proposition}
Next, we can aggregate two incomparable metrics.
\begin{proposition}[Preference aggregation~{\cite[Def.~8]{zanardiPosetalGamesEfficiency2022}}]
    \label{prop:preference_aggregation}
    Let $\preference = \langle \metrics, \posetPreceq_{\preference} \rangle$ be a preference and $m_1, m_2 \in \metrics$ be two incomparable metrics.
    Further, suppose $\hat{m}$ is a metric with $\hat{m}(\actionProfile) = f(m_1(\actionProfile), m_2(\actionProfile))$, where $f : \mathbb{R} \times \mathbb{R} \rightarrow \mathbb{R}$ is strictly increasing in both arguments.
    Let $\mathcal{N} := \metrics \setminus \{ m_1, m_2 \}$ and $\metrics' := \mathcal{N} \cup \{ \hat{m} \}$.
    Define a binary relation $R \subseteq \metrics' \times \metrics'$ by
    \begin{subequations}
        \label{eq:aggregation_order}
        \begin{alignat}{2}
            \label{subeq:aggregation_order:none}
            R :={} &          &  &
            \{ \langle m, m' \rangle \in {\posetPreceq_{\preference}} \mid m, m' \in \mathcal{N} \}                                                      \\
            \label{subeq:aggregation_order:prec}
                   & {}\cup{} &  & \{ \langle \hat{m}, m \rangle \mid m \in \mathcal{N}, \exists k \in \{ 1, 2 \} : m_k \posetPreceq_{\preference} m \}  \\
            \label{subeq:aggregation_order:succ}
                   & {}\cup{} &  & \{ \langle m, \hat{m} \rangle \mid m \in \mathcal{N}, \exists k \in \{ 1, 2 \} : m \posetPreceq_{\preference} m_k \}.
        \end{alignat}
    \end{subequations}
    Then, $\preference' = \langle \metrics', \posetPreceq_{\preference'} \rangle$ with $\posetPreceq_{\preference'} := R^*$ is a refinement of $\preference$, where $R^*$ is the reflexive-transitive hull of $R$.
    Intuitively, this means we replace $m_1$ and $m_2$ by $\hat{m}$ and relate $\hat{m}$ to all metrics that $m_1$ or $m_2$ were related to.
\end{proposition}
Finally, we can compose two disjoint preferences by appending one below the other.
This generalizes the operation from \cite[Def.~10]{zanardiPosetalGamesEfficiency2022}, which only allows adding a single metric.
\begin{proposition}[Preference augmentation]
    \label{prop:preference_augmentation}
    For $k \in \{ 1, 2 \}$, let $\preference_k = \langle \metrics_k, \posetPreceq_{\preference_k} \rangle$ be preferences so that $\metrics_1 \cap \metrics_2 = \emptyset$.
    The preference $\preference = \langle \metrics, \posetPreceq_{\preference} \rangle$ with $\metrics := \metrics_1 \cup \metrics_2$ and
    \begin{equation*}
        {\posetPreceq_{\preference}} := {\posetPreceq_{\preference_1}} \cup {\posetPreceq_{\preference_2}} \cup \{ \langle m_2, m_1 \rangle \mid m_1 \in \metrics_1, m_2 \in \metrics_2 \}
    \end{equation*}
    is a refinement of $\preference_1$.
\end{proposition}

In words, the metrics of~$\preference_2$ only break ties between actions that are equal in every metric of~$\preference_1$, so that no gain in $\preference_2$ can compensate for a loss in~$\preference_1$.

\subsubsection{Posetal Games}
\label{subsubsec:posetal_games}

A \emph{posetal game}
\begin{equation*}
    \game := \langle \players, \{ \actionSpace[\egoIndex] \}_{\egoIndex \in \players}, \{ \preference[\egoIndex] \}_{\egoIndex \in \players} \rangle
\end{equation*}
consists of a set of \emph{agents} $\players$ and, for each agent $\egoIndex \in \players$, an \emph{action space} $\actionSpace[\egoIndex]$ and a \emph{preference} $\preference[\egoIndex]$~\cite[Def.~12]{zanardiPosetalGamesEfficiency2022}.
The joint action space of the game is the product $\actionSpace := \prod_{\egoIndex \in \players} \actionSpace[\egoIndex]$ of the individual action spaces, and the preferences $\preference[\egoIndex]$ are defined over this joint action space as in \cref{subsubsec:preferences}.
We refer to the elements of $\actionSpace$ as \emph{action profiles} to emphasize that they consist of an action for each agent.
Given an action profile $\actionProfile \in \actionSpace$, we write $\actionProfile[\egoIndex] \in \actionSpace[\egoIndex]$ for the action of agent $\egoIndex \in \players$ and $\actionProfile[-\egoIndex]$ for the joint action of the other agents $\otherPlayers := \players \setminus \{ \egoIndex \}$.
For any $\egoIndex \in \players$, we use the pair $\langle \actionProfile[\egoIndex], \actionProfile[-\egoIndex] \rangle$ and the action profile $\actionProfile$ interchangeably.

Using the pre-orders $\actionPreceq[\preference[\egoIndex]]$ induced by the preferences to compare action profiles, we can transfer classical game theory concepts.
Rational agents always strive to play their best response to the actions of the others.
Given a posetal game $\game$, a \emph{best response} of agent $\egoIndex \in \players$ to the joint action $\actionProfile[-\egoIndex] \in \actionSpace[-\egoIndex]$ of the other agents is any action from the set~\cite[Def.~13]{zanardiPosetalGamesEfficiency2022}
\begin{equation*}
    \bestResponse{\preference[\egoIndex]}(\actionProfile[-\egoIndex]) := \{ \actionProfile[\egoIndex] \in \actionSpace[\egoIndex] \mid \forall \actionProfile[\egoIndex]' \in \actionSpace[\egoIndex] : \langle \actionProfile[\egoIndex], \actionProfile[-\egoIndex] \rangle \not\actionPrec[\preference[\egoIndex]] \langle \actionProfile[\egoIndex]', \actionProfile[-\egoIndex] \rangle \}.
\end{equation*}
An action profile $\actionProfile \in \actionSpace$ is a \emph{Nash equilibrium (\nasheq)} of $\game$ if and only if every agent plays a best response, \ie~\cite[Def.~14]{zanardiPosetalGamesEfficiency2022}
\begin{equation*}
    \forall \egoIndex \in \players : \actionProfile[\egoIndex] \in \bestResponse{\preference[\egoIndex]}(\actionProfile[-\egoIndex]).
\end{equation*}
Thus, at \nasheq, no agent can unilaterally deviate to strictly improve its outcome.
We denote the set of \nasheqs of $\game$ as $\nes{\game}$.

We can show that refining the preferences of the agents results in filtering the \nasheqs of the game.
\begin{proposition}[Equilibrium refinement~{\cite[Theorem~28]{zanardiPosetalGamesEfficiency2022}}]
    \label{prop:equilibrium_refinement}
    Suppose $\game := \langle \players, \{ \actionSpace[\egoIndex] \}_{\egoIndex \in \players}, \{ \preference[\egoIndex] \}_{\egoIndex \in \players} \rangle$ is a posetal game.
    Consider a second game $\game'$ that is equal to $\game$ except that every agent uses a refined preference, \ie $\game' = \langle \players, \{ \actionSpace[\egoIndex] \}_{\egoIndex \in \players}, \{ \preference[\egoIndex]' \}_{\egoIndex \in \players} \rangle$, where $\preference[\egoIndex]' \refines \preference[\egoIndex]$ for all $\egoIndex \in \players$.
    It holds $\nes{\game} \supseteq \nes{\game'}$.
\end{proposition}
\Cref{prop:equilibrium_refinement} states that the set of \nasheqs depends \emph{antitonically} on the preferences: refining the preferences of the agents can only remove joint behaviors, but never introduce new ones.

\subsection{Problem Statement}
\label{subsec:problem_statement}

Let $\scenario$ be a \emph{scenario} describing the involved agents $\players$ with their action spaces $\{ \actionSpace[\egoIndex] \}_{\egoIndex \in \players}$ and the environment.
To coordinate the agents in case multiple equilibria exist, we assume that the law provides a \emph{social metric}~$\socialMetric$, which maps action profiles to a totally ordered set and which the selected equilibrium should optimize.
We are now ready to state the problem.

\begin{problem}[Legal motion planning under private preferences]
\label{prob:first}
Given a scenario~$\scenario$ and the applicable legal requirements, (i) automatically derive preferences~$\{ \preference[\egoIndex] \}_{\egoIndex \in \players}$ that encode the legal requirements and strictly prioritize them over the individual objectives of the agents, which defines a posetal game~$\game$, and (ii) determine the socially optimal Nash equilibrium
\begin{equation}
    \label{eq:socially_optimal_true_ne}
    \actionProfile^* = \argmax_{\actionProfile \in \nes{\game}}(\socialMetric(\actionProfile)),
\end{equation}
using a distributed protocol that all agents follow and that does not require the agents to disclose their preferences.
\end{problem}
Ties in \eqref{eq:socially_optimal_true_ne} are broken by a fixed public rule, so that~$\actionProfile^*$ is unique.
Typically, the individual goals of each agent are private or hard to communicate so that the true preference of an agent is unknown to the others.
This is why part (ii) requires a distributed protocol rather than a centralized solver. 
We address part (i) in \cref{sec:guideline_games}, where~$\game$ becomes the true game~$\trueGame$, and part (ii) in \cref{sec:solution_protocols}.

\section{From Legal Requirements to Posetal Games}
\label{sec:guideline_games}

\begin{figure}[tb]
    \centering
    \begin{tikzpicture}
    \tikzset{trajectory/.style={-latex'}}
    \tikzset{op_area/.style={thick, TUMMediumGray}}
    \tikzset{node/.style={draw,rounded corners=2pt}}
    \tikzset{link/.style={-latex'}}

    \newcommand{\robotWithTrajs}[4][TUMBlack]{%
        \begin{scope}[shift={#3}, rotate=#4, trajectory/.append style={#1!60}]
            \robot[#1]{#2}
            \begin{scope}[on background layer]
                \draw[trajectory] (#2_front) -- +(0, 1);
                \draw[trajectory] (#2_front) to[out=90, in=-90] +(0.4, 1);
                \draw[trajectory] (#2_front) to[out=90, in=-90] +(-0.4, 1);
            \end{scope}
        \end{scope}
    }

    \begin{scope}[local bounding box=scenarios]
        \begin{scope}[shift={(0, -0.25)}, scale=0.4]
            \robotWithTrajs[TUMBlue]{r1}{(0, 0)}{-90}
            \robotWithTrajs[TUMOrange]{r2}{(3.5, 0)}{90}
            
            \draw[op_area] ($(r1_center) + (-0.5, 1.1)$) -- ($(r2_center) + (0.5, 1.1)$);
            \draw[op_area] ($(r1_center) + (-0.5, -1.1)$) -- ($(r2_center) + (0.5, -1.1)$);
        \end{scope}
        \begin{scope}[shift={(0, -1.75)}, scale=0.4]
            \robotWithTrajs[TUMBlue]{r1}{(0, 0)}{-90}
            \robotWithTrajs[TUMOrange]{r2}{(3.5, 0)}{90}
            \robotWithTrajs[TUMGreen]{r3}{(1.75, -2)}{0}
            
            \draw[op_area] ($(r1_center) + (-0.5, 1.1)$) -- ($(r2_center) + (0.5, 1.1)$);
            \draw[op_area] ($(r1_center) + (-0.5, -1.1)$) -| ($(r3_center) + (-1.1, -0.5)$);
            \draw[op_area] ($(r2_center) + (0.5, -1.1)$) -| ($(r3_center) + (1.1, -0.5)$);
        \end{scope}
    \end{scope}

    \NewDocumentCommand{\prefOne}{sm}{%
        \node[
            cloud,
            draw=TUMBlue,
            fill=TUMBlue!20,
            aspect=2,
            inner sep=6pt,
        ] (s1_1#2) at (0, -0.1) {};
        \begin{scope}[font=\scriptsize, shift={(s1_1#2)}]
            \node (#2_p1m1) at (0, 0.15) {$m_1$};
            \node (#2_p1m2) at (-0.25, -0.15) {$m_2$};
            \node (#2_p1m3) at (0.25, -0.15) {$m_3$};

            \IfBooleanT{#1}{
                \draw (#2_p1m1) -- (#2_p1m2);
                \draw (#2_p1m1) -- (#2_p1m3);
            }
        \end{scope}
        \node[
            cloud,
            draw=TUMOrange,
            fill=TUMOrange!20,
            aspect=2,
            inner sep=0pt,
        ] (s1_2#2) at ($(s1_1#2) + (0.75, -0.5)$) {\dots};
    }

    \NewDocumentCommand{\prefTwo}{sm}{%
        \node[
            cloud,
            draw=TUMBlue,
            fill=TUMBlue!20,
            aspect=1.8,
            inner sep=10pt,
        ] (s2_1#2) at (0, -0.1) {};
        \begin{scope}[font=\scriptsize, shift={(s2_1#2)}]
            \node (#2_p2m1) at (-0.25, 0.3) {$m'_1$};
            \node (#2_p2m2) at (0.25, 0.3) {$m'_2$};
            \node (#2_p2m3) at (-0.5, -0.3) {$m'_3$};
            \node (#2_p2m4) at (0, -0.3) {$m'_4$};
            \node (#2_p2m5) at (0.5, -0.3) {$m'_5$};

            \IfBooleanT{#1}{
                \draw (#2_p2m1) -- (#2_p2m3);
                \draw (#2_p2m1) -- (#2_p2m4);
                \draw (#2_p2m1) -- (#2_p2m5);
                \draw (#2_p2m2) -- (#2_p2m3);
                \draw (#2_p2m2) -- (#2_p2m4);
                \draw (#2_p2m2) -- (#2_p2m5);
            }
        \end{scope}
        \node[
            cloud,
            draw=TUMOrange,
            fill=TUMOrange!20,
            aspect=2,
            inner sep=0pt,
        ] (s2_2#2) at ($(s2_1#2) + (0.75, -0.8)$) {\dots};
        \node[
            cloud,
            draw=TUMGreen,
            fill=TUMGreen!20,
            aspect=2,
            inner sep=0pt,
        ] (s2_3#2) at ($(s2_1#2) + (-0.75, -0.8)$) {\dots};
    }

    \begin{scope}[local bounding box=metrics]
        \begin{scope}[shift={(3.7, 0)}, local bounding box=metrics1]
            \prefOne{metrics}
        \end{scope}
        \begin{scope}[shift={(3.7, -1.75)}, local bounding box=metrics2]
            \prefTwo{metrics}
        \end{scope}
    \end{scope}

    \begin{scope}[local bounding box=preferences]
        \begin{scope}[shift={(7, 0)}, local bounding box=preferences1]
            \prefOne*{pref}
        \end{scope}
        \begin{scope}[shift={(7, -1.75)}, local bounding box=preferences2]
            \prefTwo*{pref}
        \end{scope}
    \end{scope}

    \coordinate (label_y) at ($(metrics.north) + (0, 0.4)$) {};
    \node[font=\footnotesize] at (scenarios.north |- label_y) {Scenario};
    \node[font=\footnotesize] at (metrics.north |- label_y) {Legal metrics};
    \node[font=\footnotesize] at (preferences.north |- label_y) {Legal preferences};

    \coordinate (dots_y) at ($(metrics.south) + (0, -0.25)$) {};

    \coordinate (ruleset_x) at ($(scenarios.east)!0.5!(metrics.west)$) {};
    \node[node, font=\footnotesize] (ruleset) at (ruleset_x |- label_y) {$\Phi$};
    \draw[densely dotted] (ruleset) -- (ruleset |- dots_y);

    \coordinate (guideline_x) at ($(metrics.east)!0.5!(preferences.west)$) {};
    \node[node, font=\footnotesize] (guideline) at (guideline_x |- label_y) {$\guideline$};
    \draw[densely dotted] (guideline) -- (guideline |- dots_y);

    \coordinate (scenarios_east) at ($(scenarios.east) + (0.15, 0)$) {};
    \coordinate (metrics_west) at ($(metrics.west) + (-0.15, 0)$) {};
    \coordinate (metrics_east) at ($(metrics.east) + (0.15, 0)$) {};
    \coordinate (preferences_west) at ($(preferences.west) + (-0.15, 0)$) {};

    \draw[link] (scenarios_east |- metrics1.west) -- (metrics_west |- metrics1.west);
    \draw[fill] (ruleset |- metrics1.west) circle[radius=1pt];
    \draw[link] (scenarios_east |- metrics2.west) -- (metrics_west |- metrics2.west);
    \draw[fill] (ruleset |- metrics2.west) circle[radius=1pt];

    \draw[link] (metrics_east |- metrics1.east) -- (preferences_west |- metrics1.east);
    \draw[fill] (guideline |- metrics1.east) circle[radius=1pt];
    \draw[link] (metrics_east |- metrics2.east) -- (preferences_west |- metrics2.east);
    \draw[fill] (guideline |- metrics2.east) circle[radius=1pt];
\end{tikzpicture}
    \caption{%
        Deriving legal preferences from a generic rule set $\Phi$ and a scenario-independent guideline $\guideline$.
        For any concrete scenario,~$\Phi$ yields the set of metrics that each agent is legally required to consider, denoted by~$m_k$ and~$m'_k$.
        The guideline~$\guideline$ induces the legal preference for an agent over its legal metrics.
        The color of the clouds indicates the agent to which a set of metrics or a preference belongs.
    }
    \label{fig:scenario_to_legal_preference}
\end{figure}
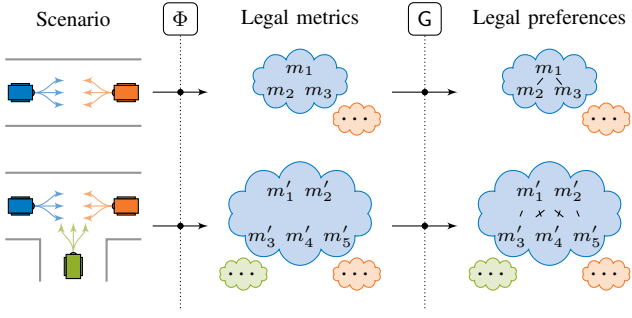

In this section, we address part (i) of \cref{prob:first} and present a process that transforms laws and a given scenario into a posetal game whose solutions meet legal requirements.
In \cref{subsec:legally_required_metrics}, we introduce legal metrics to quantify the legal concerns and explain how to derive these metrics from generic rules.
As shown in \cref{fig:scenario_to_legal_preference}, the relevant legal metrics depend on the agents and interactions present in a particular scenario so that manually specifying a preference for every possible scenario is infeasible.
\Cref{subsec:guidelines} therefore introduces guidelines to allow regulators to specify a scenario-independent partial order over legal concerns, such as collision avoidance.
From a guideline, we can automatically derive a preference over the legal metrics instantiated in any concrete scenario, as shown in \cref{fig:scenario_to_legal_preference}.
We leverage these preferences in \cref{subsec:deriving_motion_planning_game} to construct a posetal game modeling the motion planning problem that incorporates individual goals of the agents and different implementations of the legal requirements.
Moreover, we prove that the solutions of this game result in legal behavior.
Finally, \cref{subsec:handling_regulatory_change} discusses how the framework handles regulatory change.

\subsection{Legally Required Metrics}
\label{subsec:legally_required_metrics}

Legal requirements typically include pairwise rules that the ego agent must follow \wrt every other agent.
This is particularly common in traffic applications, including road~\cite{unitednationsunConventionRoadTraffic1968}, sea~\cite{internationalmaritimeorganizationimoCOLREGsConventionInternational1972}, and air~\cite{internationalcivilaviationorganizationicaoAnnex2Rules2005} traffic.
Many existing formalizations express each pairwise rule as a generic temporal logic specification $\varphi$ parameterized by the two involved agents, \eg, \cite{esterleFormalizingTrafficRules2020,maierhoferFormalizationInterstateTraffic2020,maierhoferFormalizationIntersectionTraffic2022,sunLawBreakerApproachSpecifying2023} for road traffic and \cite{krasowskiTemporalLogicFormalization2021,krasowskiProvableTrafficRule2024,torbenAutomaticSimulationbasedTesting2023,pedersenFormalizingTestingCollision2025} for sea traffic.
Temporal logics are well suited, as they can constrain the temporal evolution of a trajectory, \eg, the temporal operator $\globally$ demands that a property \emph{always} holds.
While the standard semantics of temporal logic are Boolean, indicating only whether a trajectory satisfies the specification, \emph{robustness measures}~\cite{fainekosRobustnessTemporalLogic2009} provide quantitative semantics.
The sign of the real-valued robustness value indicates satisfaction or violation and its magnitude reflects the corresponding margin.
The literature provides many robustness measures (see, \eg, \cite[Sec.~2.2]{bartocciSpecificationBasedMonitoringCyberPhysical2018} for a survey);
the most suitable choice depends on the particular rule being formalized~\cite{finkeldeiEnhancedTrafficRule2025}.

Robustness measures enable us to turn a temporal logic rule set into the numerical metrics required by our framework.
Let $\Phi$ be a set of generic temporal logic specifications that apply pairwise and $\scenario$ be a scenario.
For $\varphi \in \Phi$, let $\varphi_{i, j}$ denote the instantiation of $\varphi$ with the agents $i, j \in \players$.
The \emph{legal} metrics for $\egoIndex \in \players$ in $\scenario$ are
\begin{align}
    \nonumber
    \legalMetrics{\egoIndex}                                & := \{ m_{\varphi_{\egoIndex, \obstacleIndex}} \mid \varphi \in \Phi, \obstacleIndex \in \otherPlayers \},\text{ where} \\
    \label{eq:metric_from_robustness}
    m_{\varphi_{\egoIndex, \obstacleIndex}}(\actionProfile) & := \min(\robustness{\varphi_{\egoIndex, \obstacleIndex}}{\actionProfile}, 0)
\end{align}
for $\actionProfile \in \actionSpace$.
Here, $\robustness{\varphi_{\egoIndex, \obstacleIndex}}{\actionProfile}$ denotes the robustness of $\varphi_{\egoIndex, \obstacleIndex}$ for $\actionProfile$.
We cap the robustness at $0$ so that we penalize stronger violations more without distinguishing among satisfying behaviors (see \cite[Sec.~IV-A]{wongpiromsarnFormalSpecificationControl2026}).
This allows agents to optimize their individual objectives as long as they meet the legal requirements, instead of demanding maximally robust satisfaction of the rules.
Rules that apply only to the ego agent (e.g., speed limits) are handled analogously, with one metric per rule instead of one metric per contact agent.

We emphasize that leveraging temporal logic robustness is only one way to derive metrics from legal requirements.
At its core, our framework only requires a function that maps a scenario $\scenario$ and an agent $\egoIndex \in \players$ to the set of legal metrics $\legalMetrics{\egoIndex}$.
It therefore remains compatible with alternative methods, including hand-crafted numerical compliance scores such as \cite{woernerQuantifyingProtocolEvaluation2019}.

\subsection{Guidelines}
\label{subsec:guidelines}

To represent abstract legal concerns, we introduce a set of \emph{labels} $\metricTags$, each capturing a single aspect of the law.
For our example in \cref{fig:example_robots_in_corridor}, the labels could include $\collTag$ for collision avoidance and $\highVulnerabilityTag$ and $\lowVulnerabilityTag$ for the vulnerability of a robot.
A \emph{concern} $\concern \in \powerset{\metricTags}$ combines several labels to model a more specific aspect of the law.
To associate concrete metrics with concerns, we label each metric $m$ with a concern $\concernOfMetric{m}$, and say that it \emph{addresses} every less specific concern $\concern \subseteq \concernOfMetric{m}$.
In our example, a collision metric may carry the label $\collTag$ together with a label describing the vulnerability of the other robot, \eg, $\highVulnerabilityTag$.
This metric addresses both the general collision concern $\{ \collTag \}$ and the more specific concern $\{ \collTag, \highVulnerabilityTag \}$ that only captures collisions with highly vulnerable robots.

To model the partial order over abstract legal concerns, we introduce guidelines.
\begin{definition}[Guideline]
\label{def:guideline}
A \emph{guideline} is a poset~$\guideline := \langle \concerns_{\guideline},\posetPreceq_{\guideline} \rangle$,
where $\concerns_{\guideline}$ is a set of concerns and $\posetPreceq_{\guideline}$ specifies their relative priorities.
\end{definition}
Our representation of concerns as sets of labels provides a modular mechanism for expressing priorities at different levels of specificity.
In our example, we could start with a coarse guideline that prioritizes general collision avoidance through the concern $\{ \collTag \}$.
If vulnerability later becomes relevant to the legal priority, we can specialize the concern into
$\{ \collTag, \highVulnerabilityTag \}$ and
$\{ \collTag, \lowVulnerabilityTag \}$, which we can then order to prioritize avoiding collisions with highly vulnerable robots.
A guideline induces a preference over any given set of metrics.

\begin{definition}[Induced preference]
Given a set of metrics~$\metrics$, a guideline $\guideline = \langle \concerns_{\guideline}, \posetPreceq_{\guideline} \rangle$ induces the preference $\inducedPreference{\guideline} = \langle \metrics_{\guideline}, \posetPreceq_{\inducedPreference{\guideline}} \rangle$ with
\begin{equation}
    \label{eq:guideline_metrics}
    \metrics_{\guideline} := \{ m^\concern \mid \concern \in \concerns_{\guideline}, m \in \metrics \text{ s.t. } \concernOfMetric{m} \supseteq \concern \},
\end{equation}
and $m_1^{\concern_1} \posetPreceq_{\inducedPreference{\guideline}} m_2^{\concern_2}$ if and only if $m_1^{\concern_1} = m_2^{\concern_2}$ or $\concern_1 \posetPrec_{\guideline} \concern_2$.
In \eqref{eq:guideline_metrics}, $m^\concern$ is a renaming of the metric $m$ as explained in \cref{subsubsec:preferences}.
\end{definition}

Intuitively, the induced preference picks, for each concern in $\concerns_{\guideline}$, all metrics from $\metrics$ addressing that concern;
if a metric addresses multiple concerns, several copies of it appear in the preference.
We group the metrics according to the concern they were selected for.
The groups are prioritized according to $\posetPreceq_{\guideline}$, while metrics in the same group are incomparable.
\Cref{fig:induced_preference} shows an example.
\begin{figure}[tb]
    \centering
    \begin{tikzpicture}
    \tikzset{every node/.style={font=\footnotesize}}
    \tikzset{metric/.style={inner sep=1.5pt}}
    \tikzset{labels/.style={inner sep=1.5pt, anchor=west}}
    \tikzset{group/.style={draw=none, fill=#1!40, rounded corners=2pt, inner sep=1.5pt}}
    \tikzset{addresses/.style={-latex', thick, #1}}
    \tikzset{caption/.style={font=\scriptsize, anchor=north}}

    \def\rowdist{0.55cm}

    \node[labels] (m1) at (0, 0) {$\concernOfMetric{m_1} = \{ \metricTag{coll}, \metricTag{lv} \}$};
    \node[labels] (m2) at ($(m1.west) + (0, -\rowdist)$) {$\concernOfMetric{m_2} = \{ \metricTag{coll}, \metricTag{hv} \}$};
    \node[labels] (m3) at ($(m2.west) + (0, -\rowdist)$) {$\concernOfMetric{m_3} = \{ \metricTag{clear}, \metricTag{hv} \}$};
    \node[labels, TUMMediumGray] (m4) at ($(m3.west) + (0, -\rowdist)$) {$\concernOfMetric{m_4} = \{ \metricTag{effort} \}$};

    \coordinate (metrics_east) at (m3.east |- m1);
    \node[TUMBlue] (c1) at ($(metrics_east) + (1.9, -0.5 * \rowdist)$) {$\concern_1 = \{ \metricTag{coll} \}$};
    \node[TUMOrange] (c2) at ($(c1) + (0, -2 * \rowdist)$) {$\concern_2 = \{ \metricTag{hv} \}$};
    \draw (c1) -- (c2);

    \draw[addresses=TUMBlue] (m1.east) -- ($(c1.west) + (0, 0.05)$);
    \draw[addresses=TUMBlue] ($(m2.east) + (0, 0.05)$) -- ($(c1.west) + (0, -0.05)$);
    \draw[addresses=TUMOrange] ($(m2.east) + (0, -0.05)$) -- ($(c2.west) + (0, 0.05)$);
    \draw[addresses=TUMOrange] (m3.east) -- ($(c2.west) + (0, -0.05)$);
    \node[cross=3pt, draw=TUMMediumGray, thick] at ($(m4.east) + (0.3, 0)$) {};

    \node[metric] (m1c1) at ($(c1.east) + (1.25, 0)$) {$m_1^{\concern_1}$};
    \node[metric, right=0.2cm of m1c1] (m2c1) {$m_2^{\concern_1}$};
    \node[metric] (m2c2) at (m1c1 |- c2) {$m_2^{\concern_2}$};
    \node[metric] (m3c2) at (m2c1 |- c2) {$m_3^{\concern_2}$};

    \begin{scope}[on background layer]
        \node[group=TUMBlue,fit=(m1c1) (m2c1)] (g1) {};
        \node[group=TUMOrange,fit=(m2c2) (m3c2)] (g2) {};
    \end{scope}

    \draw (m1c1) -- (m2c2);
    \draw (m1c1) -- (m3c2);
    \draw (m2c1) -- (m2c2);
    \draw (m2c1) -- (m3c2);

    \coordinate (metrics_mid) at ($(m2)!0.5!(m3)$);
    \node[font=\Large] at ($(c1.east |- metrics_mid)!0.5!(g1.west |- metrics_mid)$) {$\mapsto$};

    \coordinate (baseline) at ($(m4.south) + (0, -0.05)$);
    \coordinate (metrics_mid) at ($(m1.west)!0.5!(metrics_east)$);
    \node[caption] at (metrics_mid |- baseline) {$\metrics$};
    \node[caption] at (c1 |- baseline) {$\guideline$};
    \node[caption] at (g1 |- baseline) {$\inducedPreference{\guideline}$};
\end{tikzpicture}
    \caption{%
        Inducing a preference from a guideline.
        Arrows indicate which concerns a metric addresses;
        colors indicate the concern a metric was selected for.
        Since $m_1$ addresses both concerns from $\guideline$, two copies of it appear in $\inducedPreference{\guideline}$; $m_4$ addresses no concern and is dropped.
        Metrics within a group are incomparable, and the groups are prioritized according to $\guideline$.
        We abbreviate the labels $\collTag$, $\highVulnerabilityTag$, and $\lowVulnerabilityTag$ as $\metricTag{coll}$, $\metricTag{hv}$, and $\metricTag{lv}$, respectively.
    }
    \label{fig:induced_preference}
\end{figure}
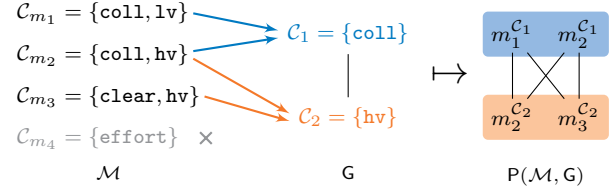
To prove that $\inducedPreference{\guideline}$ is well-defined, we need to show that $\posetPreceq_{\inducedPreference{\guideline}}$ is indeed a partial order.
\begin{lemma}[Well-definedness of induced preferences]
    \label{lem:induced_preference_well_defined}
    Given a guideline $\guideline = \langle \concerns_{\guideline}, \posetPreceq_{\guideline} \rangle$ and a set of metrics $\metrics$, the induced preference $\inducedPreference{\guideline} = \langle \metrics_{\guideline}, \posetPreceq_{\inducedPreference{\guideline}} \rangle$ is well-defined, \ie $\posetPreceq_{\inducedPreference{\guideline}}$ is a partial order.
\end{lemma}
\begin{proof}
    Observe that $\posetPreceq_{\inducedPreference{\guideline}}$ is reflexive by construction, and transitivity follows from the transitivity of~$\posetPrec_{\guideline}$.
    For antisymmetry, let $m_1^{\concern_1} \posetPreceq_{\inducedPreference{\guideline}} m_2^{\concern_2}$ with $m_1^{\concern_1} \neq m_2^{\concern_2}$.
    Thus, we must have $\concern_1 \posetPrec_{\guideline} \concern_2$.
    For $m_1^{\concern_1} \posetSucceq_{\inducedPreference{\guideline}} m_2^{\concern_2}$, we would need $\concern_1 \posetSucc_{\guideline} \concern_2$, contradicting the asymmetry of $\posetPrec_{\guideline}$. Hence,~$\posetPreceq_{\inducedPreference{\guideline}}$ is a partial order.
\end{proof}

The induced preference is inherently \emph{non-compensatory}.
Consider two metrics $m_1^{\concern_1}$ and $m_2^{\concern_2}$ for which the guideline $\guideline$ does not specify a clear priority, \ie neither $\concern_1 \posetPrec_{\guideline} \concern_2$ nor $\concern_1 \posetSucc_{\guideline} \concern_2$.
In the induced preference, $m_1^{\concern_1}$ and $m_2^{\concern_2}$ are incomparable so that a gain in one metric cannot compensate for a loss in the other;
instead, we accept every optimal trade-off between them.
In particular, action profiles $\actionProfile, \actionProfile' \in \actionSpace$ which differ only in that $m_1^{\concern_1}(\actionProfile) > m_1^{\concern_1}(\actionProfile')$ and $m_2^{\concern_2}(\actionProfile) < m_2^{\concern_2}(\actionProfile')$ remain incomparable according to \eqref{eq:action_order}.
Therefore, guidelines are well-suited to model the non-compensatory priorities of the law, while compensatory rules, \eg, aggregation in the sense of \cref{prop:preference_aggregation}, are reserved for subsequent private modeling choices of individual agents.

\begin{remark}[Good guideline design]
    \label{rem:good_guideline_design}
    Designing \emph{unambiguous} guidelines, \ie ensuring that each metric addresses at most one concern from $\concerns_{\guideline}$, improves the interpretability of guidelines and the induced preference.
    However, we do not impose unambiguity as a hard constraint since it would limit the expressiveness of the framework.
    Moreover, it is good practice to ensure that every legal metric addresses at least one concern so that all of them are included in the induced preference.
\end{remark}

\subsection{Deriving the Motion Planning Game}
\label{subsec:deriving_motion_planning_game}

Let us fix a scenario $\scenario$ and suppose the guideline $\legalGuideline$ captures the legal requirements.
Recall from \cref{subsec:legally_required_metrics} that $\legalMetrics{\egoIndex}$ denotes the legal metrics for the agent $\egoIndex \in \players$ in $\scenario$.
For each agent, we can automatically derive its \emph{legal} preference as $\inducedPreference[\legalMetrics{\egoIndex}]{\legalGuideline}$, \ie the preference that $\legalGuideline$ induces over its legal metrics.
Evaluating actions purely on legal grounds is unrealistic, since each agent typically also cares about performance metrics that measure how well it achieves its individual goals.
Moreover, in cases where the law leaves the priorities between metrics unspecified, different agents may implement the legal requirements differently.
Thus, we assume that each agent $\egoIndex \in \players$ possesses a \emph{true} preference $\preference[\egoIndex]$ that refines $\inducedPreference[\legalMetrics{\egoIndex}]{\legalGuideline}$.
This assumption allows agents to add objectives and priorities wherever the law is vague, but never to override a priority that the law prescribes.
It is the formal counterpart of the requirement that legal requirements take precedence over individual objectives.
Agents that violate this assumption are \emph{defectors}.
As we show in \cref{subsubsec:identifying_defectors}, a defector whose action is not a legal best response can be identified from public information alone.
In our example from \cref{fig:example_robots_in_corridor}, the robot may use \cref{prop:preference_augmentation} to include a performance metric measuring control effort below the legal preference.
Using the true preferences, we obtain the \emph{true} game
\begin{equation*}
    \trueGame := \langle \players, \{ \actionSpace[\egoIndex] \}_{\egoIndex \in \players}, \{ \preference[\egoIndex] \}_{\egoIndex \in \players} \rangle,
\end{equation*}
whose \nasheqs are incentive-stable solutions to the motion planning problem.

To characterize legality, we call an action profile $\actionProfile \in \actionSpace$ \emph{legal} if every agent chooses an action that minimally violates the law given the joint action of the others.
Note that the minimum-violation condition results in full compliance if the rules do not conflict.
Formally, this corresponds to $\actionProfile[\egoIndex] \in \bestResponse{\inducedPreference[\legalMetrics{\egoIndex}]{\legalGuideline}}(\actionProfile[-\egoIndex])$ for every agent $\egoIndex \in \players$.
Therefore, the set of all legal action profiles precisely corresponds to the \nasheqs of the \emph{legal} game
\begin{equation*}
    \legalGame := \langle \players, \{ \actionSpace[\egoIndex] \}_{\egoIndex \in \players}, \{ \inducedPreference[\legalMetrics{\egoIndex}]{\legalGuideline} \}_{\egoIndex \in \players} \rangle.
\end{equation*}
Consequently, we can prove that all solutions of $\trueGame$ result in legal behavior by applying \cref{prop:equilibrium_refinement}.
\begin{theorem}[Legality]
    \label{thm:true_game_is_legal}
    We have $\nes{\trueGame} \subseteq \nes{\legalGame}$.
    Specifically, the \nasheqs of $\trueGame$ are exactly those legal action profiles at which every agent plays a best response under its true preference, \ie
    \begin{equation*}
        \nes{\trueGame} = \{ \actionProfile \in \nes{\legalGame} \mid \forall \egoIndex \in \players : \actionProfile[\egoIndex] \in \bestResponse{\preference[\egoIndex]}(\actionProfile[- \egoIndex]) \}.
    \end{equation*}
\end{theorem}
\begin{proof}
    The preference of each agent in $\trueGame$ is a refinement of its preference in $\legalGame$.
    Thus, \cref{prop:equilibrium_refinement} directly applies and yields $\nes{\trueGame} \subseteq \nes{\legalGame}$.
    The second statement then follows from the definition of a \nasheq.
\end{proof}

In words, private objectives can only select among legal behaviors and never turn an illegal action profile into an equilibrium.
Since \cref{prop:equilibrium_refinement} holds for any refinement, this guarantee does not depend on the particular private objectives of the agents, which can therefore remain undisclosed.

\subsection{Handling Regulatory Change}
\label{subsec:handling_regulatory_change}

Laws evolve over time, so our framework needs to accommodate regulatory change.
We focus on revisions that refine the specification of legal behavior so that previously illegal behavior remains forbidden.
Such revisions are practically relevant.
For example, a regulatory scoping exercise by the International Maritime Organization found that the current \colregs should keep as much of their content as possible when adapted for autonomous vessels~\cite[pp.~86--87]{internationalmaritimeorganizationimoOutcomeRegulatoryScoping2021}.

Formally, this means that the legal preference under the revised law should be a refinement of that under the original law.
Consequently, we can freely modify the set of legal metrics and their labels as long as the set of metrics \eqref{eq:guideline_metrics} picked by the guideline $\legalGuideline$ remains the same.
To enable changes of the guideline itself, we introduce the notion of guideline refinement.
\begin{definition}[Guideline refinement]
The guideline $\guideline'$ \emph{refines} $\guideline$, denoted $\guideline' \refines \guideline$, if its induced preference $\inducedPreference{\guideline'}$ is a refinement of $\inducedPreference{\guideline}$ for all sets of metrics $\metrics$ and label assignments.
\end{definition}

By construction, the induction map associating a guideline~$\guideline$ with its induced preference $\inducedPreference{\guideline}$ is therefore monotone in $\guideline$, \ie
\begin{equation*}
    \guideline \refinedby \guideline' \implies \inducedPreference{\guideline} \refinedby \inducedPreference{\guideline'},
\end{equation*}
for any set of metrics $\metrics$.
Composed with the antitone equilibrium map of \cref{prop:equilibrium_refinement}, refining a guideline can only shrink the set of legal equilibria.
The entire pipeline thus reduces to a single monotonicity statement, summarized in \cref{fig:induction_homomorphism}, in which every map preserves refinement except the equilibrium map, which reverses it.

Subsequently, we introduce three guideline refinement operations to support commonly required changes.
To prove their correctness, we reduce guideline refinements to the preference-level operations of \cref{subsec:posetal_games} so that we need not reason about action orders from scratch.
Beyond handling regulatory changes, the refinement operations also enable us construct the final legal guideline by composing simple and understandable guidelines covering different aspects of the law.

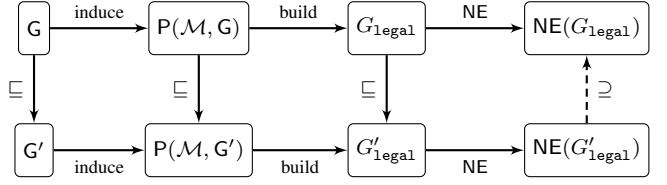
\begin{figure}[t]
    \centering
    \begin{tikzpicture}[node distance=9mm and 13mm]
    \tikzset{every node/.style={font=\footnotesize}}
    \tikzset{stage/.style={draw, rounded corners=2pt, minimum height=7mm, inner sep=3pt, align=center}}
    \tikzset{mono/.style={-latex', thick}}
    \tikzset{anti/.style={-latex', thick, densely dashed}}
    \tikzset{reflab/.style={font=\scriptsize, midway}}

    \node[stage] (G) {$\guideline$};
    \node[stage, right=of G] (Pi) {$\inducedPreference[\metrics]{\guideline}$};
    \node[stage, right=of Pi] (Ga) {$\legalGame$};
    \node[stage, right=of Ga] (NE) {$\nes{\legalGame}$};
    \node[stage, below=of G] (Gp) {$\guideline'$};
    \node[stage, below=of Pi] (Pip) {$\inducedPreference[\metrics]{\guideline'}$};
    \node[stage, below=of Ga] (Gap) {$\legalGame'$};
    \node[stage, below=of NE] (NEp) {$\nes{\legalGame'}$};
    
    \draw[mono] (G) -- node[reflab, above]{induce} (Pi);
    \draw[mono] (Pi) -- node[reflab, above]{build} (Ga);
    \draw[mono] (Ga) -- node[reflab, above]{$\mathsf{NE}$} (NE);
    \draw[mono] (Gp) -- node[reflab, below]{induce} (Pip);
    \draw[mono] (Pip) -- node[reflab, below]{build} (Gap);
    \draw[mono] (Gap) -- node[reflab, below]{$\mathsf{NE}$} (NEp);
    \draw[mono] (G) -- node[reflab, left]{$\refinedby$} (Gp);
    \draw[mono] (Pi) -- node[reflab, left]{$\refinedby$} (Pip);
    \draw[mono] (Ga) -- node[reflab, left]{$\refinedby$} (Gap);
    \draw[anti] (NEp) -- node[reflab, right]{$\supseteq$} (NE);
\end{tikzpicture}
    \caption{%
    Propagation of a guideline refinement~$\guideline \refinedby \guideline'$ through the pipeline. The horizontal maps take a guideline to its induced preference, the assembled legal game, and its equilibria.
    Induction and game construction preserve~$\refinedby$ (solid vertical arrows, monotone homomorphisms), whereas the equilibrium map is \emph{antitone} (dashed arrow). Hence, a finer guideline can only shrink the set of joint behaviors, i.e.,~$\nes{\legalGame'} \subseteq \nes{\legalGame}$. 
    }
    \label{fig:induction_homomorphism}
\end{figure}

\subsubsection{Priority Refinement}
\label{subsubsec:guideline_priority_refinement}

Regulators may need to clarify priorities between concerns that they previously left unspecified.
This corresponds to \enquote{adding an edge} to $\posetPreceq_{\guideline}$.
Analogous to \cref{prop:preference_priority_refinement} for preferences, this \emph{priority refinement} constitutes a refinement operation on guidelines.
\begin{proposition}[Guideline priority refinement]
    \label{prop:guideline_priority_refinement}
    Let $\guideline = \langle \concerns_{\guideline}, \posetPreceq_{\guideline} \rangle$ be a guideline.
    Every guideline $\guideline' = \langle \concerns_{\guideline}, \posetPreceq_{\guideline'} \rangle$ with ${\posetPreceq_{\guideline}} \subseteq {\posetPreceq_{\guideline'}}$ refines $\guideline$.
\end{proposition}
\begin{proof}
    The induced metric set \eqref{eq:guideline_metrics} depends only on $\concerns_{\guideline}$, which is unchanged, so $\metrics_{\guideline'} = \metrics_{\guideline}$.
    Given $m_1^{\concern_1} \posetPrec_{\inducedPreference{\guideline}} m_2^{\concern_2}$, we find $\concern_1 \posetPrec_{\guideline} \concern_2$.
    As ${\posetPreceq_{\guideline}} \subseteq {\posetPreceq_{\guideline'}}$, we also have $\concern_1 \posetPrec_{\guideline'} \concern_2$, so $m_1^{\concern_1} \posetPrec_{\inducedPreference{\guideline'}} m_2^{\concern_2}$.
    Thus, \cref{prop:preference_priority_refinement} yields that $\inducedPreference{\guideline'}$ refines $\inducedPreference{\guideline}$ since ${\posetPreceq_{\inducedPreference{\guideline}}} \subseteq {\posetPreceq_{\inducedPreference{\guideline'}}}$.
\end{proof}

A regulator typically wishes to amend an existing guideline with a set of additional priorities and expects a canonical, least-restrictive way to enforce them.
Formally, this corresponds to finding a $\refinedby$-minimal priority refinement of the guideline.
\begin{corollary}[Minimal priority refinement]
    \label{cor:minimal_priority_refinement}
    Let $\guideline = \langle \concerns_{\guideline}, \posetPreceq_{\guideline} \rangle$ be a guideline and $D \subseteq \concerns_{\guideline} \times \concerns_{\guideline}$ a set of required priorities.
    Further, let $({\posetPreceq_{\guideline}} \cup D)^*$ be the reflexive-transitive hull of ${\posetPreceq_{\guideline}} \cup D$.
    If the hull is antisymmetric, $\guideline^\star := \langle \concerns_{\guideline}, ({\posetPreceq_{\guideline}} \cup D)^* \rangle$ is the unique $\refinedby$-minimal priority refinement of $\guideline$ enforcing $D$.
    Otherwise, ${\posetPreceq_{\guideline}} \cup D$ contains a cycle showing that $D$ cannot be enforced consistently with existing regulations.
\end{corollary}
\begin{proof}
    Let $R := ({\posetPreceq_{\guideline}} \cup D)^*$.
    If $R$ is antisymmetric, it is the smallest partial order with ${\posetPreceq_{\guideline}} \subseteq R$ that enforces $D$.
    Since ${\posetPreceq_{\guideline}} \subseteq R$, we have $\guideline^\star \refines \guideline$ by \cref{prop:guideline_priority_refinement}.
    The order of any priority refinement $\guideline'$ of $\guideline$ that enforces $D$ must contain $R$, so $\guideline^\star \refinedby \guideline'$ by \cref{prop:guideline_priority_refinement}.
    Uniqueness follows as $R$ is the smallest partial order with the required properties.
    A non-antisymmetric $R$ contains a cycle that any extending partial order would inherit.
\end{proof}
In words,~$\guideline^\star$ adds the priorities in~$D$ together with those implied by transitivity, and every other priority refinement that enforces~$D$ is at least as restrictive.
If~$D$ is inconsistent with the existing guideline, the cycle identifies the conflicting priorities, which the regulator must resolve.
\subsubsection{Injection}
\label{subsubsec:guideline_injection}

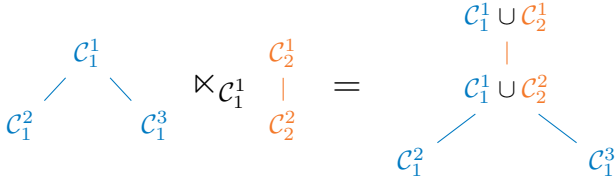
\begin{figure}
    \centering
    \begin{tikzpicture}
    \tikzset{hasse/.style={inner sep=0pt}}

    \node[hasse] (g) at (0, 0) {
        \begin{tikzcd}[row sep=small,column sep=tiny,cells={color=TUMBlue},arrows={color=TUMBlue}]
            & \concern_1^1 \ar[ld,dash] \ar[rd,dash] & \\
            \concern_1^2 & & \concern_1^3
        \end{tikzcd}
    };

    \node[hasse,right=1cm of g] (g_prime) {
        \begin{tikzcd}[row sep=small,column sep=tiny,cells={color=TUMOrange},arrows={color=TUMOrange}]
            \concern_2^1 \ar[d,dash] \\
            \concern_2^2
        \end{tikzcd}
    };

    \node[hasse,right=1cm of g_prime] (res) {
        \begin{tikzcd}[row sep=small,column sep=tiny,cells={color=TUMBlue},arrows={color=TUMBlue}]
            & \color{black} {\color{TUMBlue} \concern_1^1} \cup {\color{TUMOrange} \concern_2^1} \ar[d,dash,color=TUMOrange] & \\
            & \color{black} {\color{TUMBlue} \concern_1^1} \cup {\color{TUMOrange} \concern_2^2} \ar[ld,dash] \ar[rd,dash] & \\
            \concern_1^2 & & \concern_1^3
        \end{tikzcd}
    };

    \node[font=\Large] at ($(g.east)!0.5!(g_prime.west)$) {$\inject{\concern_1^1}$};
    \node[font=\Large] at ($(g_prime.east)!0.5!(res.west)$) {$=$};
\end{tikzpicture}
    \caption{%
        Visualization of guideline injection using Hasse diagrams.
        The colors indicate which parts of the refined guideline originate from which inputs.
        Continuing the example from \cref{fig:example_robots_in_corridor}, $\concern_1^1$ could be the $\{ \collTag \}$ concern, which is split according to a guideline prioritizing $\{ \highVulnerabilityTag \}$ over $\{ \lowVulnerabilityTag \}$.
    }
    \label{fig:guideline_injection}
\end{figure}

While \cref{prop:guideline_priority_refinement} allows refining priorities between concerns, it does not enable regulators to specify priorities within an existing concern.
Therefore, we introduce a mechanism for \enquote{splitting} a concern into several more specific ones that can be prioritized individually.
To this end, we define the \emph{injection} operator $\inject{\hat{\concern}}$.
Given guidelines $\guideline_k = \langle \concerns_{\guideline_k}, \posetPreceq_{\guideline_k} \rangle$ for $k \in \{ 1, 2 \}$ and a concern $\hat{\concern} \in \concerns_{\guideline_1}$, $\guideline_1 \inject{\hat{\concern}} \guideline_2$ injects $\guideline_2$ into the concern $\hat{\concern}$ in $\guideline_1$;
\cref{fig:guideline_injection} shows an abstract example.
Intuitively, the resulting guideline first orders metrics according to $\guideline_1$ and then refines the priority of metrics addressing $\hat{\concern}$ according to $\guideline_2$.
Formally, $\guideline_1 \inject{\hat{\concern}} \guideline_2$ is the guideline $\guideline = \langle \concerns_{\guideline}, \posetPreceq_{\guideline} \rangle$, where
\begin{subequations}
    \label{eq:injection}
    \begin{align}
        \label{subeq:injection:set}
        \concerns_{\guideline} :={}      & (\concerns_{\guideline_1} \setminus \{ \hat{\concern} \}) \cup \{ \hat{\concern} \cup \concern_2 \mid \concern_2 \in \concerns_{\guideline_2} \} \text{ and}                          \\
        \label{subeq:injection:none}
        {\posetPreceq_{\guideline}} :={} & \hphantom{{}\cup{}} \{ \langle \concern_1, \concern'_1 \rangle \in {\posetPreceq_{\guideline_1}} \mid \concern_1, \concern'_1 \in \concerns_{\guideline_1} \setminus \{ \hat{\concern} \} \} \\
        \label{subeq:injection:left}
                                         & {}\cup \{ \langle \hat{\concern} \cup \concern_2, \concern_1 \rangle \mid \langle \hat{\concern}, \concern_1 \rangle \in {\posetPrec_{\guideline_1}}, \concern_2 \in \concerns_{\guideline_2} \}    \\
        \label{subeq:injection:right}
                                         & {}\cup \{ \langle \concern_1, \hat{\concern} \cup \concern_2 \rangle \mid \langle \concern_1, \hat{\concern} \rangle \in {\posetPrec_{\guideline_1}}, \concern_2 \in \concerns_{\guideline_2} \}    \\
        \label{subeq:injection:both}
                                         & {}\cup \{ \langle \hat{\concern} \cup \concern_2, \hat{\concern} \cup \concern'_2 \rangle \mid \langle \concern_2, \concern'_2 \rangle \in {\posetPreceq_{\guideline_2}} \}.
    \end{align}
\end{subequations}
For \eqref{eq:injection} to define a guideline, we require that the injected concerns are new, i.e.,~$\hat{\concern} \cup \concern_2\notin \concerns_{\guideline_1} \setminus \{ \hat{\concern}\}$ for all~$\concern_2\in \concerns_{\guideline_2}$.
Otherwise, an injected concern could inherit conflicting priorities from~$\hat{\concern}$ and from an existing concern of~$\guideline_1$, and~$\posetPreceq_{\guideline}$ would not be antisymmetric.
Under this condition, \eqref{subeq:injection:left} to \eqref{subeq:injection:both} copy the priorities of~$\hat{\concern}$ to every injected concern and order the injected concerns among themselves according to~$\guideline_2$.
Injection also preserves unambiguity, i.e., if~$\guideline_1$ and~$\guideline_2$ are unambiguous, so is~$\guideline_1 \inject{\hat{\concern}} \guideline_2$.
Indeed, a metric addressing an injected concern~$\hat{\concern} \cup \concern_2$ addresses~$\hat{\concern}$ and thus no other concern of~$\guideline_1$, and it cannot address two injected concerns since it would then address two concerns of~$\guideline_2$.
Furthermore, injection results in a refinement if every metric addressing~$\hat{\concern}$ also addresses a concern from $\concerns_{\guideline_2}$.
\begin{proposition}[Guideline injection]
    \label{prop:guideline_injection}
    Let $\guideline_k = \langle \concerns_{\guideline_k}, \posetPreceq_{\guideline_k} \rangle$ be guidelines for $k \in \{ 1, 2 \}$ and let $\hat{\concern} \in \concerns_{\guideline_1}$.
    For all metrics $m$ so that $\concernOfMetric{m} \supseteq \hat{\concern}$, suppose that there is $\concern_2 \in \concerns_{\guideline_2}$ with $\concernOfMetric{m} \supseteq \concern_2$.
    Then, the guideline $\guideline_1 \inject{\hat{\concern}} \guideline_2$ is a refinement of $\guideline_1$.
\end{proposition}
\begin{proof}
    We have $\metrics_{\guideline_1 \inject{\hat{\concern}} \guideline_2} = \metrics_{\guideline_1}$ modulo metric renaming and duplication, since every metric addressing $\hat{\concern}$ also addresses some $\concern_2 \in \concerns_{\guideline_2}$ by assumption.
    With \eqref{subeq:injection:none} to \eqref{subeq:injection:right} we reproduce every relation of $\posetPreceq_{\inducedPreference{\guideline_1}}$ and \eqref{subeq:injection:both} only adds relations among metrics addressing $\hat{\concern}$, which are incomparable under $\posetPreceq_{\inducedPreference{\guideline_1}}$.
    Thus, ${\posetPreceq_{\inducedPreference{\guideline_1}}} \subseteq {\posetPreceq_{\inducedPreference{\guideline_1 \inject{\hat{\concern}} \guideline_2}}}$ modulo metric renaming and duplication, so that \cref{prop:preference_priority_refinement} concludes.
\end{proof}

\subsubsection{Augmentation}
\label{subsubsec:guideline_augmentation}

Finally, regulators need a way to combine two guidelines while prioritizing one over the other, \eg, to place generally applicable rules below jurisdiction-specific rules.
\emph{Augmenting} a guideline by appending a disjoint guideline at the bottom of the poset is a refinement operation, analogous to \cref{prop:preference_augmentation} for preferences.
\begin{proposition}[Guideline augmentation]
    \label{prop:guideline_augmentation}
    For $k \in \{ 1, 2 \}$, let $\guideline_k = \langle \concerns_{\guideline_k}, \posetPreceq_{\guideline_k} \rangle$ be guidelines so that $\concerns_{\guideline_1} \cap \concerns_{\guideline_2} = \emptyset$.
    The guideline $\guideline = \langle \concerns_{\guideline}, \posetPreceq_{\guideline} \rangle$ with $\concerns_{\guideline} := \concerns_{\guideline_1} \cup \concerns_{\guideline_2}$ and
    \begin{equation}
        \label{eq:guideline_augmentation_order}
        {\posetPreceq_{\guideline}} := {\posetPreceq_{\guideline_1}} \cup {\posetPreceq_{\guideline_2}} \cup \{ \langle \concern_2, \concern_1 \rangle \mid \concern_1 \in \concerns_{\guideline_1}, \concern_2 \in \concerns_{\guideline_2} \}
    \end{equation}
    is a refinement of $\guideline_1$.
\end{proposition}
\begin{proof}
    By \eqref{eq:guideline_metrics}, we have $\metrics_{\guideline} = \metrics_{\guideline_1} \cup \metrics_{\guideline_2}$.
    Further, since $\concerns_{\guideline_1} \cap \concerns_{\guideline_2} = \emptyset$, it is also $\metrics_{\guideline_1} \cap \metrics_{\guideline_2} = \emptyset$.
    The added relations $\langle \concern_2, \concern_1 \rangle$ from \eqref{eq:guideline_augmentation_order} rank every concern of $\guideline_2$ below every concern of $\guideline_1$, so every metric of $\metrics_{\guideline_2}$ sits strictly below every metric of $\metrics_{\guideline_1}$ in $\inducedPreference{\guideline}$.
    With this, \cref{prop:preference_augmentation} yields that $\guideline$ refines $\guideline_1$.
\end{proof}

\section{Solution Protocols}
\label{sec:solution_protocols}

\begin{table}
    \centering
    \caption{Solution protocols, the information that each agent requires about the other agents, and whether the protocol is guaranteed to commit to the socially optimal equilibrium}
    \label{tab:solution_protocols}
    \renewcommand{\arraystretch}{1.2}
    \small
    \begin{tabular}{@{}lp{0.35\columnwidth}>{\centering\arraybackslash}p{5em}@{}}
        \toprule
        \textbf{Protocol}    & \textbf{Known about others}                 & \textbf{Optimal?} \\
        \midrule
        Complete information & action space, \newline true preference      & \checkmark        \\
        Belief play          & action space, \newline belief on preference & \checkmark        \\
        Distributed \ibr     & none                                        & after full search \\
        \bottomrule
    \end{tabular}
\end{table}

In this section, we address part (ii) of \cref{prob:first}, which requires computing the socially optimal \nasheq $\actionProfile^*$~\eqref{eq:socially_optimal_true_ne} of the true game $\trueGame$ derived in \cref{sec:guideline_games}.
To this end, we propose the solution protocols summarized in \cref{tab:solution_protocols}.
The design of a solution protocol depends strongly on which information each agent has about the others.
If the agents have \emph{complete information} about the others, every agent can independently compute $\actionProfile^*$ without communication as they can directly construct $\trueGame$.
However, assuming complete information is rarely realistic:
in many application domains, the agents do not know the true preferences of the others.
Thus, the complete information protocol is not deployable and can only serve as a baseline for evaluation.

Incomplete information about the preferences of the agents could be modeled by adapting Bayesian games~\cite{harsanyiGamesIncompleteInformation1967} or robust games~\cite{aghassiRobustGameTheory2006}, which provide equilibrium concepts that account for such uncertainty, to the posetal setting.
However, since each agent knows its own true preference, we aim to recover the true equilibria through communication between the agents, rather than reasoning about equilibria under uncertainty.
We assume that communication is limited to action profiles, as allowing the agents to share their true preferences would reduce the problem to the complete information case.
This makes our setting similar to uncoupled dynamics~\cite{hartUncoupledDynamicsNot2003}, except that we use explicit communication instead of repeated play.
Subsequently, we propose two protocols:
\emph{belief play} recovers the true \nasheqs from conservative local approximations, whereas \emph{distributed iterative best-response (\ibr)} requires no information about the action spaces or preferences of other agents.

\subsection{Belief Play}
\label{subsec:belief_play}

For our first solution protocol, we assume that every agent $\egoIndex \in \players$ knows the action space of all other agents.
While the true preference of each agent remains private, we suppose that the agent $\egoIndex$ has a belief $\preference[\obstacleIndex]^{\egoIndex}$ about the true preference of agent $\obstacleIndex \in \players$, where $\preference[\egoIndex]^{\egoIndex} = \preference[\egoIndex]$, since every agent knows its own true preference.
We assume that the beliefs are conservative, \ie the true preference $\preference[\obstacleIndex]$ is a refinement of the belief $\preference[\obstacleIndex]^{\egoIndex}$ for all $\egoIndex, \obstacleIndex \in \players$.
In words, a belief may omit objectives and priorities of another agent, but it never attributes to that agent a priority it does not have.
Based on its beliefs, the agent $\egoIndex$ constructs its \emph{local} game
\begin{equation*}
    \localGame{\egoIndex} := \langle \players, \{ \actionSpace[\obstacleIndex] \}_{\obstacleIndex \in \players}, \{ \preference[\obstacleIndex]^{\egoIndex} \}_{\obstacleIndex \in \players} \rangle.
\end{equation*}

We emphasize that the agents do not strictly require side information to form conservative beliefs.
The legal guideline $\legalGuideline$ and the set of legal metrics $\legalMetrics{\obstacleIndex}$ are publicly known, so agent $\egoIndex \in \players$ can construct the legal preference $\inducedPreference[\legalMetrics{\obstacleIndex}]{\legalGuideline}$ of agent $\obstacleIndex \in \otherPlayers$.
As all agents follow the legal requirements, the legal preference is the most conservative belief about $\preference[\obstacleIndex]$, so that $\inducedPreference[\legalMetrics{\obstacleIndex}]{\legalGuideline} \refinedby \preference[\obstacleIndex]^{\egoIndex} \refinedby \preference[\obstacleIndex]$.

Leveraging the conservative beliefs, we can solve $\trueGame$ by intersecting the \nasheqs of the local games.
This fundamental property forms the basis of our \emph{belief play} protocol.
\begin{theorem}[Belief play]
    \label{thm:global_nes_eq_local_nes}
    Suppose that $\preference[\obstacleIndex]$ is a refinement of $\preference[\obstacleIndex]^{\egoIndex}$ for all distinct $\egoIndex, \obstacleIndex \in \players$.
    The \nasheqs of $\trueGame$ are exactly those action profiles that are \nasheqs in all local games $\localGame{\egoIndex}$, \ie
    \begin{equation*}
        \nes{\trueGame} = \bigcap_{\egoIndex \in \players} \nes{\localGame{\egoIndex}}.
    \end{equation*}
\end{theorem}
\begin{proof}
    We first show $\nes{\trueGame} \subseteq \bigcap_{\egoIndex \in \players} \nes{\localGame{\egoIndex}}$.
    Due to our assumption, \cref{prop:equilibrium_refinement} applies to every local game.
    Thus, we have $\nes{\trueGame} \subseteq \nes{\localGame{\egoIndex}}$ for all $\egoIndex \in \players$, and therefore $\nes{\trueGame} \subseteq \bigcap_{\egoIndex \in \players} \nes{\localGame{\egoIndex}}$.
    To prove the reverse inclusion, let $\actionProfile \in \bigcap_{\egoIndex \in \players} \nes{\localGame{\egoIndex}}$ be arbitrary.
    Since $\actionProfile$ is a \nasheq of $\localGame{\egoIndex}$ for all $\egoIndex \in \players$, we have $\actionProfile[\obstacleIndex] \in \bestResponse{\preference[\obstacleIndex]^{\egoIndex}}(\actionProfile[-\obstacleIndex])$ for all $\egoIndex, \obstacleIndex \in \players$.
    By instantiating this statement with $\obstacleIndex = \egoIndex$, we obtain $\actionProfile[\egoIndex] \in \bestResponse{\preference[\egoIndex]^{\egoIndex}}(\actionProfile[-\egoIndex])$ for all $\egoIndex \in \players$.
    Hence, we have $\actionProfile[\egoIndex] \in \bestResponse{\preference[\egoIndex]}(\actionProfile[-\egoIndex])$ for all $\egoIndex \in \players$ as $\preference[\egoIndex]^{\egoIndex} = \preference[\egoIndex]$.
    This confirms that $\actionProfile \in \nes{\trueGame}$.
\end{proof}

In words, an action profile is a true equilibrium if and only if every agent certifies it as an equilibrium of its own local game.
Each agent can certify its own best response exactly because it knows its true preferences, whereas the conservative beliefs guarantee that no true equilibrium is discared by any agent.

\subsubsection{Algorithm}
\label{subsubsec:belief_play_algorithm}

\begin{algorithm}[tb]
    \caption{Belief Play for Agent $\egoIndex \in \players$}
    \label{alg:belief_play}
    \begin{algorithmic}[1]
        \Require Local game $\localGame{\egoIndex}$, social metric $\socialMetric$
        \Ensure Socially optimal \nasheq of the true game $\trueGame$

        \Statex

        \LineComment{Compute and propose \nasheqs of the local game}
        \State $\mathit{solver} \gets \Call{NESolver}{\localGame{\egoIndex}}$
        \State $\set{N}_{\egoIndex} \gets \emptyset$
        \While{$\neg \mathit{solver}.\Call{Done}{{}}$}
        \label{alg:belief_play:start_while}
        \State $\actionProfile \gets \mathit{solver}.\Call{NextNE}{{}}$
        \State $\set{N}_\egoIndex \gets \set{N}_\egoIndex \cup \{ \actionProfile \}$
        \State \Call{Broadcast}{$\actionProfile$}
        \EndWhile
        \label{alg:belief_play:end_while}

        \Statex

        \LineComment{Receive candidates proposed by others in the meantime}
        \For{$\obstacleIndex \in \otherPlayers$}
        \State $\set{N}_{\obstacleIndex} \gets \Call{RecvFrom}{\obstacleIndex}$
        \EndFor

        \Statex

        \State $\set{N} \gets \bigcap_{\obstacleIndex \in \players} \set{N}_{\obstacleIndex}$ \Comment{Equilibria proposed by everyone}
        \label{alg:belief_play:candidate_intersection}
        \State \Return $\argmax_{\actionProfile \in \set{N}}(\socialMetric(\actionProfile))$
    \end{algorithmic}
\end{algorithm}

\Cref{alg:belief_play} outlines the belief play protocol from the perspective of the agent $\egoIndex \in \players$.
The agent first computes the \nasheqs of its local game, collects them in $\set{N}_{\egoIndex}$, and broadcasts them to all other agents.
Here, we assume that \textsc{NESolver} incrementally produces the \nasheqs of the provided game.
Next, the agent receives the candidate equilibria proposed by the others.
According to \cref{thm:global_nes_eq_local_nes}, each action profile that is proposed by every agent is a \nasheq of the true game.
Thus, we compute the intersection of the candidates in line~\ref{alg:belief_play:candidate_intersection} and select the one that maximizes the social metric $\socialMetric$.
Since all agents follow the protocol and receive the same proposals, they compute the same intersection and, with the fixed tie-breaking rule of \cref{prob:first}, they commit to the same \nasheq.

\subsubsection{Anytime Legality}
\label{subsubsec:belief_play_anytime_legality}

We can safely interrupt the main loop of \cref{alg:belief_play} in lines~\ref{alg:belief_play:start_while}--\ref{alg:belief_play:end_while} before all \nasheqs have been found.
Let $\mathcal{N}_{\egoIndex} \subseteq \nes{\localGame{\egoIndex}}$ denote the set of \nasheqs proposed by agent $\egoIndex \in \players$ before interruption, and let $\mathcal{N}$ be their intersection as computed in line~\ref{alg:belief_play:candidate_intersection}.
\Cref{thm:global_nes_eq_local_nes} guarantees that $\mathcal{N} \subseteq \nes{\trueGame}$.
Thus, if $\mathcal{N}$ is not empty, the agents commit to the socially best \nasheq among the currently known true equilibria.
Since~$\mathcal{N}$ can only grow as the enumeration proceeds, the social value of the committed equilibrium is nondecreasing and attains the optimum once all local equilibria have been enumerated.
Moreover, if every belief refines the corresponding legal preference, as is the case for the default belief described above, we have $\bigcup_{\egoIndex \in \players} \mathcal{N}_{\egoIndex} \subseteq \nes{\legalGame}$ due to \cref{prop:equilibrium_refinement}.
Thus, even if no true \nasheq has been found yet, \ie $\mathcal{N} = \emptyset$, the agents can fall back to any proposed equilibrium to at least ensure legal behavior.
However, such a fallback is not incentive-stable \wrt the true preferences since it is not a true \nasheq.

\subsubsection{Complexity and Communication}
\label{subsubsec:belief_play_complexity_and_communication}

Belief play is agnostic of the solver used for the local game.
A straightforward option is brute-force enumeration of the \nasheqs of the local game $\localGame{\egoIndex}$, which requires $\bigO(\abs{\actionSpace} n)$ best-response queries for $n := \abs{\players}$ agents, where the size~$\abs{\actionSpace}=\Pi_{\egoIndex\in \players}\abs{\actionSpace[\egoIndex]}$ of the joint action space grows exponentially in~$n$.
The computation of~$\nes{\localGame{\egoIndex}}$ can become tractable when the game has special structure. 
For example, a posetal game can be shown to be a potential game under certain conditions on the metrics and preferences~\cite[Theorem~26]{zanardiPosetalGamesEfficiency2022}.
Every \nasheq found for the local games needs to be communicated.
Thus, if the beliefs refine the legal preferences, the number of required messages is bounded by $n \abs{\nes{\legalGame}}$ as $\nes{\localGame{\egoIndex}} \subseteq \nes{\legalGame}$ by \cref{prop:equilibrium_refinement}.
Observe that finer preference beliefs lead to fewer equilibrium candidates according to \cref{prop:equilibrium_refinement} and thus can only reduce the communication effort.

\subsection{Distributed Iterative Best-Response}
\label{subsec:distributed_ibr}

For our second protocol, we drop the assumption that every agent knows the action space of the others.
Thus, the agent $\egoIndex \in \players$ only knows its own action space $\actionSpace[\egoIndex]$ and true preference $\preference[\egoIndex]$.
The protocol performs distributed best-response iteration, which yields a \nasheq of the true game once it reaches a fixed point.
The agents first assemble an initial action profile, to which each agent contributes its own action, and then forward the current profile from agent to agent.
Upon receipt, an agent replaces its own action with a best response under its true preference and forwards the updated profile.
Once no agent changes its action, every agent plays a best response, and the profile is a \nasheq of~$\trueGame$.
Existing algorithms for distributed \ibr, \eg, \cite{durandDistributedBestResponse2019}, which are defined for classical games, can be applied to posetal games as the notion of a best response transfers (see \cref{subsubsec:posetal_games}).
If time permits, we can restart the iteration from different action profiles to obtain multiple \nasheqs and select the one that maximizes the social metric $\socialMetric$.
Note that the protocol can only guarantee optimality of the solution after exhaustive exploration of the action space.

\begin{remark}[Convergence]
    \label{rem:distributed_ibr_convergence}
    If $\trueGame$ is a potential game, best-response iteration, and thus our distributed protocol, converges according to~\cite[Theorem~26]{zanardiPosetalGamesEfficiency2022}.
    Otherwise, convergence is not guaranteed.
    If the protocol converges, it is typically computationally cheaper than belief play, as the agents only need to answer best-response queries rather than compute the \nasheqs of a local game.
    The improved efficiency comes at a cost:
    Unlike belief play, distributed \ibr is not guaranteed to find the optimal equilibrium (see \cref{tab:solution_protocols}) and does not provide a legal equilibrium as a fallback in case of non-convergence or interruption.
    Nevertheless, our empirical results in \cref{subsec:eval_protocols} suggest that, for the \colregs case study, the protocol converges within about three best-response rounds, and that restarts recover the social value of the optimal solution in all tested scenarios, although up to 20 initializations are needed.
\end{remark}

\section{Case Study: \colregs-Compliant Planning}
\label{sec:case_study}

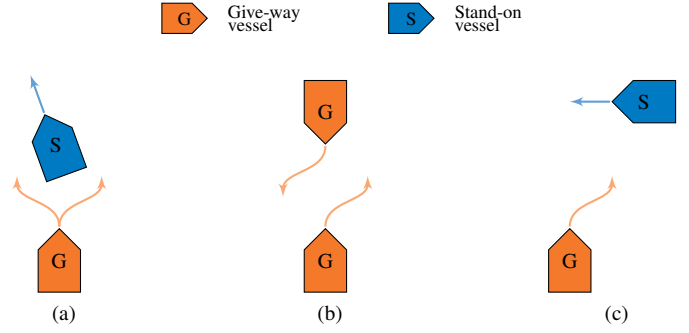
\begin{figure}
    \centering
    \subfloat{%
        \begin{tikzpicture}[scale=0.2]
    \tikzset{legend/.style={font=\scriptsize,align=left,execute at begin node=\setlength{\baselineskip}{2pt}}}

    \vessel{(0, 0)}{0}{TUMOrange}{\scriptsize G}{give_way}
    \vessel{(15, 0)}{0}{TUMBlue}{\scriptsize S}{stand_on}

    \node[right=0.15cm of front_give_way,legend] {Give-wa\smash{y} \\ vessel};
    \node[right=0.15cm of front_stand_on,legend] {Stand-on \\ vessel};
\end{tikzpicture}%
    }
    \\[0.5em]
    \setcounter{subfigure}{0}
    \subfloat[]{%
        \begin{tikzpicture}[scale=0.28]
    \tikzset{trajectory/.style={thick,-latex'}}

    \begin{scope}[local bounding box=overtaking]
        \vessel{(0, 0)}{90}{TUMOrange}{\footnotesize G}{overtaking}
        \vessel{(0, 5.5)}{110}{TUMBlue}{\footnotesize S}{overtaken}
        \draw[trajectory, TUMBlue!60] (front_overtaken) +({0.1*cos(110)}, {0.1*sin(110)}) -- +({2*cos(110)}, {2*sin(110)});
        \draw (front_overtaking) +(0, 0.1)
            edge[trajectory, TUMOrange!60, in=-90, out=90] +(2, 2.5)
            edge[trajectory, TUMOrange!60, in=-90, out=90] +(-2, 2.5);
    \end{scope}
\end{tikzpicture}%
    }
    \hfill
    \subfloat[]{%
        \begin{tikzpicture}[scale=0.28]
    \tikzset{trajectory/.style={thick,-latex'}}

    \begin{scope}[local bounding box=head_on, shift={(23, 0)}]
        \vessel{(0, 0)}{90}{TUMOrange}{\footnotesize G}{bottom}
        \vessel{(0, 8)}{-90}{TUMOrange}{\footnotesize G}{top}
        \draw (front_bottom) +(0, 0.1) edge[trajectory, TUMOrange!60, in=-90, out=90] +(2, 2.5);
        \draw (front_top) +(0, -0.1) edge[trajectory, TUMOrange!60, in=90, out=-90] +(-2, -2.5);
    \end{scope}
\end{tikzpicture}%
    }
    \hfill
    \subfloat[]{%
        \begin{tikzpicture}[scale=0.28]
    \tikzset{trajectory/.style={thick,-latex'}}

    \begin{scope}[local bounding box=crossing, shift={(33, 0)}]
        \vessel{(0, 0)}{90}{TUMOrange}{\footnotesize G}{left}
        \vessel{(4, 8)}{180}{TUMBlue}{\footnotesize S}{right}
        \draw (front_left) +(0, 0.1) edge[trajectory, TUMOrange!60, in=-90, out=90] +(2, 2.5);
        \draw[trajectory, TUMBlue!60] (front_right) +(-0.1, 0) -- +(-2, 0);
    \end{scope}
\end{tikzpicture}%
    }
    \caption{%
        \colregs encounter situations (adapted from \cite[Fig.~2]{krasowskiProvableTrafficRule2024}).
        (a) In an overtaking situation, the overtaking vessel evades to either direction, while the overtaken vessel stands on.
        (b) In a head-on situation, both vessels evade to starboard.
        (c) In a crossing situation, the vessel coming from the left evades to starboard, while the vessel coming from the right stands on.
    }
    \label{fig:colregs_encounters}
\end{figure}

In this section, we apply the modeling framework developed in \cref{sec:guideline_games} to \colregs-compliant multi-vessel motion planning to demonstrate its practical value.
The \colregs~\cite{internationalmaritimeorganizationimoCOLREGsConventionInternational1972} govern the collision avoidance of vessels in the open sea.
Our case study focuses on the steering and sailing rules that apply under good visibility conditions.
These rules assign one of two roles to each vessel in an encounter involving risk of collision:
\emph{Give-way} vessels are expected to observably alter their course to eliminate collision risk, whereas \emph{stand-on} vessels have to keep their original course and speed until the situation is resolved.
\Cref{fig:colregs_encounters} shows the assignment of roles for the situations defined in Rules~13 to 15.
The ultimate concern of the \colregs is avoiding collisions, which takes precedence over strict compliance with the encounter rules according to Rule~2(b):
\begin{quotation}
    \itshape
    In construing and complying with these Rules due regard shall be had to all dangers \textelp{} and to any special circumstances \textelp{}, which may make a departure from these Rules necessary to avoid immediate danger.
    \hfill
    \textnormal{\cite[Rule~2(b)]{internationalmaritimeorganizationimoCOLREGsConventionInternational1972}}
\end{quotation}
This is reinforced in Rule~17(b), which requires the stand-on vessel to maneuver if the give-way vessel does not take sufficient action to avoid a collision.
We emphasize that the \colregs only provide pairwise rules and do not explicitly address encounters involving more than two vessels.
Our framework is the first that systematically lifts the pairwise \colregs rules to encounters involving an arbitrary number of vessels.

In our case study, a scenario $\scenario$ comprises a set of vessels $\players$ with the set of feasible trajectories as their action spaces $\{ \actionSpace[\egoIndex] \}_{\egoIndex \in \players}$.
We begin by defining the legal metrics in \cref{subsec:colregs_metrics} and the \colregs guideline in \cref{subsec:colregs_guideline}.
In \cref{subsec:colregs_games}, we discuss the legal and true games derived from an example scenario.
Finally, \cref{subsec:colregs_regulatory_change} illustrates how our framework can accommodate potential future changes to the \colregs.
\subsection{Metrics Required by \colregs}
\label{subsec:colregs_metrics}
To define the set of legal metrics, we leverage the temporal logic formalization of the encounter rules from \cite{krasowskiTemporalLogicFormalization2021,krasowskiProvableTrafficRule2024}, focusing on the rules $R_3$ to $R_6$.
For ease of presentation, we assume that the give-way rules $R_3$ to $R_5$ are combined into a single specification.
Thus, we have two generic specifications: $\varphi^\mathtt{gw}$ for give-way obligations and $\varphi^\mathtt{so}$ for stand-on obligations.
Moreover, we can specify collision avoidance as
\begin{equation*}
    \varphi^\mathtt{coll} := \globally \neg \mathrm{collides}(\egoIndex, \obstacleIndex),
\end{equation*}
where $\egoIndex, \obstacleIndex$ are placeholders for the ego and contact vessel.

We apply the procedure from \cref{subsec:legally_required_metrics} for the vessel $\egoIndex \in \players$ in the scenario $\scenario$.
For $\varphi^\mathtt{gw}$ and $\varphi^\mathtt{so}$, using the robustness definitions from \cite[Sec.~5.2 and~5.3]{mullerFalsificationdrivenReinforcementLearning2026} yields the metrics $\giveWayMetric{\egoIndex}{\obstacleIndex}$ and $\standOnMetric{\egoIndex}{\obstacleIndex}$ for each contact vessel $\obstacleIndex \in \otherPlayers$.
For $\varphi^\mathtt{coll}$, we define the robustness of the $\mathrm{collides}$ predicate at time $t$ for $\actionProfile \in \actionSpace$ as
\begin{equation*}
    \robustness{\mathrm{collides}(\egoIndex, \obstacleIndex)}{\actionProfile}[t] :=
    \begin{cases}
        r_{\egoIndex, \obstacleIndex} - d^t(\actionProfile)                          & \text{if } d^t(\actionProfile) > r_{\egoIndex, \obstacleIndex} \\
        \abs{v_{\egoIndex}^t(\actionProfile) - v_{\obstacleIndex}^t(\actionProfile)} & \text{otherwise}
    \end{cases},
\end{equation*}
with $d^t$ the inter-vessel distance and $v^t$ the vessel velocities at time $t$, and the collision radius $r_{\egoIndex, \obstacleIndex}$.
Applying spatial robustness~\cite[Sec.~2.2]{bartocciSpecificationBasedMonitoringCyberPhysical2018} and \eqref{eq:metric_from_robustness} yields the metric $\collMetric{\egoIndex}{\obstacleIndex}$ for each $\obstacleIndex \in \otherPlayers$, which quantifies collision severity as the maximum relative collision speed.
For completeness, the closed form of the metric is
\begin{equation*}
    \collMetric{\egoIndex}{\obstacleIndex}(\actionProfile) = -\max\limits_{t \,\text{s.t.}\, d^t(\actionProfile) \le r_{\egoIndex, \obstacleIndex}}(\abs{v_{\egoIndex}^t(\actionProfile) - v_{\obstacleIndex}^t(\actionProfile)})
\end{equation*}
with value $0$ if $d^t(\actionProfile) > r_{\egoIndex, \obstacleIndex}$ for all $t$.
Overall, the set of metrics required by the \colregs for the vessel $\egoIndex \in \players$ is
\begin{equation*}
    \colregsMetrics{\egoIndex} := \bigcup_{\obstacleIndex \in \otherPlayers} \{ \collMetric{\egoIndex}{\obstacleIndex}, \giveWayMetric{\egoIndex}{\obstacleIndex}, \standOnMetric{\egoIndex}{\obstacleIndex} \}.
\end{equation*}

\subsection{\colregs Guideline}
\label{subsec:colregs_guideline}

Before defining the \colregs guideline, we need to organize the metrics using labels.
For every $\egoIndex, \obstacleIndex \in \players$, $\collMetric{\egoIndex}{\obstacleIndex}$ is labeled with $\collTag$, $\giveWayMetric{\egoIndex}{\obstacleIndex}$ with $\giveWayTag$, and $\standOnMetric{\egoIndex}{\obstacleIndex}$ with $\standOnTag$.
Moreover, we include labels describing properties of the contact vessel $\obstacleIndex$, \eg, its classification according to \cite[Rule~3]{internationalmaritimeorganizationimoCOLREGsConventionInternational1972}.
The possible classes include $\nucTag$ for vessels not under command, $\hamperedTag$ for vessels restricted in their ability to maneuver, $\fishingTag$ for vessels engaged in fishing, $\sailingTag$ for vessels under sail, and $\powerDrivenTag$ for power-driven vessels not falling into any of the other classes.
Thus, the concern $\{ \giveWayTag \}$ relates to give-way obligations \wrt any other vessel, while the more specific concern $\{ \giveWayTag, \nucTag \}$ only cares about give-way obligations \wrt vessels not under command.
The metric $\giveWayMetric{\egoIndex}{\obstacleIndex}$ addresses both concerns provided that $\obstacleIndex$ is not under command.

\begin{figure}
    \centering
    \begin{tikzcd}[row sep=tiny,column sep=-25pt]
    & \{ \collTag \} \ar[dl,dash] \ar[dr,dash] & \\
    \{ \giveWayTag \} & & \{ \standOnTag \}
\end{tikzcd}%
    \caption{%
        The guideline $\colregsGuideline$ representing the priorities over concerns specified by the \colregs.
    }
    \label{fig:colregs_guideline_base}
\end{figure}
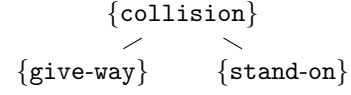

To define the \colregs guideline $\colregsGuideline$, recall from the beginning of the section that the \colregs prioritize collision avoidance over strict compliance with the encounter rules, but leave the priority between different rules unspecified.
Thus, $\colregsGuideline$, which is shown in \cref{fig:colregs_guideline_base}, puts the concern $\{ \collTag \}$ over $\{ \giveWayTag \}$ and $\{ \standOnTag \}$, but leaves the last two incomparable.
Note that the guideline is unambiguous \wrt the legal metrics $\colregsMetrics{\egoIndex}$ for a vessel $\egoIndex \in \players$.

Let us now consider the preference induced by $\colregsGuideline$ over the set of legal metrics $\colregsMetrics{\egoIndex}$ for a vessel $\egoIndex \in \players$ in an arbitrary scenario $\scenario$.
The preference does not specify a priority between collision metrics $\collMetric{\egoIndex}{\obstacleIndex}$ \wrt different contact vessels $\obstacleIndex \in \otherPlayers$.
Thus, collisions with different contact vessels are incomparable, as the law does not favor a particular vessel.
This applies analogously to give-way and stand-on metrics.

\subsection{Legal and True Game}
\label{subsec:colregs_games}

\begin{figure}
    \centering
    \begin{tikzpicture}[scale=0.35]
    \pgfdeclarelayer{callouts}
    \pgfsetlayers{callouts,main}

    \def\vdist{0.2cm}
    \def\hdist{0.0cm}

    \tikzset{trajectory/.style={very thick,-latex'}}
    \tikzset{metric/.style={font=\footnotesize, inner sep=2pt}}
    \tikzset{preference/.style={draw,cloud callout,thick}}

    \vessel{(0, 0)}{90}{TUMBlue}{$V_0$}{1}
    \vessel{(7, 12)}{210}{TUMOrange}{$V_1$}{2}
    \vessel{(-7, 12)}{-30}{TUMGreen}{$V_2$}{3}

    \begin{scope}[shift={(front_1)}, rotate=90]
        \tikzset{every edge/.append style={trajectory,TUMBlue!60,out=0,in=180}}
        \draw (0, 0)
            edge (4, 0)
            edge (4, -3)
            edge (4, 3);
        \draw (2, 0)
            edge (4, -1.5)
            edge (4, 1.5);
    \end{scope}
    \node[metric] (1_m11) at (6.1, 0) {$\collMetric{0}{1}$};
    \node[metric,right=\hdist of 1_m11] (1_m12) {$\collMetric{0}{2}$};
    
    \node[metric,below=\vdist of 1_m11] (1_m22) {$\standOnMetric{0}{1}$};
    \node[metric,left=\hdist of 1_m22] (1_m21) {$\giveWayMetric{0}{1}$};
    \node[metric,below=\vdist of 1_m12] (1_m23) {$\giveWayMetric{0}{2}$};
    \node[metric,right=\hdist of 1_m23] (1_m24) {$\standOnMetric{0}{2}$};
    
    \begin{scope}
        \draw (1_m21) -- (1_m11);
        \draw (1_m22) -- (1_m11);
        \draw (1_m23) -- (1_m11);
        \draw (1_m24) -- (1_m11);

        \draw (1_m21) -- (1_m12);
        \draw (1_m22) -- (1_m12);
        \draw (1_m23) -- (1_m12);
        \draw (1_m24) -- (1_m12);
    \end{scope}
    \begin{pgfonlayer}{callouts}
        \coordinate (mid1) at ($(1_m22)!0.5!(1_m23)$) {};
        \node[
            preference,
            draw=TUMBlue,
            fill=TUMBlue!20,
            aspect=2,
            callout absolute pointer={(1.4,1.2)},
            inner sep=16pt,
        ] at ($(mid1) + (0, 0.6)$) {};
    \end{pgfonlayer}

    \begin{scope}[shift={(front_2)}, rotate=210]
        \tikzset{every edge/.append style={trajectory,TUMOrange!60,out=0,in=180}}
        \draw (0, 0)
            edge (4, 0)
            edge (4, -3)
            edge (4, 3);
        \draw (2, 0)
            edge (4, -1.5)
            edge (4, 1.5);
    \end{scope}
    \begin{scope}
        \node[metric] (2_m11) at (8, 7) {$\collMetric{1}{0}$};
        \node[metric,right=\hdist of 2_m11] (2_m12) {$\collMetric{1}{2}$};
        
        \coordinate (mid2) at ($(2_m11.south east)!0.5!(2_m12.south west)$);
        \node[metric,below=\vdist of mid2] (2_m21) {
            $\giveWayMetric{1}{0} +
            \standOnMetric{1}{0} +
            \giveWayMetric{1}{2} +
            \standOnMetric{1}{2}$
        };
        \node[metric,below=\vdist of 2_m21] (2_m31) {$m^\mathtt{perf}_1$};
        
        \draw (2_m21) -- (2_m11);
        \draw (2_m21) -- (2_m12);
        \draw (2_m31) -- (2_m21);
        \begin{pgfonlayer}{callouts}
            \node[
                preference,
                draw=TUMOrange,
                fill=TUMOrange!20,
                aspect=2,
                callout absolute pointer={(7,10.2)},
                inner sep=18pt,
            ] at ($(2_m21) + (-0.05, 0.2)$) {};
        \end{pgfonlayer}
    \end{scope}

    \begin{scope}[shift={(front_3)}, rotate=-30]
        \tikzset{every edge/.append style={trajectory,TUMGreen!60,out=0,in=180}}
        \draw (0, 0)
            edge (4, 0)
            edge (4, -3)
            edge (4, 3);
        \draw (2, 0)
            edge (4, -1.5)
            edge (4, 1.5);
    \end{scope}
    \begin{scope}
        \node[metric] (3_m12) at (-5, 4) {$\collMetric{2}{0}$};
        \node[metric,left=\hdist of 3_m12] (3_m11) {$\collMetric{2}{1}$};
        
        \node[metric,below=\vdist of 3_m11] (3_m22) {$\giveWayMetric{2}{0}$};
        \node[metric,below=\vdist of 3_m12] (3_m23) {$\giveWayMetric{2}{1}$};
        \node[metric,below=\vdist of 3_m22] (3_m21) {$\standOnMetric{2}{0}$};
        \node[metric,below=\vdist of 3_m23] (3_m24) {$\standOnMetric{2}{1}$};
        
        \coordinate (mid3) at ($(3_m21.south east)!0.5!(3_m24.south west)$);
        \node[metric,below=\vdist of mid3] (3_m31) {$m^\mathtt{perf}_2$};
        
        \draw (3_m22) -- (3_m11);
        \draw (3_m23) -- (3_m11);
        \draw (3_m22) -- (3_m12);
        \draw (3_m23) -- (3_m12);

        \draw (3_m21) -- (3_m22);
        \draw (3_m24) -- (3_m22);
        \draw (3_m21) -- (3_m23);
        \draw (3_m24) -- (3_m23);

        \draw (3_m31) -- (3_m21);
        \draw (3_m31) -- (3_m24);
        
        \coordinate (mid3_2) at ($(3_m22)!0.5!(3_m23)$);
        \begin{pgfonlayer}{callouts}
            \node[
                preference,
                cloud puffs=12,
                draw=TUMGreen,
                fill=TUMGreen!20,
                aspect=0.5,
                callout absolute pointer={(-6.6,10)},
                inner sep=17pt,
            ] at ($(mid3_2) + (0, -0.4)$) {};
        \end{pgfonlayer}
    \end{scope}
\end{tikzpicture}
    \caption{%
        Modeling motion planning under \colregs in an encounter of three vessels as a posetal game.
        We visualize the continuous action space of each vessel using a finite number of trajectories.
        In the legal game, every vessel would use its legal preference as shown for $V_0$.
        In the true game, the vessels use refined preferences including their individual objectives and customized implementations of the \colregs as shown for $V_1$ and $V_2$.
    }
    \label{fig:colregs_game}
\end{figure}
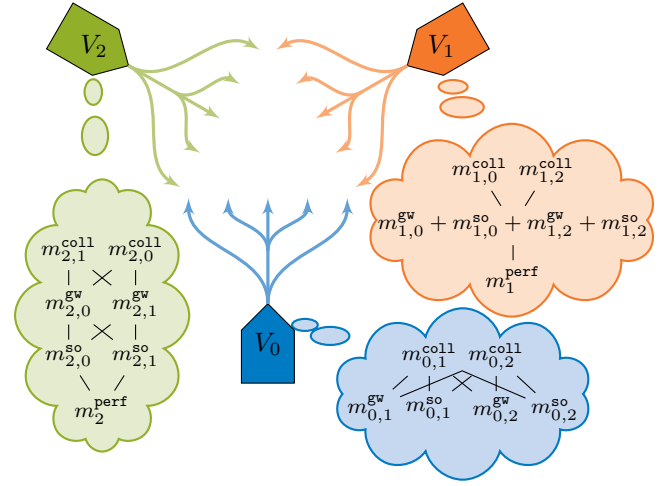

To examine the automatically derived posetal games, we consider a scenario $\scenario$ representing an encounter of three vessels as shown in \cref{fig:colregs_game}.
In the legal game $\legalGame$, every vessel uses its legal preference, which only involves metrics and priorities prescribed by the law, as shown for $V_0$.
In the true game $\trueGame$, the vessels may use any refinement of the legal preference.
We demonstrate that \cref{prop:preference_priority_refinement,prop:preference_aggregation,prop:preference_augmentation} handle many commonly desired refinements.
To incorporate the individual objectives of each vessel, we can use \cref{prop:preference_augmentation} to augment the legal preference with the individual preference over suitable performance metrics.
In \cref{fig:colregs_game}, both $V_1$ and $V_2$ include a performance metric $m^\mathtt{perf}_{\egoIndex}$ using this technique.
\Cref{prop:preference_priority_refinement} allows us to define priorities between metrics that are incomparable according to the law.
Doing so, $V_2$ prioritizes give-way obligations over stand-on obligations.
Finally, we may aggregate any incomparable metrics into a single objective using \cref{prop:preference_aggregation}, \eg, to improve computational tractability.
This is shown for $V_1$, which aggregates all encounter rule metrics.

\subsection{Regulatory Change}
\label{subsec:colregs_regulatory_change}

We outline how the guideline refinement operations from \cref{subsec:handling_regulatory_change} could be used to accommodate potential future changes to the \colregs.
Although these modifications are hypothetical, they reflect changes that may reasonably be considered by regulators.

Many existing implementations of the \colregs implicitly or explicitly prioritize give-way obligations over stand-on obligations, \eg, \cite{wingqvistCollisionAvoidanceASVs2024,meyerCOLREGCompliantCollisionAvoidance2020,johansenShipCollisionAvoidance2016}.
To make this prioritization a legal requirement, we use \cref{prop:guideline_priority_refinement,cor:minimal_priority_refinement} to add the edge $\{ \standOnTag \} \posetPreceq_{\colregsGuideline} \{ \giveWayTag \}$ to $\colregsGuideline$.
Under the refined guideline, every vessel must prioritize give-way metrics $\giveWayMetric{\egoIndex}{\obstacleIndex}$ over stand-on metrics $\standOnMetric{\egoIndex}{\obstacleIndex}$ in its preference.
For the scenario from \cref{fig:colregs_game}, this makes the preference shown for $V_2$ the new legal preference.

\begin{figure}
    \centering
    \subfloat[]{%
        \label{subfig:adjusted_colregs_guideline:hierarchy}
        \centering
        \begin{tikzcd}[row sep=tiny,column sep=-30pt]
    \{ \nucTag \} \ar[dr,dash] & & \{ \hamperedTag \} \ar[dl,dash] \\
    & \{ \fishingTag \} \ar[d,dash] & \\
    & \{ \sailingTag \} \ar[d,dash] & \\
    & \{ \powerDrivenTag \} &
\end{tikzcd}%
    }
    \\
    \subfloat[]{%
        \label{subfig:adjusted_colregs_guideline:result}
        \centering
        \begin{tikzpicture}
    \def\vdist{0.2cm}
    \def\hdist{0.1cm}

    \tikzset{prio/.style={TUMBlue}}

    \node (collision) at (0, 0) {$\{ \collTag \}$};
    
    \node[below left=\vdist and \hdist / 2 of collision.south] (gw_h) {$\{ \metricTag{gw}, \metricTag{h} \}$};
    \node[below right=\vdist and \hdist / 2 of collision.south] (so_nuc) {$\{ \metricTag{so}, \metricTag{nuc} \}$};
    \node[left=\hdist of gw_h] (gw_nuc) {$\{ \metricTag{gw}, \metricTag{nuc} \}$};
    \node[right=\hdist of so_nuc] (so_h) {$\{ \metricTag{so}, \metricTag{h} \}$};

    \node[below left=\vdist and \hdist / 2 of gw_h.south] (gw_f) {$\{ \metricTag{gw}, \metricTag{f} \}$};
    \node[below right=\vdist and \hdist / 2 of so_nuc.south] (so_f) {$\{ \metricTag{so}, \metricTag{f} \}$};

    \node[below=\vdist of gw_f.south] (gw_s) {$\{ \metricTag{gw}, \metricTag{s} \}$};
    \node[below=\vdist of so_f.south] (so_s) {$\{ \metricTag{so}, \metricTag{s} \}$};

    \node[below=\vdist of gw_s.south] (gw_p) {$\{ \metricTag{gw}, \metricTag{p} \}$};
    \node[below=\vdist of so_s.south] (so_p) {$\{ \metricTag{so}, \metricTag{p} \}$};

    \draw (collision) -- (gw_nuc);
    \draw (collision) -- (gw_h);
    \draw (collision) -- (so_nuc);
    \draw (collision) -- (so_h);

    \draw (gw_nuc) -- (gw_f);
    \draw[prio] (gw_nuc) -- (so_f);
    \draw (gw_h) -- (gw_f);
    \draw[prio] (gw_h) -- (so_f);
    \draw[prio] (so_nuc) -- (gw_f);
    \draw (so_nuc) -- (so_f);
    \draw[prio] (so_h) -- (gw_f);
    \draw (so_h) -- (so_f);

    \draw (gw_f) -- (gw_s);
    \draw[prio] (gw_f) -- (so_s);
    \draw[prio] (so_f) -- (gw_s);
    \draw (so_f) -- (so_s);

    \draw (gw_s) -- (gw_p);
    \draw[prio] (gw_s) -- (so_p);
    \draw[prio] (so_s) -- (gw_p);
    \draw (so_s) -- (so_p);
\end{tikzpicture}%
    }
    \hfil
    \caption{%
        Modifying $\colregsGuideline$ to prioritize the encounter rule metrics according to the vessel class hierarchy.
        (a) The guideline $\guideline^\mathtt{R18}$ encoding the hierarchy.
        (b) The refined guideline after injection and priority refinement.
        The edges added via priority refinement are highlighted.
        We abbreviate some labels to keep the figure compact, \eg, $\metricTag{gw}$ is short for $\giveWayTag$.
    }
    \label{fig:adjusted_colregs_guideline}
\end{figure}
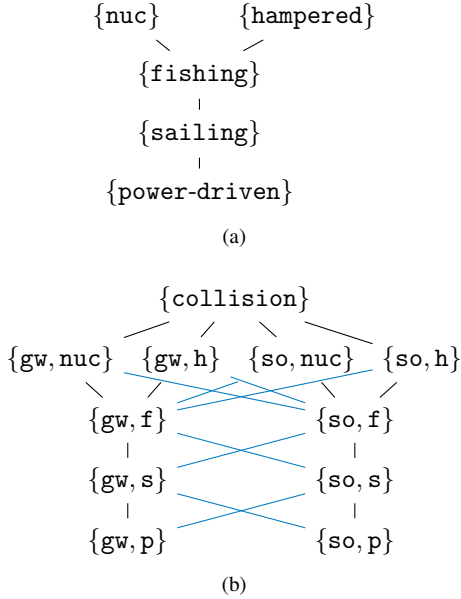

Let us now consider the responsibilities between vessels from \cite[Rule~18]{internationalmaritimeorganizationimoCOLREGsConventionInternational1972}.
The rule establishes a hierarchy of vessel classes and requires vessels to give-way to any vessel higher up in the hierarchy, which in turn stands on.
For multi-vessel encounters, a natural extension is to require that encounter obligations are fulfilled \wrt the highest-priority vessel first.
To incorporate this requirement, we first define the guideline $\guideline^\mathtt{R18}$ as shown in \cref{subfig:adjusted_colregs_guideline:hierarchy} to capture the hierarchy.
Using \cref{prop:guideline_injection}, we split both the $\{ \giveWayTag \}$ and $\{ \standOnTag \}$ concerns of $\colregsGuideline$ according to $\guideline^\mathtt{R18}$, \ie we compute $(\colregsGuideline \inject{\{ \giveWayTag \}} \guideline^\mathtt{R18}) \inject{\{ \standOnTag \}} \guideline^\mathtt{R18}$.
Adding the cross-priorities via \cref{prop:guideline_priority_refinement} yields the guideline shown in \cref{subfig:adjusted_colregs_guideline:result}, which enforces the desired prioritization.

Finally, \cite[Rule~1(b)]{internationalmaritimeorganizationimoCOLREGsConventionInternational1972} lets local authorities issue special rules that take precedence over the general regulations.
An example are the U.S. Inland Navigation Rules, which contain special provisions for the Great Lakes~\cite{unitedstatescoastguardNavigationRulesRegulations2024}.
Modeling the priorities of such a regime over its own concerns as a guideline $\guideline^\mathtt{local}$, we use \cref{prop:guideline_augmentation} to augment $\guideline^\mathtt{local}$ with $\colregsGuideline$.
In the resulting guideline, the local rules always take precedence over the general \colregs rules.

\tikzexternalenable
\section{Evaluation}
\label{sec:evaluation}

\providecommand{\changed}[1]{\textcolor{TUMOrange}{#1}}

\begin{figure*}[t]
  \centering
  \begin{tikzpicture}	
	\begin{axis}[
        hide axis,
        xmin=0, xmax=1,
        ymin=0, ymax=1,
        legend columns=-1,
        legend style={draw=none, fill=none, /tikz/every even column/.append style={column sep=1em}, /tikz/nodes={font=\scriptsize}},
    ]
        \addlegendimage{line legend, ultra thick, color=TUMBlue}
        \addlegendentry{Alfa}

        \addlegendimage{line legend, ultra thick, color=TUMOrange}
        \addlegendentry{Bravo}

        \addlegendimage{line legend, ultra thick, color=TUMGreen}
        \addlegendentry{Charlie}

        \addlegendimage{line legend, ultra thick, color=TUMPurple}
        \addlegendentry{Delta}

        \addlegendimage{initial state, TUMDarkGray}
        \addlegendentry{Initial state}

        \addlegendimage{goal, TUMDarkGray}
        \addlegendentry{Goal}
    \end{axis}
\end{tikzpicture} \\[0.2cm]
  \input{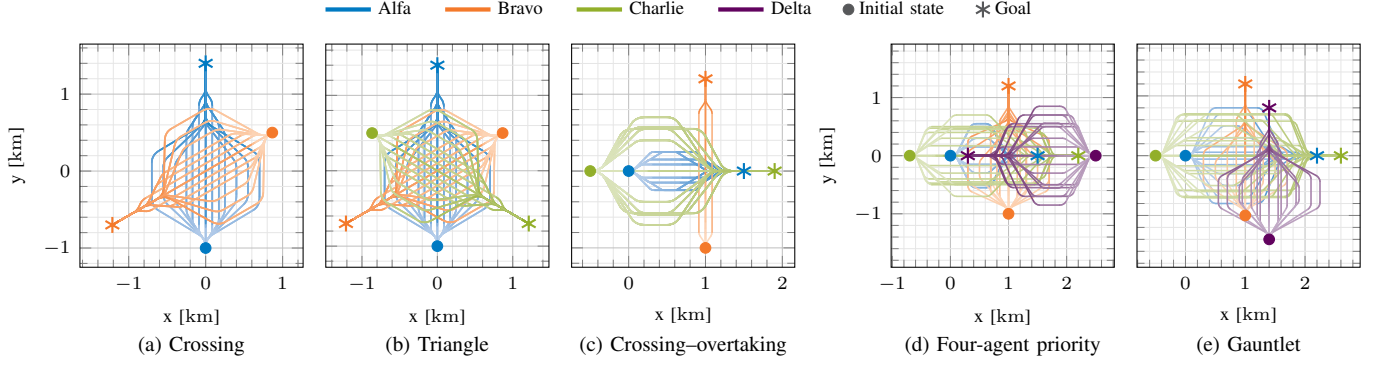}%
{%
\pgfplotstableread[col sep=comma]{data/trajectories/crossing/alfa.csv}{\trajCrossingAlfa}%
\pgfplotstableread[col sep=comma]{data/trajectories/crossing/bravo.csv}{\trajCrossingBravo}%
\pgfplotstableread[col sep=comma]{data/trajectories/triangle/alfa.csv}{\trajTriangleAlfa}%
\pgfplotstableread[col sep=comma]{data/trajectories/triangle/bravo.csv}{\trajTriangleBravo}%
\pgfplotstableread[col sep=comma]{data/trajectories/triangle/charlie.csv}{\trajTriangleCharlie}%
\pgfplotstableread[col sep=comma]{data/trajectories/crossing_overtaking/alfa.csv}{\trajCrossingOvertakingAlfa}%
\pgfplotstableread[col sep=comma]{data/trajectories/crossing_overtaking/bravo.csv}{\trajCrossingOvertakingBravo}%
\pgfplotstableread[col sep=comma]{data/trajectories/crossing_overtaking/charlie.csv}{\trajCrossingOvertakingCharlie}%
\pgfplotstableread[col sep=comma]{data/trajectories/four_agent_prio/alfa.csv}{\trajFourAgentPrioAlfa}%
\pgfplotstableread[col sep=comma]{data/trajectories/four_agent_prio/bravo.csv}{\trajFourAgentPrioBravo}%
\pgfplotstableread[col sep=comma]{data/trajectories/four_agent_prio/charlie.csv}{\trajFourAgentPrioCharlie}%
\pgfplotstableread[col sep=comma]{data/trajectories/four_agent_prio/delta.csv}{\trajFourAgentPrioDelta}%
\pgfplotstableread[col sep=comma]{data/trajectories/gauntlet/alfa.csv}{\trajGauntletAlfa}%
\pgfplotstableread[col sep=comma]{data/trajectories/gauntlet/bravo.csv}{\trajGauntletBravo}%
\pgfplotstableread[col sep=comma]{data/trajectories/gauntlet/charlie.csv}{\trajGauntletCharlie}%
\pgfplotstableread[col sep=comma]{data/trajectories/gauntlet/delta.csv}{\trajGauntletDelta}%

\pgfplotsset{action_space/.style={
    width=0.25\linewidth,
    height=0.25\linewidth,
    axis equal,
    title style={anchor=north, font=\footnotesize},
    grid style={line width=.1pt, draw=gray!20},
    major grid style={line width=.2pt,draw=gray!50},
    grid=both,
    label style={font=\scriptsize},
    tick label style={font=\scriptsize},
    minor tick num=4,
    xlabel={x [\si{\kilo \meter}]},
    xtick={-2, -1, ..., 3},
    ylabel={y [\si{\kilo \meter}]},
    ytick={-2, -1, ..., 3},
    cycle list name=trajectory,
}}

\begin{tikzpicture}
    \begin{groupplot}[
        group style={
            group size=3 by 1,
            x descriptions at=edge bottom,
            y descriptions at=edge left,
            horizontal sep=0.3cm,
            vertical sep=0.75cm,
        },
        action_space,
        ymin=-1.25, ymax=1.65,
    ]
        \nextgroupplot[title={(a) Crossing}, title style={at={(0.5,-1cm)}}, point meta min=0, point meta max=68]
        \addplot+[forget plot, trajectory, line width=0.6pt, all trajs, gradient] table[x expr=\thisrow{x} / 1000, y expr=\thisrow{y} / 1000] {\trajCrossingAlfa};
        \addplot+[initial state, forget plot] coordinates {(0.0, -1.0000)};
        \addplot+[goal] coordinates {(0.0, 1.4000)};
        \addplot+[forget plot, trajectory, line width=0.6pt, all trajs, gradient] table[x expr=\thisrow{x} / 1000, y expr=\thisrow{y} / 1000] {\trajCrossingBravo};
        \addplot+[initial state, forget plot] coordinates {(0.8660, 0.5000)};
        \addplot+[goal] coordinates {(-1.2124, -0.7000)};

        \nextgroupplot[title={(b) Triangle}, title style={at={(0.5,-1cm)}}, point meta min=0, point meta max=68]
        \addplot+[forget plot, trajectory, line width=0.6pt, all trajs, gradient] table[x expr=\thisrow{x} / 1000, y expr=\thisrow{y} / 1000] {\trajTriangleAlfa};
        \addplot+[initial state, forget plot] coordinates {(0.0, -1.0000)};
        \addplot+[goal] coordinates {(0.0, 1.4000)};
        \addplot+[forget plot, trajectory, line width=0.6pt, all trajs, gradient] table[x expr=\thisrow{x} / 1000, y expr=\thisrow{y} / 1000] {\trajTriangleBravo};
        \addplot+[initial state, forget plot] coordinates {(0.8660, 0.5000)};
        \addplot+[goal] coordinates {(-1.2124, -0.7000)};
        \addplot+[forget plot, trajectory, line width=0.6pt, all trajs, gradient] table[x expr=\thisrow{x} / 1000, y expr=\thisrow{y} / 1000] {\trajTriangleCharlie};
        \addplot+[initial state, forget plot] coordinates {(-0.8660, 0.5000)};
        \addplot+[goal] coordinates {(1.2124, -0.7000)};

        \nextgroupplot[title={(c) Crossing--overtaking}, title style={at={(0.5,-1cm)}}, point meta min=0, point meta max=94]
        \addplot+[forget plot, trajectory, line width=0.6pt, all trajs, gradient] table[x expr=\thisrow{x} / 1000, y expr=\thisrow{y} / 1000] {\trajCrossingOvertakingAlfa};
        \addplot+[initial state, forget plot] coordinates {(0.0, 0.0)};
        \addplot+[goal] coordinates {(1.5000, 0.0)};
        \addplot+[forget plot, trajectory, line width=0.6pt, all trajs, gradient] table[x expr=\thisrow{x} / 1000, y expr=\thisrow{y} / 1000] {\trajCrossingOvertakingBravo};
        \addplot+[initial state, forget plot] coordinates {(1.0000, -1.0000)};
        \addplot+[goal] coordinates {(1.0000, 1.2000)};
        \addplot+[forget plot, trajectory, line width=0.6pt, all trajs, gradient] table[x expr=\thisrow{x} / 1000, y expr=\thisrow{y} / 1000] {\trajCrossingOvertakingCharlie};
        \addplot+[initial state, forget plot] coordinates {(-0.5000, 0.0)};
        \addplot+[goal] coordinates {(1.9000, 0.0)};
    \end{groupplot}
\end{tikzpicture}
\hfil
\begin{tikzpicture}
    \begin{groupplot}[
        group style={
            group size=2 by 1,
            x descriptions at=edge bottom,
            y descriptions at=edge left,
            horizontal sep=0.3cm,
            vertical sep=0.75cm,
        },
        action_space,
        ymin=-1.65, ymax=1.65,
    ]
        \nextgroupplot[title={(d) Four-agent priority}, title style={at={(0.5,-1cm)}}, point meta min=0, point meta max=127]
        \addplot+[forget plot, trajectory, line width=0.6pt, all trajs, gradient] table[x expr=\thisrow{x} / 1000, y expr=\thisrow{y} / 1000] {\trajFourAgentPrioAlfa};
        \addplot+[initial state, forget plot] coordinates {(0.0, 0.0)};
        \addplot+[goal] coordinates {(1.5000, 0.0)};
        \addplot+[forget plot, trajectory, line width=0.6pt, all trajs, gradient] table[x expr=\thisrow{x} / 1000, y expr=\thisrow{y} / 1000] {\trajFourAgentPrioBravo};
        \addplot+[initial state, forget plot] coordinates {(1.0000, -1.0000)};
        \addplot+[goal] coordinates {(1.0000, 1.2000)};
        \addplot+[forget plot, trajectory, line width=0.6pt, all trajs, gradient] table[x expr=\thisrow{x} / 1000, y expr=\thisrow{y} / 1000] {\trajFourAgentPrioCharlie};
        \addplot+[initial state, forget plot] coordinates {(-0.7000, 0.0)};
        \addplot+[goal] coordinates {(2.2000, 0.0)};
        \addplot+[forget plot, trajectory, line width=0.6pt, all trajs, gradient] table[x expr=\thisrow{x} / 1000, y expr=\thisrow{y} / 1000] {\trajFourAgentPrioDelta};
        \addplot+[initial state, forget plot] coordinates {(2.5000, 0.0)};
        \addplot+[goal] coordinates {(0.3000, 0.0)};

        \nextgroupplot[title={(e) Gauntlet}, title style={at={(0.5,-1cm)}}, point meta min=0, point meta max=160]
        \addplot+[forget plot, trajectory, line width=0.6pt, all trajs, gradient] table[x expr=\thisrow{x} / 1000, y expr=\thisrow{y} / 1000] {\trajGauntletAlfa};
        \addplot+[initial state, forget plot] coordinates {(0.0, 0.0)};
        \addplot+[goal] coordinates {(2.2000, 0.0)};
        \addplot+[forget plot, trajectory, line width=0.6pt, all trajs, gradient] table[x expr=\thisrow{x} / 1000, y expr=\thisrow{y} / 1000] {\trajGauntletBravo};
        \addplot+[initial state, forget plot] coordinates {(1.0000, -1.0000)};
        \addplot+[goal] coordinates {(1.0000, 1.2000)};
        \addplot+[forget plot, trajectory, line width=0.6pt, all trajs, gradient] table[x expr=\thisrow{x} / 1000, y expr=\thisrow{y} / 1000] {\trajGauntletCharlie};
        \addplot+[initial state, forget plot] coordinates {(-0.5000, 0.0)};
        \addplot+[goal] coordinates {(2.6000, 0.0)};
        \addplot+[forget plot, trajectory, line width=0.6pt, all trajs, gradient] table[x expr=\thisrow{x} / 1000, y expr=\thisrow{y} / 1000] {\trajGauntletDelta};
        \addplot+[initial state, forget plot] coordinates {(1.4000, -1.4000)};
        \addplot+[goal] coordinates {(1.4000, 0.8000)};
    \end{groupplot}
\end{tikzpicture}%
}%
%
  \vspace{-0.2cm} %
  \caption{%
    Action spaces for all five geometric scenario configurations.
    Each colored curve represents a candidate trajectory for one agent, shaded from a pale tint at $t=0$ to the saturated hue at the horizon.
  }
  \label{fig:action_spaces}
\end{figure*}

We evaluate the framework in one-shot \colregs motion-planning games, in which each agent commits to a complete dynamically feasible trajectory before the encounter begins.
Our goals are threefold.
First, we verify that every rule violation at the committed equilibria can be audited and that defectors can be identified from public information.
Second, we confirm that private and regulatory customizations shrink the set of equilibria as predicted by \cref{thm:global_nes_eq_local_nes} and \cref{prop:equilibrium_refinement}, and we quantify the efficiency cost of private preferences.
Finally, we compare the communication cost and solution quality of the two distributed protocols.
We first describe the common experimental setup, and then address the three goals.

\subsection{Experimental Setup}
\label{subsec:eval_setup}

We implement our evaluation in CommonOcean\footnote{\url{https://commonocean.cps.cit.tum.de}}~\cite{krasowskiCommonOceanComposableBenchmarks2022} to ensure reproducibility of the results.
Our evaluation uses five geometric scenarios, which are shown in \cref{fig:action_spaces}: a two-agent crossing, a three-agent triangle, the crossing--overtaking conflict of \cref{fig:conflict_encounter}, a four-agent priority encounter, and a four-agent gauntlet.
If not specified otherwise, the vessels are general power-driven vessels according to \cite[Rule~3]{internationalmaritimeorganizationimoCOLREGsConventionInternational1972}.
For the crossing and triangle scenarios, we also consider mixed-class variants, where the vessel Alfa is engaged in fishing.
In the four-agent priority encounter, Bravo is a vessel engaged in fishing, and Delta is a sailing vessel.
In the scenarios with vessels of different classes, \ie crossing + fishing, triangle + fishing, and the four-agent priority encounter, the legal guideline is $\colregsGuideline$ refined with the vessel-class hierarchy of \cite[Rule~18]{internationalmaritimeorganizationimoCOLREGsConventionInternational1972} as described in \cref{subsec:colregs_regulatory_change} (\cref{subfig:adjusted_colregs_guideline:result}); we use it for the legal and true games, the social metric, the audits, and the price of preference alike.
In all other scenarios, the legal guideline is $\colregsGuideline$, which induces the same preferences as the refined guideline when all vessels belong to the same class.
The initial speed is \SI{5}{\meter\per\second} for every vessel except Charlie, which travels at \SI{11}{\meter\per\second} in the crossing--overtaking scenario and at \SI{7.5}{\meter\per\second} in the four-agent priority encounter and the gauntlet.

\subsubsection{Vessel Models}

To model the vessel dynamics, we use the yaw-constrained kinematic model from \cite[Sec.~V-B]{krasowskiCommonOceanComposableBenchmarks2022} with acceleration and yaw-rate as bounded inputs.
Hull dimensions and acceleration limits are taken from two CommonOcean vessel types:
type 1 is \SI{175}{\meter} by \SI{25.4}{\meter} with \SI{0.24}{\meter\per\second\squared}, and type 2 is \SI{304.8}{\meter} by \SI{32}{\meter} with \SI{0.0127}{\meter\per\second\squared}.
The maximum yaw rate is set per vessel between \SI{0.03}{\radian\per\second} and \SI{0.06}{\radian\per\second}.

\subsubsection{Action Spaces}

For every scenario, each agent receives a finite set of dynamically feasible trajectories, computed using the above vessel model.
Each agent receives paths that alter course to either side, hold a parallel offset, and return to the original track; the offsets range from \SI{50}{\meter} to \SI{850}{\meter}.
For each path, we consider multiple speed profiles, \eg, constant speed and brake-and-resume.
As a result, the triangle has $41$ actions per agent, the crossing--overtaking game has $36/4/90$ actions for Alfa/Bravo/Charlie, the gauntlet has $104{,}975$ joint profiles, and the largest game, the four-agent priority encounter, has $164{,}934$ joint profiles.
Unless stated otherwise, we enumerate the joint action space exhaustively to determine equilibria.

\subsubsection{True Preferences}

In every scenario, each vessel has a circular goal region with radius \SI{150}{\meter}.
Each vessel $\egoIndex \in \players$ measures its performance using a goal-reaching metric and a path-length metric
\begin{subequations}
  \label{eq:performance_metrics}
  \begin{align}
    m_i^{\mathrm{goal}}(\actionProfile) & := -d_i^{\mathrm{goal}}(\actionProfile), \label{eq:goal_metric}                                                                        \\
    m_i^{\mathrm{len}}(\actionProfile)  & := -\sum_{t}\norm{p_i^{t+\Delta t}(\actionProfile)-p_i^t(\actionProfile)} - v_i^0\, T_i(\actionProfile), \label{eq:path_length_metric}
  \end{align}
\end{subequations}
where $d_i^{\mathrm{goal}}(\actionProfile)$ is the distance from the final position of the vessel to its goal region (zero inside it), $p_i^t(\actionProfile)$ its position at time $t$, $v_i^0$ its initial speed, and $T_i(\actionProfile)$ its travel time until it stops at its goal.
We arrange these into a preference $\preference[\egoIndex]^\mathtt{perf}$ by prioritizing $m_i^{\mathrm{goal}}$ over $m_i^{\mathrm{len}}$ so that reaching the goal dominates not reaching it and, among trajectories that reach, the shorter and faster one is preferred.
Hence, a give-way detour costs path length even if the vessel still reaches its goal, so compliance is never free.
To obtain the true preference of the vessel, we use \cref{prop:preference_augmentation} to place $\preference[\egoIndex]^\mathtt{perf}$ below its legal preference according to \cref{sec:case_study}, analogous to the performance metrics of $V_1$ and $V_2$ in \cref{fig:colregs_game}.
Hence, performance can never compensate for a worse legal outcome.

\subsubsection{Social Metric}

To select the committed solution $\actionProfile^*$ from $\nes{\trueGame}$ according to \eqref{eq:socially_optimal_true_ne}, we use a lexicographic social metric that follows the two levels of the default \colregs guideline.
For an action profile $\actionProfile \in \actionSpace$, let
\begin{align*}
  \socialMetric^{\mathrm{coll}}(\actionProfile) & := \textstyle\sum_{\egoIndex \in \players} \sum_{\obstacleIndex \in \otherPlayers} \collMetric{\egoIndex}{\obstacleIndex}(\actionProfile),                                                                                                               \\
  \socialMetric^{\mathrm{enc}}(\actionProfile)  & := \textstyle\sum_{\egoIndex \in \players} \sum_{\obstacleIndex \in \otherPlayers} \bigl(\giveWayMetric{\egoIndex}{\obstacleIndex}(\actionProfile) / s^{\mathrm{gw}} + \standOnMetric{\egoIndex}{\obstacleIndex}(\actionProfile) / s^{\mathrm{so}}\bigr)
\end{align*}
be the total collision and encounter compliance, where $s^{\mathrm{gw}}$ and $s^{\mathrm{so}}$ are the median magnitudes of the violated give-way and stand-on metrics over all pairs of actions in the game for normalization, and
\begin{equation}
  \label{eq:social_metric}
  \socialMetric(\actionProfile) := \bigl(\socialMetric^{\mathrm{coll}}(\actionProfile),\, \socialMetric^{\mathrm{enc}}(\actionProfile)\bigr),
\end{equation}
with remaining ties broken by a fixed public rule.
Since all legal metrics are capped at zero, both components are non-positive and become $0$ exactly when all vessels comply with the corresponding rules.
Across levels, $\socialMetric$ respects the priority of the guideline: no improvement in encounter compliance can compensate for a worse collision outcome.
Within a level, the guideline does not prioritize among the metrics of different vessels, so $\socialMetric$ weighs them equally.
Formally, $\socialMetric$ corresponds to a refinement of the preference $\inducedPreference[\colregsGuideline]{\bigcup_{\egoIndex \in \players} \colregsMetrics{\egoIndex}}$ that the guideline induces over the union of all legal metrics, obtained by aggregating the incomparable metrics of each level with \cref{prop:preference_aggregation}.
Other refinements could also serve as the social metric.
Observe that $\socialMetric$ is based only on public information, so every agent can evaluate it independently.

\begin{figure*}[t]
  \centering
  \begin{tikzpicture}	
	\begin{axis}[
        hide axis,
        xmin=0, xmax=1,
        ymin=0, ymax=1,
        legend columns=-1,
        legend style={draw=none, fill=none, /tikz/every even column/.append style={column sep=1em}, /tikz/nodes={font=\scriptsize}},
    ]
        \addlegendimage{line legend, ultra thick, color=TUMBlue}
        \addlegendentry{Alfa}

        \addlegendimage{line legend, ultra thick, color=TUMOrange}
        \addlegendentry{Bravo}

        \addlegendimage{line legend, ultra thick, color=TUMGreen}
        \addlegendentry{Charlie}

        \addlegendimage{line legend, ultra thick, color=TUMPurple}
        \addlegendentry{Delta}

        \addlegendimage{initial state, TUMDarkGray}
        \addlegendentry{Initial state}

        \addlegendimage{goal, TUMDarkGray}
        \addlegendentry{Goal}
    \end{axis}
\end{tikzpicture} \\[0.2cm]
  \input{figures/eval_nice/committed.tex}%
  \vspace{-0.2cm} %
  \caption{%
    Socially optimal equilibrium $\actionProfile^*$ \eqref{eq:socially_optimal_true_ne} of the true game.
    The shading of the trajectories encodes time like in \cref{fig:action_spaces}.
  }
  \label{fig:committed_close_goals}
\end{figure*}

\subsection{Auditability}
\label{subsec:eval_auditability}

\Cref{fig:committed_close_goals} shows the committed equilibrium $\actionProfile^*$ for our seven scenarios.
Since not all rules can be satisfied simultaneously in general, the committed equilibrium $\actionProfile^*$ may violate individual rules.
The posetal structure enables us to audit these violations and justify them in terms of conflicting higher-priority requirements.
Furthermore, it enables us to identify defectors, \ie vessels that violate legal requirements even though compliance would have been possible.

\subsubsection{Auditing Violations}
\label{subsubsec:auditing_violations}

For each committed solution, we audit every violated rule against the legal hierarchy, considering the unilateral deviations of the violating vessel.
A violation is \emph{unsatisfiable} within the action space if no such deviation satisfies the rule.
Otherwise, it is \emph{forced} if no deviation can improve the rule without worsening a metric of higher or incomparable priority in the legal preference.
In \cref{fig:committed_close_goals}, we can see that none of the vessels collide, choosing trajectories that maintain the highest priority objective of collision avoidance.
All lower-priority violations at the selected equilibria are forced; none is unsatisfiable within the action space.
Hence, every violation is explained by the legal hierarchy.
We audit all seven scenarios.
The committed equilibria of the crossing and the crossing + fishing scenario satisfy every rule, so they have no violations.
In the triangle, every vessel is the give-way vessel toward one vessel and the stand-on vessel toward another; all three alter course to starboard and thereby violate their stand-on duties, since each of the $2$ deviations that satisfy the stand-on duty of a vessel worsens its give-way metric.
The triangle + fishing scenario has a single forced violation, the stand-on duty of Charlie toward Bravo.
Crossing--overtaking has two forced violations, the stand-on duty of Alfa toward Charlie and the stand-on duty of Charlie toward Alfa: every deviation that satisfies one of them worsens either a collision metric or a give-way metric.
Under the guideline $\colregsGuideline$, give-way and stand-on are incomparable, so the collision concern forces at least one encounter violation, and the social metric selects a trade-off in which the stand-on duties of Alfa and Charlie are violated while every give-way duty is satisfied.
The four-agent priority encounter has five forced violations: the stand-on duty of Alfa toward Charlie, the give-way duties of Charlie toward Alfa and toward the fishing vessel Bravo, and the stand-on duties of Delta toward Alfa and toward Charlie.
The gauntlet has two forced violations, the stand-on duty of Alfa toward Charlie and the give-way duty of Charlie toward Alfa.

\subsubsection{Identifying Defectors}
\label{subsubsec:identifying_defectors}

\begin{figure*}[t]
  \centering
  \begin{tikzpicture}	
	\begin{axis}[
        hide axis,
        xmin=0, xmax=1,
        ymin=0, ymax=1,
        legend columns=-1,
        legend style={draw=none, fill=none, /tikz/every even column/.append style={column sep=1em}, /tikz/nodes={font=\scriptsize}},
    ]
        \addlegendimage{line legend, ultra thick, color=TUMBlue}
        \addlegendentry{Alfa}

        \addlegendimage{line legend, ultra thick, color=TUMOrange}
        \addlegendentry{Bravo}

        \addlegendimage{line legend, ultra thick, color=TUMGreen}
        \addlegendentry{Charlie}

        \addlegendimage{line legend, ultra thick, color=TUMDarkRed}
        \addlegendentry{Defector}

        \addlegendimage{initial state, TUMDarkGray}
        \addlegendentry{Initial state}

        \addlegendimage{goal, TUMDarkGray}
        \addlegendentry{Goal}
    \end{axis}
\end{tikzpicture} \\[0.2cm]
  \input{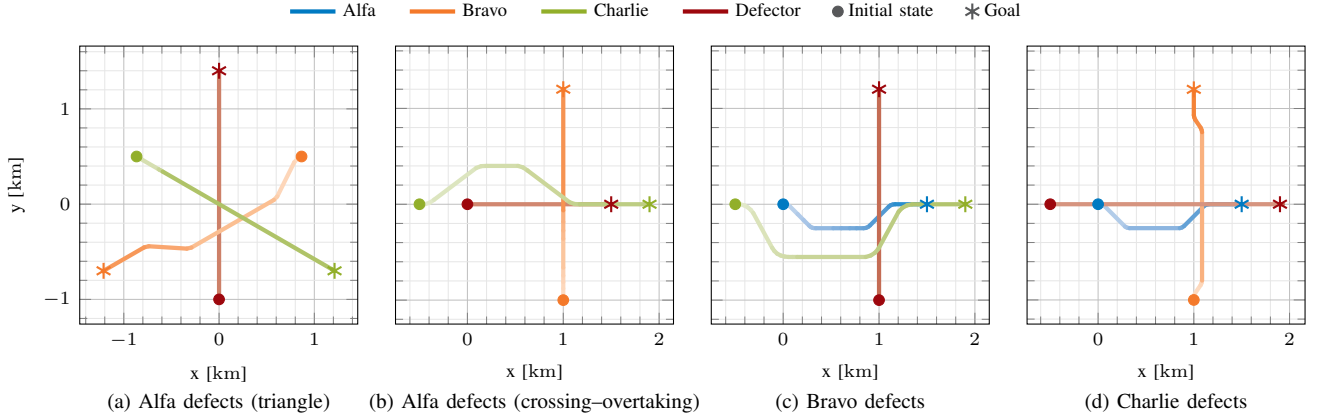}%
\begin{tikzpicture}
    \pgfplotstableread[col sep=comma]{data/trajectories/triangle/alfa.csv}{\trajTriangleAlfa}%
    \pgfplotstableread[col sep=comma]{data/trajectories/triangle/bravo.csv}{\trajTriangleBravo}%
    \pgfplotstableread[col sep=comma]{data/trajectories/triangle/charlie.csv}{\trajTriangleCharlie}%
    \trajset{DefectorTriangleAlfa}{0}%
    \trajset{DefectorTriangleBravo}{3}%
    \trajset{DefectorTriangleCharlie}{40}%

    \pgfplotstableread[col sep=comma]{data/trajectories/crossing_overtaking/alfa.csv}{\trajAlfa}%
    \pgfplotstableread[col sep=comma]{data/trajectories/crossing_overtaking/bravo.csv}{\trajBravo}%
    \pgfplotstableread[col sep=comma]{data/trajectories/crossing_overtaking/charlie.csv}{\trajCharlie}%
    \trajset{DefectorCrossingOvertakingAlfaAlfa}{0}%
    \trajset{DefectorCrossingOvertakingAlfaBravo}{3}%
    \trajset{DefectorCrossingOvertakingAlfaCharlie}{1}%
    \trajset{DefectorCrossingOvertakingBravoAlfa}{25}%
    \trajset{DefectorCrossingOvertakingBravoBravo}{0}%
    \trajset{DefectorCrossingOvertakingBravoCharlie}{84}%
    \trajset{DefectorCrossingOvertakingCharlieAlfa}{8}%
    \trajset{DefectorCrossingOvertakingCharlieBravo}{2}%
    \trajset{DefectorCrossingOvertakingCharlieCharlie}{0}%

    \begin{groupplot}[
        group style={
            group size=4 by 1,
            x descriptions at=edge bottom,
            y descriptions at=edge left,
            horizontal sep=0.5cm,
            vertical sep=0.75cm,
        },
        width=0.29\linewidth,
        height=0.29\linewidth,
        axis equal,
        title style={anchor=north, font=\footnotesize},
        grid style={line width=.1pt, draw=gray!20},
        major grid style={line width=.2pt,draw=gray!50},
        grid=both,
        label style={font=\scriptsize},
        tick label style={font=\scriptsize},
        minor tick num=4,
        xlabel={x [\si{\kilo \meter}]},
        xtick={-2, -1, ..., 2},
        ylabel={y [\si{\kilo \meter}]},
        ytick={-2, -1, ..., 3},
        ymin=-1.25, ymax=1.65,
        cycle list name=trajectory,
    ]
        \nextgroupplot[title={(a) Alfa defects (triangle)}, title style={at={(0.5,-1cm)}}, point meta min=0, point meta max=68]
        \addplot+[forget plot, trajectory, pick trajs=DefectorTriangleAlfa, defector, gradient] table[x expr=\thisrow{x} / 1000, y expr=\thisrow{y} / 1000] {\trajTriangleAlfa};
        \addplot+[initial state, forget plot, defector] coordinates {(0.0, -1.0000)};
        \addplot+[goal, defector] coordinates {(0.0, 1.4000)};
        \addplot+[forget plot, trajectory, pick trajs=DefectorTriangleBravo, gradient] table[x expr=\thisrow{x} / 1000, y expr=\thisrow{y} / 1000] {\trajTriangleBravo};
        \addplot+[initial state, forget plot] coordinates {(0.8660, 0.5000)};
        \addplot+[goal] coordinates {(-1.2124, -0.7000)};
        \addplot+[forget plot, trajectory, pick trajs=DefectorTriangleCharlie, gradient] table[x expr=\thisrow{x} / 1000, y expr=\thisrow{y} / 1000] {\trajTriangleCharlie};
        \addplot+[initial state, forget plot] coordinates {(-0.8660, 0.5000)};
        \addplot+[goal] coordinates {(1.2124, -0.7000)};

        \nextgroupplot[title={(b) Alfa defects (crossing--overtaking)}, title style={at={(0.5,-1cm)}}, point meta min=0, point meta max=94]
        \addplot+[forget plot, trajectory, pick trajs=DefectorCrossingOvertakingAlfaAlfa, defector, gradient] table[x expr=\thisrow{x} / 1000, y expr=\thisrow{y} / 1000] {\trajAlfa};
        \addplot+[initial state, forget plot, defector] coordinates {(0.0, 0.0)};
        \addplot+[goal, defector] coordinates {(1.5000, 0.0)};
        \addplot+[forget plot, trajectory, pick trajs=DefectorCrossingOvertakingAlfaBravo, gradient] table[x expr=\thisrow{x} / 1000, y expr=\thisrow{y} / 1000] {\trajBravo};
        \addplot+[initial state, forget plot] coordinates {(1.0000, -1.0000)};
        \addplot+[goal] coordinates {(1.0000, 1.2000)};
        \addplot+[forget plot, trajectory, pick trajs=DefectorCrossingOvertakingAlfaCharlie, gradient] table[x expr=\thisrow{x} / 1000, y expr=\thisrow{y} / 1000] {\trajCharlie};
        \addplot+[initial state, forget plot] coordinates {(-0.5000, 0.0)};
        \addplot+[goal] coordinates {(1.9000, 0.0)};

        \nextgroupplot[title={(c) Bravo defects}, title style={at={(0.5,-1cm)}}, point meta min=0, point meta max=94]
        \addplot+[forget plot, trajectory, pick trajs=DefectorCrossingOvertakingBravoAlfa, gradient] table[x expr=\thisrow{x} / 1000, y expr=\thisrow{y} / 1000] {\trajAlfa};
        \addplot+[initial state, forget plot] coordinates {(0.0, 0.0)};
        \addplot+[goal] coordinates {(1.5000, 0.0)};
        \addplot+[forget plot, trajectory, pick trajs=DefectorCrossingOvertakingBravoBravo, defector, gradient] table[x expr=\thisrow{x} / 1000, y expr=\thisrow{y} / 1000] {\trajBravo};
        \addplot+[initial state, forget plot, defector] coordinates {(1.0000, -1.0000)};
        \addplot+[goal, defector] coordinates {(1.0000, 1.2000)};
        \addplot+[forget plot, trajectory, pick trajs=DefectorCrossingOvertakingBravoCharlie, gradient] table[x expr=\thisrow{x} / 1000, y expr=\thisrow{y} / 1000] {\trajCharlie};
        \addplot+[initial state, forget plot] coordinates {(-0.5000, 0.0)};
        \addplot+[goal] coordinates {(1.9000, 0.0)};

        \nextgroupplot[title={(d) Charlie defects}, title style={at={(0.5,-1cm)}}, point meta min=0, point meta max=94]
        \addplot+[forget plot, trajectory, pick trajs=DefectorCrossingOvertakingCharlieAlfa, gradient] table[x expr=\thisrow{x} / 1000, y expr=\thisrow{y} / 1000] {\trajAlfa};
        \addplot+[initial state, forget plot] coordinates {(0.0, 0.0)};
        \addplot+[goal] coordinates {(1.5000, 0.0)};
        \addplot+[forget plot, trajectory, pick trajs=DefectorCrossingOvertakingCharlieBravo, gradient] table[x expr=\thisrow{x} / 1000, y expr=\thisrow{y} / 1000] {\trajBravo};
        \addplot+[initial state, forget plot] coordinates {(1.0000, -1.0000)};
        \addplot+[goal] coordinates {(1.0000, 1.2000)};
        \addplot+[forget plot, trajectory, pick trajs=DefectorCrossingOvertakingCharlieCharlie, defector, gradient] table[x expr=\thisrow{x} / 1000, y expr=\thisrow{y} / 1000] {\trajCharlie};
        \addplot+[initial state, forget plot, defector] coordinates {(-0.5000, 0.0)};
        \addplot+[goal, defector] coordinates {(1.9000, 0.0)};
    \end{groupplot}
\end{tikzpicture}%
  \vspace{-0.2cm} %
  \caption{%
    Defector trajectories for (a) the triangle scenario and (b)--(d) the crossing--overtaking scenario.
    The highlighted agent ignores legal requirements and always moves straight to its goal.
    The shading of the trajectories encodes time like in \cref{fig:action_spaces}.
  }
  \label{fig:defector}
\end{figure*}

We now consider four scenarios that include a defector:
triangle with Alfa as the defector and crossing--overtaking with each of Alfa, Bravo, and Charlie as the defector in turn.
The defector always moves straight to its goal at full speed while the other agents respond legally;
the resulting trajectories are shown in \cref{fig:defector}.
Given the committed action profile $\actionProfile$, we can identify the agent $\egoIndex \in \players$ as the defector if its action is not a legal best-response, \ie $\actionProfile[\egoIndex] \notin \bestResponse{\inducedPreference[\legalMetrics{\egoIndex}]{\legalGuideline}}(\actionProfile[-\egoIndex])$.
Observe that this test only uses the publicly available information from the legal game and does not require the auditor to know the true preference of any agent.
The test successfully identifies the defector in three of the four cases.
The exception is Bravo, which is the stand-on vessel toward both other vessels and satisfies its stand-on duties:
moving straight to its goal at full speed is its action in the committed equilibrium $\actionProfile^*$, and the other agents respond with their actions in $\actionProfile^*$.
If the action of the defector is itself a legal best response, as in this case, the resulting behavior is legal, and the defector has merely steered the agents toward an equilibrium it favors.

\subsection{Customizability}
\label{subsec:eval_customizability}

We demonstrate customizability on the triangle scenario, holding the legal guideline fixed to $\colregsGuideline$ throughout this experiment.
Private preferences are refined on the homogeneous triangle with three power-driven vessels, and the legal guideline is refined on its mixed-class variant, due to the vessel-class hierarchy of \cite[Rule~18]{internationalmaritimeorganizationimoCOLREGsConventionInternational1972}.

\subsubsection{Private Preference Refinements}

Out of the $102$ \nasheqs of the legal game, only $17$ are \nasheqs of the true game, in which \cref{prop:preference_augmentation} places the performance metrics \eqref{eq:performance_metrics} below the legal preference of every vessel.
Every true equilibrium is legal, as shown in \cref{fig:customizability}a and predicted by \cref{thm:true_game_is_legal}.
We now further refine the private preferences of the vessels.
Using \cref{prop:preference_priority_refinement} to let Bravo prioritize its give-way metrics over its stand-on metrics reduces the number of true equilibria to $11$.
After Charlie applies the opposite priority refinement, prioritizing stand-on over give-way, only $6$ equilibria are left.
Every refinement only shrinks the set of equilibria rather than adding new ones, confirming \cref{prop:equilibrium_refinement}: the four sets are nested, as \cref{fig:customizability} shows.

\subsubsection{Guideline Refinements}

To reflect a change in the law, we refine the legal guideline by adding the vessel-class hierarchy of \cite[Rule~18]{internationalmaritimeorganizationimoCOLREGsConventionInternational1972} as described in \cref{subsec:colregs_regulatory_change}.
On the mixed-class triangle, this shrinks the legal set from $92$ to $50$ equilibria.
The number of legal equilibria differs from that of the homogeneous one, because the encounter metrics toward a fishing vessel follow the class-specific rules. %
The $42$ removed equilibria are those in which a power-driven vessel violates its give-way duty toward the fishing vessel.
The legal equilibria under the refined guideline are a subset of those under the original one, as expected (see \cref{fig:induction_homomorphism}).

\begin{figure}[t]
  \centering
  \begin{tikzpicture}	
	\begin{axis}[
        hide axis,
        xmin=0, xmax=1,
        ymin=0, ymax=1,
        legend columns=-1,
        legend style={draw=none, fill=none, /tikz/every even column/.append style={column sep=1em}, /tikz/nodes={font=\scriptsize}},
    ]
        \addlegendimage{original}
        \addlegendentry{Original}

        \addlegendimage{refined}
        \addlegendentry{Refined}
    \end{axis}
\end{tikzpicture} \\[0.2cm]
  \input{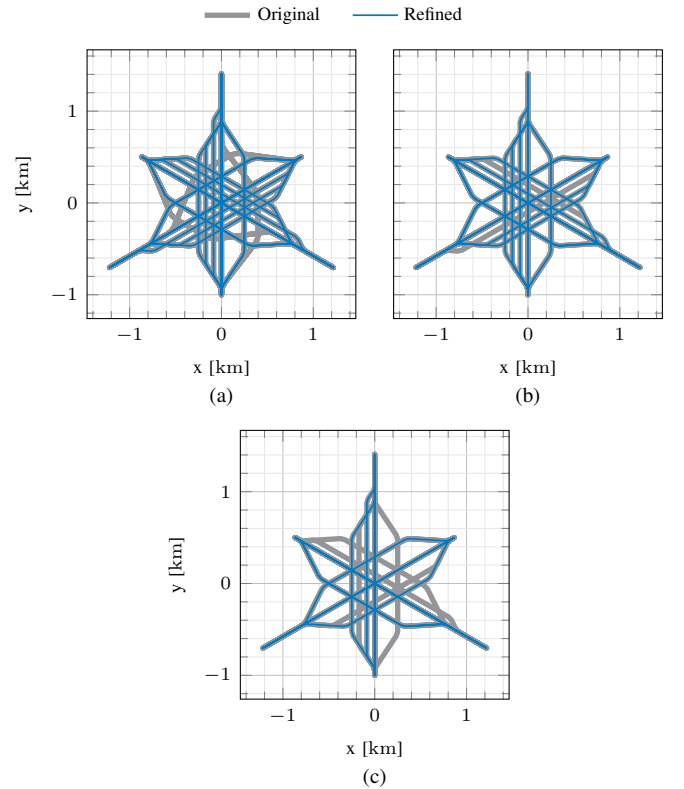}%
{%
\pgfplotstableread[col sep=comma]{data/trajectories/triangle/alfa.csv}{\trajAlfaTriangle}%
\pgfplotstableread[col sep=comma]{data/trajectories/triangle/bravo.csv}{\trajBravoTriangle}%
\pgfplotstableread[col sep=comma]{data/trajectories/triangle/charlie.csv}{\trajCharlieTriangle}%
\trajset{CustomTriangleAAlfa}{0, 1, 2, 3, 9, 13, 15, 19, 20, 31, 32, 33, 34, 39, 40}%
\trajset{CustomTriangleABravo}{0, 1, 2, 3, 9, 13, 15, 19, 20, 31, 32, 33, 34, 39, 40}%
\trajset{CustomTriangleACharlie}{0, 1, 2, 3, 9, 13, 15, 19, 20, 31, 32, 33, 34, 39, 40}%
\trajset{CustomTriangleBAlfa}{0, 1, 2, 3, 9, 13, 31, 39, 40}%
\trajset{CustomTriangleBBravo}{0, 1, 2, 3, 9, 13, 31, 39, 40}%
\trajset{CustomTriangleBCharlie}{0, 1, 2, 3, 9, 13, 31, 39, 40}%
\trajset{CustomTriangleCAlfa}{0, 1, 2, 3, 9, 13, 39, 40}%
\trajset{CustomTriangleCBravo}{2, 3, 9, 31, 39, 40}%
\trajset{CustomTriangleCCharlie}{0, 1, 3, 9, 13, 31, 39, 40}%
\trajset{CustomTriangleDAlfa}{0, 1, 3, 13, 39}%
\trajset{CustomTriangleDBravo}{3, 9, 39, 40}%
\trajset{CustomTriangleDCharlie}{0, 9, 13, 31, 39, 40}%

\pgfplotsset{customization/.style={
    width=0.58\linewidth,
    height=0.58\linewidth,
    axis equal,
    title style={anchor=north, font=\footnotesize},
    grid style={line width=.1pt, draw=gray!20},
    major grid style={line width=.2pt,draw=gray!50},
    grid=both,
    label style={font=\scriptsize},
    tick label style={font=\scriptsize},
    minor tick num=4,
    xlabel={x [\si{\kilo \meter}]},
    xtick={-2, -1, ..., 2},
    ylabel={y [\si{\kilo \meter}]},
    ytick={-2, -1, ..., 2},
    cycle list name=trajectory,
}}

\begin{tikzpicture}
    \begin{groupplot}[
        group style={
            group size=2 by 1,
            x descriptions at=edge bottom,
            y descriptions at=edge left,
            horizontal sep=0.5cm,
            vertical sep=0.75cm,
        },
        customization,
    ]
        \nextgroupplot[title={(a)}, title style={at={(0.5,-1cm)}}]
        \addplot+[forget plot, trajectory, pick trajs=CustomTriangleAAlfa, original] table[x expr=\thisrow{x} / 1000, y expr=\thisrow{y} /1000] {\trajAlfaTriangle};
        \addplot+[forget plot, trajectory, pick trajs=CustomTriangleABravo, original] table[x expr=\thisrow{x} / 1000, y expr=\thisrow{y} /1000] {\trajBravoTriangle};
        \addplot+[forget plot, trajectory, pick trajs=CustomTriangleACharlie, original] table[x expr=\thisrow{x} / 1000, y expr=\thisrow{y} /1000] {\trajCharlieTriangle};
        \addplot+[forget plot, trajectory, pick trajs=CustomTriangleBAlfa, refined] table[x expr=\thisrow{x} / 1000, y expr=\thisrow{y} /1000] {\trajAlfaTriangle};
        \addplot+[forget plot, trajectory, pick trajs=CustomTriangleBBravo, refined] table[x expr=\thisrow{x} / 1000, y expr=\thisrow{y} /1000] {\trajBravoTriangle};
        \addplot+[forget plot, trajectory, pick trajs=CustomTriangleBCharlie, refined] table[x expr=\thisrow{x} / 1000, y expr=\thisrow{y} /1000] {\trajCharlieTriangle};

        \nextgroupplot[title={(b)}, title style={at={(0.5,-1cm)}}]
        \addplot+[forget plot, trajectory, pick trajs=CustomTriangleBAlfa, original] table[x expr=\thisrow{x} / 1000, y expr=\thisrow{y} /1000] {\trajAlfaTriangle};
        \addplot+[forget plot, trajectory, pick trajs=CustomTriangleBBravo, original] table[x expr=\thisrow{x} / 1000, y expr=\thisrow{y} /1000] {\trajBravoTriangle};
        \addplot+[forget plot, trajectory, pick trajs=CustomTriangleBCharlie, original] table[x expr=\thisrow{x} / 1000, y expr=\thisrow{y} /1000] {\trajCharlieTriangle};
        \addplot+[forget plot, trajectory, pick trajs=CustomTriangleCAlfa, refined] table[x expr=\thisrow{x} / 1000, y expr=\thisrow{y} /1000] {\trajAlfaTriangle};
        \addplot+[forget plot, trajectory, pick trajs=CustomTriangleCBravo, refined] table[x expr=\thisrow{x} / 1000, y expr=\thisrow{y} /1000] {\trajBravoTriangle};
        \addplot+[forget plot, trajectory, pick trajs=CustomTriangleCCharlie, refined] table[x expr=\thisrow{x} / 1000, y expr=\thisrow{y} /1000] {\trajCharlieTriangle};
    \end{groupplot}
\end{tikzpicture}
\\[0.2cm]
\begin{tikzpicture}
    \begin{axis}[
        customization,
        title={(c)},
        title style={at={(0.5,-1cm)}},
    ]
        
        \addplot+[forget plot, trajectory, pick trajs=CustomTriangleCAlfa, original] table[x expr=\thisrow{x} / 1000, y expr=\thisrow{y} /1000] {\trajAlfaTriangle};
        \addplot+[forget plot, trajectory, pick trajs=CustomTriangleCBravo, original] table[x expr=\thisrow{x} / 1000, y expr=\thisrow{y} /1000] {\trajBravoTriangle};
        \addplot+[forget plot, trajectory, pick trajs=CustomTriangleCCharlie, original] table[x expr=\thisrow{x} / 1000, y expr=\thisrow{y} /1000] {\trajCharlieTriangle};
        \addplot+[forget plot, trajectory, pick trajs=CustomTriangleDAlfa, refined] table[x expr=\thisrow{x} / 1000, y expr=\thisrow{y} /1000] {\trajAlfaTriangle};
        \addplot+[forget plot, trajectory, pick trajs=CustomTriangleDBravo, refined] table[x expr=\thisrow{x} / 1000, y expr=\thisrow{y} /1000] {\trajBravoTriangle};
        \addplot+[forget plot, trajectory, pick trajs=CustomTriangleDCharlie, refined] table[x expr=\thisrow{x} / 1000, y expr=\thisrow{y} /1000] {\trajCharlieTriangle};
    \end{axis}
\end{tikzpicture}%
}%
%
  \vspace{-0.2cm} %
  \caption{%
    Customization of the private preferences via refinement of the legal preferences by adding (a) performance, (b) prioritization of give-way metrics over stand-on metrics for Bravo, and (c) stand-on metrics over give-way metrics for Charlie.
    Each plot shows one refinement step by drawing the trajectories occurring in at least one \nasheq of the original game below those of the refined game.
  }
  \label{fig:customizability}
\end{figure}

\subsection{Solution Protocols}
\label{subsec:eval_protocols}

We compare the two protocols of \cref{sec:solution_protocols}, belief play and distributed \ibr, on the triangle, crossing--overtaking, four-agent priority, and gauntlet scenarios.
\Cref{fig:protocol_communication} reports the communication cost of the protocols and \cref{tab:ibr_gap} the gap to the socially optimal equilibrium $\actionProfile^*$ \eqref{eq:socially_optimal_true_ne}.
The communication cost counts how many messages were sent via the communication channel, independent of the number of recipients.
Thus, a belief play broadcast and an \ibr point-to-point forward each count as one send.
The reported count is proportional to the total volume of data exchanged, because each message consists of one action profile and thus has a fixed size.
As the local solver in belief play, we employ exhaustive enumeration.
A run of \ibr depends on a random seed, which determines the initial action profile and breaks ties among multiple best responses.
\Cref{fig:protocol_communication} characterizes the communication cost of a single run: it reports the mean, standard deviation, and range of the number of sends over 15 independently seeded runs.
\Cref{tab:ibr_gap} characterizes the solution quality that is attainable by restarting the protocol: for $k \in \{1, 5, 10, 20\}$, \ibr is run from $k$ differently seeded initializations, the distinct equilibria reached are pooled, and the socially best among them according to \eqref{eq:social_metric} is committed.
Each entry of the table is therefore the outcome of one pool of $k$ runs and not an average, and its communication cost is the sum over these $k$ runs.
The pools are nested, so the gap cannot increase with $k$.

Distributed \ibr converges to a true equilibrium in all four considered scenarios while communicating far fewer action profiles than belief play.
As the local solver of belief play is deterministic, so is the communication cost of the protocol.
It sends every distinct local-equilibrium proposal: $159$ profiles on triangle, $22$ on crossing--overtaking, $124$ on four-agent priority, and $271$ on gauntlet.
Distributed \ibr first uses $\abs{\players}$ sends to construct the initial profile and then forwards one profile per agent in each best-response round;
its means are $11.4$, $11.4$, $17.6$, and $16.8$ profiles per run, respectively, which corresponds to about three best-response rounds.
The trade-off is the quality of the committed equilibrium:
\ibr commits to \emph{a} true \nasheq but not necessarily to the socially optimal one.
In contrast, belief play always commits $\actionProfile^*$, as it recovers all true \nasheqs by construction.
\Cref{tab:ibr_gap} quantifies this by reporting the optimality gap at the encounter level, \ie $\socialMetric^{\mathrm{enc}}(\actionProfile^*) - \socialMetric^{\mathrm{enc}}(\actionProfile)$, where $\actionProfile$ is the committed action profile.
Five initializations suffice for \ibr to reach $\actionProfile^*$ on the triangle, and ten on crossing--overtaking and four-agent priority.
On the gauntlet, its encounter compliance remains $0.08$ below that of $\actionProfile^*$ with $10$ and with $20$ initializations.

\begin{figure}[t]
  \centering
  \begin{tikzpicture}
    \pgfplotstableread[col sep=comma]{data/avg_communications.csv}{\avgCommunications}%

    \begin{axis}[
        ybar,
        bar width=0.6cm,
        width=\linewidth,
        height=4.5cm,
        grid style={line width=.1pt, draw=gray!20},
        major grid style={line width=.2pt,draw=gray!50},
        ymajorgrids={true},
        yminorgrids={true},
        label style={font=\scriptsize},
        tick label style={font=\scriptsize},
        enlarge x limits=0.2,
        symbolic x coords={triangle, crossing_overtaking, four_agent_prio, gauntlet},
        xtick={triangle, crossing_overtaking, four_agent_prio, gauntlet},
        xtick style={draw=none},
        xticklabels={Triangle, Crossing--overtaking, Four-agent priority, Gauntlet},
        x tick label style={
            text width=1.25cm,
            execute at begin node=\setlength{\baselineskip}{2pt}
        },
        ylabel={Action profiles sent [\#]},
        ymin={0},
        ymax={300},
        ytick={0, 50, ..., 300},
        minor y tick num=4,
        cycle list name=tumcolornomark,
        legend pos=north west,
        legend style={font=\scriptsize, inner ysep=0.7pt, cells={anchor=west}},
        legend image code/.code={
            \draw[#1, fill] (0cm,-0.1cm) rectangle (0.25cm,0.1cm);
        },
    ]
        \addplot+[thick] table[x=scenario, y=belief_play] {\avgCommunications};
        \addlegendentry{Belief play}

        \addplot+[thick, error bars/.cd, y dir=both, y explicit] table[x=scenario, y=dibr, y error=dibr_stddev] {\avgCommunications};
        \addlegendentry{Distributed \ibr}
    \end{axis}
\end{tikzpicture}%
  \caption{%
    Communication cost of a single run of each solution protocol.
    The error bars indicate the standard deviation over 15 independently seeded distributed \ibr runs.
  }
  \label{fig:protocol_communication}
\end{figure}
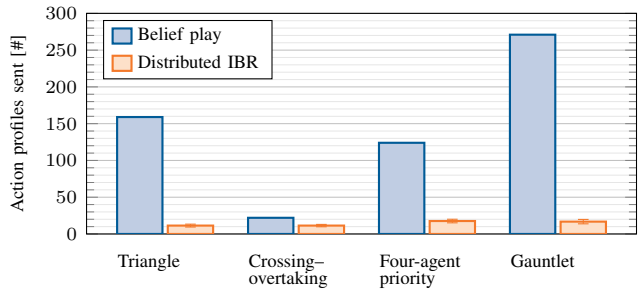

\begin{table}[t]
  \caption{%
    Solution quality of \ibr with restarts: optimality gap of the socially best equilibrium among those reached from $k$ initializations / number of distinct equilibria reached.
    Each entry stems from one pool of $k$ seeded runs.
  }
  \label{tab:ibr_gap}
  \centering
  \resizebox{\columnwidth}{!}{%
    \begin{tabular}{@{}lrrrrr@{}}
      \toprule
      \textbf{Scenario}             & \multicolumn{4}{c}{\textbf{Distributed \ibr}} & \textbf{Belief play}                                                                                             \\
      \cmidrule(lr){2-5}
                           & 1 init.                  & 5 inits.       & 10 inits.      & 20 inits.      &                \\
      \midrule
      triangle             & $0.84$\,/\,$1$           & $0.0$\,/\,$3$  & $0.0$\,/\,$5$  & $0.0$\,/\,$10$ & $0.0$\,/\,$17$ \\
      crossing--overtaking & $8.82$\,/\,$1$           & $8.82$\,/\,$2$ & $0.0$\,/\,$4$  & $0.0$\,/\,$4$  & $0.0$\,/\,$5$  \\
      four-agent priority  & $9.92$\,/\,$1$           & $4.05$\,/\,$4$ & $0.0$\,/\,$9$  & $0.0$\,/\,$13$ & $0.0$\,/\,$29$ \\
      gauntlet             & $15.59$\,/\,$1$          & $0.24$\,/\,$3$ & $0.08$\,/\,$7$ & $0.08$\,/\,$8$ & $0.0$\,/\,$43$ \\
      \bottomrule
    \end{tabular}
  }
\end{table}

\tikzexternaldisable
\section{Conclusion and Future Work}
\label{sec:conclusion}

We presented a framework for multi-agent motion planning under legal requirements with partially ordered objectives, building on posetal games~\cite{zanardiPosetalGamesEfficiency2022}.
Guidelines capture priorities among legal concerns independently of the concrete scenario, from which we systematically derive partially ordered agent preferences and the corresponding legal game.
To customize the framework while guaranteeing legality, agents augment these preferences with private objectives and regulators refine the guideline.
Since every such customization is a refinement and the set of equilibria depends antitonically on the preferences, customizations can only remove equilibria, so that legality is preserved by construction.
Even under private preferences, we can deploy our framework using a distributed \ibr protocol or belief play, which coordinates agents through conservative local beliefs and provably recovers the equilibria of the true game.
We instantiated the framework for \colregs-compliant motion planning, where it naturally lifts the pairwise rules to multi-vessel encounters.
Using the case study, we demonstrated customizability and auditability, including the ability to explain required rule violations and to identify defectors.
Our evaluation showed that distributed \ibr exchanges solution quality for communication efficiency and weaker information assumptions compared to belief play.

Several directions remain open.
First, we would like to extend our framework into a receding-horizon scheme, in which agents repeatedly re-solve the game as the encounter evolves and refine their beliefs about the preferences of the others based on observed behavior.
The main difficulty is to update the beliefs from observations while keeping them conservative, since the technical result only holds for conservative beliefs.
Second, deploying such a scheme on real vessels requires real-time capable solvers.
While exhaustive enumeration of the joint action space provides exact references, its cost grows exponentially with the number of agents.
Specialized planners for \colregs encounters that exploit the structure of the game, such as the pairwise symmetric collision metrics, could substantially reduce this cost.
Establishing conditions under which the COLREGS game is a potential game would, in addition, guarantee the convergence of distributed \ibr, which we have so far observed empirically.
Third, the auditability results are relative to the finite action lattice, so that a violation classified as unsatisfiable may become satisfiable under a richer action space.
Characterizing how the audit outcome depens on the discretization would strengthen the legal interpretation of such certificates.
Fourth, a defector whose action is still a legal best response can steer the equilibrium selection without violating the law, which calls for a mechanism-design perspective on the social metric.
Finally, applying our framework to other domains in which regulators specify only partial priorities, like autonomous driving or airborne collision avoidance, would further test its generality.
In both domains, interaction rules are also stated pairwise, so the lifting of pairwise rules to many agents would apply directly.

\appendix

\subsection{Proofs of the Preference Refinement Operations}
\label{app:preference_proofs}

To prove \cref{prop:preference_priority_refinement,prop:preference_aggregation,prop:preference_augmentation}, we first observe that two actions are indifferent if and only if they are equal in every metric:
\begin{lemma}[Indifference]
    \label{lem:action_indifference}
    Given a preference $\preference := \langle \metrics, \posetPreceq_{\preference} \rangle$, we have $\actionProfile \actionPreceq[\preference] \actionProfile'$ and $\actionProfile \actionSucceq[\preference] \actionProfile'$ if and only if
    \begin{equation*}
        \forall m \in \metrics : m(\actionProfile) = m(\actionProfile').
    \end{equation*}
\end{lemma}
\begin{proof}
    The backward direction follows immediately since the universal quantifier is vacuously satisfied in both instances of \eqref{eq:action_order}.
    For the forward direction, we assume both
    \begin{subequations}
        \label{eq:action_indifference}
        \begin{align}
            \label{eq:action_indifference:a}
            &\begin{aligned}
                &\forall m \in \metrics : m(\actionProfile) > m(\actionProfile') \implies \\
                &\quad \exists m' \in \metrics : m' \posetSucc_{\preference} m \text{ and } m'(\actionProfile) < m'(\actionProfile')
            \end{aligned} \\
            \label{eq:action_indifference:b}
            &\begin{aligned}
                &\forall m \in \metrics : m(\actionProfile') > m(\actionProfile) \implies \\ 
                &\quad \exists m' \in \metrics : m' \posetSucc_{\preference} m \text{ and } m'(\actionProfile') < m'(\actionProfile).
            \end{aligned}
        \end{align}
    \end{subequations}
    We assume that there exists a metric $m$ so that $m(\actionProfile) \neq m(\actionProfile')$ and derive a contradiction.
    Since $\metrics$ is finite, we may choose $m$ that is maximal \wrt $\posetPreceq_{\preference}$ among all metrics that differ at $\actionProfile$ and $\actionProfile'$, \ie $m'(\actionProfile) = m'(\actionProfile')$ for all $m' \posetSucc_{\preference} m$.
    But, if $m(\actionProfile) > m(\actionProfile')$, there must be $m' \posetSucc_{\preference} m$ so that $m'(\actionProfile) < m'(\actionProfile')$ according to \eqref{eq:action_indifference:a}, which is a contradiction.
    If $m(\actionProfile) < m(\actionProfile')$ we use \eqref{eq:action_indifference:b} to derive the contradiction.
\end{proof}
As a direct consequence of \cref{lem:action_indifference}, we have $\actionProfile \actionPrec[\preference] \actionProfile'$ if and only if, in addition to \eqref{eq:action_order}, there exists a separating metric $m \in \metrics$ with $m(\actionProfile) < m(\actionProfile')$.
With this, we can prove \cref{prop:preference_priority_refinement,prop:preference_aggregation,prop:preference_augmentation} by showing that every strict relation of the original preference is preserved.
\begin{proof}[Proof of \cref{prop:preference_priority_refinement}]
    Let $\actionProfile \actionPrec[\preference] \actionProfile'$.
    By \cref{lem:action_indifference}, \eqref{eq:action_order} holds for $\preference$ and some $m \in \metrics$ has $m(\actionProfile) < m(\actionProfile')$.
    Since ${\posetPreceq_{\preference}} \subseteq {\posetPreceq_{\preference'}}$, every dominating metric required by \eqref{eq:action_order} for $\preference$ still dominates under $\preference'$, so \eqref{eq:action_order} holds for $\preference'$ as well, while the separating metric $m$ persists.
    Hence, $\actionProfile \actionPrec[\preference'] \actionProfile'$.
\end{proof}
\begin{proof}[Proof of \cref{prop:preference_aggregation}]
    Let $\actionProfile \actionPrec[\preference] \actionProfile'$.
    We verify \eqref{eq:action_order} for $\preference'$.
    Consider any metric $m$ of $\metrics'$ with $m(\actionProfile) > m(\actionProfile')$.
    If $m \neq \hat{m}$, the dominating metric supplied by \eqref{eq:action_order} for $\preference$ is preserved by~\eqref{subeq:aggregation_order:none} unless it is $m_k$ for some $k \in \{ 1, 2 \}$.
    In the latter case, \eqref{subeq:aggregation_order:succ} provides $\hat{m} \posetSucceq_{\preference'} m$.
    If $\hat{m}$ improves, it is the required dominating metric.
    Otherwise, since $m_k$ improves and $f$ is strictly increasing in each argument, the other metric $m_l$, with $l \in \{ 1, 2 \} \setminus \{ l \}$, worsens.
    Since $m_l$ worsens and $\actionProfile \actionPrec[\preference] \actionProfile'$, \eqref{eq:action_order} supplies an improving metric above $m_l$ that survives the aggregation.
    By~\eqref{subeq:aggregation_order:prec}, this metric sits above $\hat{m}$, and hence above $m$ by transitivity.
    If $m = \hat{m}$, strict monotonicity implies $m_k(\actionProfile) > m_k(\actionProfile')$ for some $k \in \{ 1, 2\}$, and the dominating metric for $m_k$ under $\preference$, which by~\eqref{subeq:aggregation_order:prec} sits above $\hat{m}$ in $\preference'$, improves on $\actionProfile'$.
    A separating metric of $\preference$ likewise yields one for $\preference'$, so $\actionProfile \actionPrec[\preference'] \actionProfile'$ by \cref{lem:action_indifference}.
\end{proof}
\begin{proof}[Proof of \cref{prop:preference_augmentation}]
    Let $\actionProfile \actionPrec[\preference_1] \actionProfile'$.
    Every metric of $\metrics_2$ is ranked, in $\preference$, strictly below every metric of $\metrics_1$, so it can neither serve as nor invalidate a dominating metric for a comparison decided within $\metrics_1$.
    Since $\actionProfile \actionPrec[\preference_1] \actionProfile'$ is strict, there exists at least one separating metric $m \in \metrics_1$ with $m(\actionProfile) < m(\actionProfile')$.
    $m$ persists and dominates every metric in $\metrics_2$, yielding $\actionProfile \actionPrec[\preference] \actionProfile'$ by \cref{lem:action_indifference}.
\end{proof}

\bibliographystyle{IEEEtran}
\bibliography{bib/clean}

\end{document}